\documentclass[11pt]{article}

\usepackage{amscd}
\usepackage{amssymb}
\usepackage{amsthm}
\usepackage{amsmath}
\usepackage{amsfonts}
\usepackage{natbib}
\usepackage{url}
\usepackage{graphicx,times}
\usepackage{mathtools}
\usepackage{xcolor}
\usepackage{dsfont} 

\usepackage{geometry}
\newcommand{\R}{\mathbb{R}}
\newcommand{\Prob}{\mathbb{P}}

\newcommand{\Copt}{C^{(\alpha)}}
\newcommand{\Copthat}{\widehat{C}^{(\alpha)}}
\newcommand{\Copttilde}{\widetilde{C}^{(\alpha)}}
\newcommand{\Coptloc}{\widehat{C}^{(\alpha)}_{\text{loc}}}
\newcommand{\Copttrans}{\widehat{C}^{(\alpha)}_{\text{trans}}}

\newcommand{\Coptglm}{\widehat{C}^{(\alpha)}_{n, k}}

\newcommand{\ptrue}{p_{\psi}}

\newcommand{\psihat}{\hat{\psi}}

\newcommand{\ptilnon}{\tilde{p}}

\newcommand{\E}{\mathbb{E}}
\newcommand{\A}{\mathcal{A}}
\newcommand{\X}{\mathcal{X}}

\newcommand{\talpha}{t^{(\alpha)}}
\newcommand{\salpha}{s^{(\alpha)}}

\newcommand{\Aaxe}{A^{(\alpha)}_{x,\varepsilon}}
\newcommand{\thatalpha}{\hat{t}^{(\alpha)}}
\newcommand{\ttildealpha}{\tilde{t}^{(\alpha)}}

\newcommand{\inner}[1]{\langle #1 \rangle}
\newcommand{\indicator}[1]{\mathds{1}\left\{ #1 \right\}}

\DeclarePairedDelimiter\ceil{\lceil}{\rceil}
\DeclarePairedDelimiter\floor{\lfloor}{\rfloor}

\DeclareMathOperator{\dom}{dom}

\newtheorem{lem}{Lemma} 
\newtheorem{thm}{Theorem}
\newtheorem{cor}{Corollary} 
\newtheorem{defn}{Definition} 
\newtheorem{prop}{Proposition} 

\title{Efficient and minimal length parametric conformal prediction regions}

\author{Daniel J. Eck$^1$ and Forrest W. Crawford$^{2,3,4,5}$ \\[1em]
\normalsize 1. Department of Statistics, University of Illinois Urbana-Champaign \\
\normalsize 2. Department of Biostatistics, Yale School of Public Health \\
\normalsize 3. Department of Statistics \& Data Science, Yale University \\
\normalsize 4. Department of Ecology \& Evolutionary Biology, Yale University \\
\normalsize 5. Yale School of Management }

\begin{document}

\maketitle

\begin{abstract}
\noindent 
  Conformal prediction methods construct prediction regions for iid 
  data that are valid in finite samples.
  We provide two parametric conformal prediction regions that are applicable 
  for a wide class of continuous statistical models.  
  This class of statistical models includes generalized linear models (GLMs) 
  with continuous outcomes.
  Our parametric conformal prediction regions possesses finite sample validity, 
  even when the model is misspecified, and are asymptotically of minimal length 
  when the model is correctly specified.
  The first parametric conformal prediction region is constructed through binning 
  of the predictor space, guarantees finite-sample local validity and is 
  asymptotically minimal at the $\sqrt{\log(n)/n}$ rate when the dimension $d$ 
  of the predictor space is one or two, and converges at the 
  $O\{(\log(n)/n)^{1/d}\}$ rate when $d > 2$.  
  The second parametric conformal prediction region is constructed by transforming 
  the outcome variable to a common distribution via the probability integral 
  transform, guarantees finite-sample marginal validity, 
  and is asymptotically minimal at the $\sqrt{\log(n)/n}$ rate.  
  We develop a novel concentration inequality for maximum likelihood 
  estimation that induces these convergence rates.  
  We analyze prediction region coverage properties, large-sample 
  efficiency, and robustness properties of four methods for constructing 
  conformal prediction intervals for GLMs: 
  fully nonparametric kernel-based conformal, 
  residual based conformal, 
  normalized residual based conformal, 
  and parametric conformal which uses the assumed GLM density as a conformity 
  measure.  Extensive simulations compare these approaches to standard 
  asymptotic prediction regions.  The utility of the parametric conformal 
  prediction region is demonstrated in an application to interval prediction 
  of glycosylated hemoglobin levels, a blood measurement used to 
  diagnose diabetes. \\[1em]
\end{abstract}

\section{Introduction}

\citet{vovk2005springer} introduced conformal prediction to construct 
finite sample valid prediction regions for predictions from nearest-neighbor 
methods, support-vector machines, and sequential classification and both 
linear and ridge regression problems. 
The goal in conformal prediction is to construct a prediction 
region $C_n$ from an iid random sample $Y_1$, $\ldots$, $Y_n$, such 
that the probability of a future observation $Y_{n+1}$ belonging to the region 
$C_n$ exceeds a desired coverage level \citep{shafer2008tutorial}.
That is, given $Y_1,\ldots,Y_n$, we seek a set $C_n=C_n(Y_1,\ldots,Y_n)$ such 
that for $0 < \alpha < 1$,
\begin{equation} \label{setup}
  \Prob(Y_{n+1} \in C_n) \geq 1-\alpha.
\end{equation}
\citet{lei2013distribution} extended the framework of 
\citet{shafer2008tutorial} to provide a framework for which the 
prediction region $C_n$ is also asymptotically of minimal length. 
These prediction regions make use of a conformity measure $\sigma$ which   
measures the agreement of a point $y$ with a probability measure $P$. 
For example, when $y \in \R$, the probability density function $p$ 
corresponding to $P$ is a conformity measure $\sigma$. 
\citet{lei2013distribution} proposed nonparametric kernel density estimation 
of $p$ to construct conformal prediction regions that are asymptotically 
minimal under smoothness conditions on $p$.  Beyond real-valued outcomes, 
conformal methods have been proposed for 
functional data \citep{lei2015conformal},  
nonparametric regression \citep{lei2014distribution}, 
time series data \citep{chernozhukov2018timeseries}, 
high-dimensional regression problems \citep{lei2018distribution}, 
random effects \citep{dunn2018distribution},
causal inference \citep{chernozhukov2018causal},  
machine learning \citep{vovk2005springer, gammerman2007hedging, 
  papa2011neural, vovk2012conditional, vovk2014ridge, 
  balasubramanian2014conformal, johansson2018interpretable, wang2018fast}, 
and quantile regression \citep{romano2019malice, romano2019conformalized, 
  sesia2019conformal}.
The finite sample validity property of conformal prediction regions has broad 
appeal in many scientific domains, including 
astrophysics \citep{ciollaro2018functional},
medical applications \citep{lambrou2009diagnosis, devetyarov2012cancer, 
  eklund2015drug, bosc2019large},
genetics \citep{norinder2018predictinga, norinder2018predictingb},
and chemistry \citep{cortes2018deep, ji2018emoltox, svensson2018maximizing, 
  svensson2018conformal, toccaceli2017conformal}.

However, usefulness of the finite sample validity property of conformal prediction regions in 
regression may be offset by the size of the prediction region. 
Recently, \citet{cortes2019concepts} stated that 
``a major issue in conformal prediction applied to regression is the low efficiency of 
most conformal prediction models\ldots Such large intervals are not informative and thus 
hamper the practical usefulness of conformal prediction'',
calling for the development of non-conformity functions to 
reduce the size of the predicted confidence regions \citep{cortes2019concepts}. 
The parametric conformal prediction regions that we develop in this paper answer the 
call made by \citet{cortes2019concepts} in the context of regression.  
We show that when one uses the parametric density with maximum likelihood estimators 
(MLEs) of model parameters plugged in for modeling parameters as the conformity 
measure, then one can can construct conformal predictions that are asymptotically 
minimal length and our sharp under our proof technique.  
Our parametric conformal prediction regions possess the same finite-sample 
validity guarantees as existing conformal prediction procedures, even under model 
misspecification. 

In a recent paper, \citet{chernozhukov2019distributional} proposed parametric 
conformal prediction regions based on regression models for conditional 
distributions, establishing the conditional validity of the resulting 
prediction intervals under consistent estimation of the conditional 
distributions.  Their proposed prediction method is based on modeling the 
entire  conditional distribution (rather than the conditional mean) and 
thereby generates conditionally valid prediction sets that fully utilize the 
information in the covariates \citep{chernozhukov2019distributional}.  
In another recent paper, \citet{izbicki2019distribution} proposed 
parametric conformal prediction regions that are asymptotically minimal 
length, without providing explicit convergence rates.

We build upon the frameworks outlined by 
\citet{chernozhukov2019distributional} and  
\citet{izbicki2019distribution}  
by providing two methods for constructing parametric 
conformal prediction regions which are asymptotically of minimal length with 
desirable rates of convergence.  
Both methods are constructed by using the maximum likelihood estimator as a 
plugin estimate of unknown modeling parameters in the underlying density and 
distribution functions.  The first of these is conditionally valid in local 
regions of the predictor space \citep[pg. 72]{lei2014distribution}, and 
is asymptotically conditionally valid for point predictions. 
The rate of convergence for this parametric conformal prediction region is 
$O\{\sqrt{\log(n)/n}\}$ when the dimension of predictor space is $d \leq 2$ 
and is $O\{(\log(n)/n)^{1/d}\}$ when $d > 2$ where the predictor region of 
interest shrinks at a suitable rate.  The construction of this conformal 
prediction region follows the binning technique in \citet{lei2014distribution}. 
The $O\{(\log(n)/n)^{1/d}\}$ rate is a consequence of how quickly the 
predictor region of interest shrinks.  
The second parametric conformal prediction region is obtained by first 
transforming the outcome variable to a ``common'' distribution via the 
probability integral transform, and conducting 
conformal prediction with respect to this common distributions to obtain a 
finite sample marginally valid prediction region, and then backtransforming this 
region to the scale of the original outcome variables.  
We leverage the estimated density to obtain estimates of the quantiles 
corresponding to the minimal length $(1-\alpha)\times 100\%$ prediction region 
before we backtransform to the scale of the response.
We show that this parametric conformal prediction region is asymptotically minimal 
at rate $O\{\sqrt{\log(n)/n}\}$. 
The convergence rates for both parametric conformal regions are appreciably 
faster than that presented by \cite{lei2014distribution} which gave rates of 
$\left\{\log(n)/n\right\}^{1/(d + 3)}$ for nonparametric conformal prediction 
regions in the same regression context.  

\citet{chernozhukov2019distributional} developed a similar transformation  
conformal prediction region that is motivated by the probability integral transform.  
Their conformal prediction technique is based on the $\alpha/2$ and $1-\alpha/2$ 
quantiles of the underlying data generating model.  This technique will return 
asymptotically minimal length prediction regions when the data generating model is 
symmetric.  However, this asymptotic property will not hold for non-symmetric 
distributions.  Moreover, \citet{chernozhukov2019distributional} did not provide 
convergence rates for their transformation based conformal prediction region. 
\citet{izbicki2019distribution} developed a similar binned and transformation 
parametric conformal prediction regions which are asymptotically minimal length.  
However, \citet{izbicki2019distribution} did not provide convergence rates for 
either conformal prediction region. Furthermore, their transformation conformal 
prediction region may obtain larger regions than necessary when the underlying 
distribution is neither symmetric nor unimodal \citep[Section 2]{izbicki2019distribution}.
We show that maximum likelihood estimation within a class of models that possess 
subexponential tail behavior and additional mild regularity conditions is 
necessary for our convergence rates to hold.  Furthermore, the asymptotic 
properties for both of our conformal prediction regions hold for nonsymmetric 
and multimodal distributions. 

We motivate our parametric conformal prediction techniques through the 
application of GLMs with continuous outcomes.  
GLMs are widely used regression models for outcomes that follow an exponential 
family distribution \citep{mccullagh1989generalized}.  GLMs are popular in empirical 
research in the biomedical and social sciences; procedures for fitting GLMs is 
incorporated into every major statistical software package.  Often predictive 
inference from a fitted GLM serves the primary scientific goals of the 
analysis.  Point predictions from these models are usually combined with 
variance estimates from the bootstrap or delta method to construct prediction 
intervals.  For example, interval predictions on the outcome scale may be 
constructed using the \texttt{predict.glm} function in R \citep{R-core} 
to construct Wald type intervals for either the mean response or its distribution. 
We show that our parametric conformal regions provide a useful alternative to such 
approaches.

In an extensive simulation study, we analyze marginal, local, and conditional 
prediction interval coverage, and large-sample efficiency of four methods 
for constructing conformal prediction intervals from fitted GLMs: 
fully nonparametric kernel-based prediction \citep{lei2014distribution}, 
prediction for residuals and locally weighted residuals 
\citep{lei2018distribution}, 
and prediction using the parametric density as conformity measure.   
We find that when sample sizes are moderate -- large enough 
to estimate regression coefficients reasonably precisely, but before 
asymptotic arguments guarantee good conditional coverage -- conformal 
prediction methods outperform traditional methods in terms of 
finite sample marginal, local, and conditional coverage. 
Importantly, we find that parametric conformal prediction regions perform 
extremely well under mild model misspecification, they are calibrated to 
give valid finite sample coverage and they are not too large for meaningful  
inference in this setting.
We demonstrate the utility of conformal prediction methods for GLMs in an 
application to diabetes diagnosis via interval prediction of glycosylated 
hemoglobin levels of subjects participating in a
community-based study \citep{willems1997prevalence}.

\section{Background}

\subsection{Conformal prediction for continuous outcomes}
\label{sec:background}

The basic intuition underlying the conformal prediction method is that, 
given an independent sample $Y_1, \ldots, Y_n$ from distribution $P$ 
defined on $\R^r$, 
the iid hypothesis that $(Y_1$, $\ldots$, $Y_n$, $Y_{n+1}) \sim P$ using 
observation $(Y_1$, $\ldots$, $Y_n$, $y)$ is tested for each $y \in \R^r$, 
and a prediction region $\Copthat$ is created by inverting this test 
\citep{lei2013distribution, vovk2005springer, shafer2008tutorial}.
More formally, suppose that $(Y_1$, $\ldots$, $Y_n$, $Y_{n+1}) \sim P$. 
Let $\widehat{P}_{n+1}$ be the corresponding empirical distribution and note 
that $\widehat{P}_{n+1}$ is symmetric in its $n + 1$ arguments.
The conformity rank is 
$
  \pi_{n+1,i} = (n+1)^{-1}\sum_{j=1}^{n+1} 
    \indicator{ \sigma(\widehat{P}_{n+1},Y_j) 
    \leq \sigma(\widehat{P}_{n+1},Y_i) }
$
where $\sigma(\cdot$, $\cdot)$ is a conformity measure which 
measures the ``agreement'' of a point $y$ with respect to a distribution $P$ 
\citep{lei2013distribution}.  Informally, a conformity measure should be large
when there is agreement between $y$ and $P$.  
One important example of a conformity measure is the density function $p$ of $P$ 
(when it exists); this idea is used in Section~\ref{sec:paraconf}. 
Exchangeability of $(\pi_{n+1,i} : 1 \leq i \leq n + 1)$ follows from 
exchangeability of the random variables 
$\{\sigma(\widehat{P}_{n+1}, Y_i) : 1 \leq i \leq n+1\}$. 
Define $\pi(y) = \pi_{n+1,n+1}|_{Y_{n+1}=y}$ as the random variable 
$\pi_{n+1,n+1}$ evaluated at $Y_{n+1}=y$. The conformal prediction region 
for $Y_{n+1}$ is 
$
  \Copthat(Y_1, \ldots, Y_n) 
    = \left\{y : \pi(y) \geq \widetilde{\alpha}\right\}, 
$
where 
$
  \widetilde{\alpha} = \floor{(n+1)\alpha}/(n+1). 
$
Then by construction, 
$
  \Prob(\pi(Y_{n+1}) \geq \widetilde{\alpha}) \geq 1 - \alpha,
$
which implies that 
$
  \Prob\left\{Y_{n+1} \in \Copthat(Y_1, \ldots, Y_n)\right\} \geq 1 - \alpha. 
$
Any conformity measure $\sigma$ can be used to construct prediction regions 
with finite sample validity. 

\begin{lem} \label{valid} 
\citep{lei2013distribution}.
Suppose $Y_1$, $\ldots$, $Y_n$, $Y_{n+1}$ is an independent random 
sample from $P$. Then $\Prob\left\{Y_{n+1} \in \Copthat\right\} \geq 1 - \alpha$,
for all probability measures $P$, and hence $\Copthat$ is valid.
\end{lem}

Figure~\ref{Fig:schematic} shows schematically how this conformal 
prediction region is constructed.

\begin{figure}[t!]
\begin{center}
\includegraphics[width=0.49\textwidth]{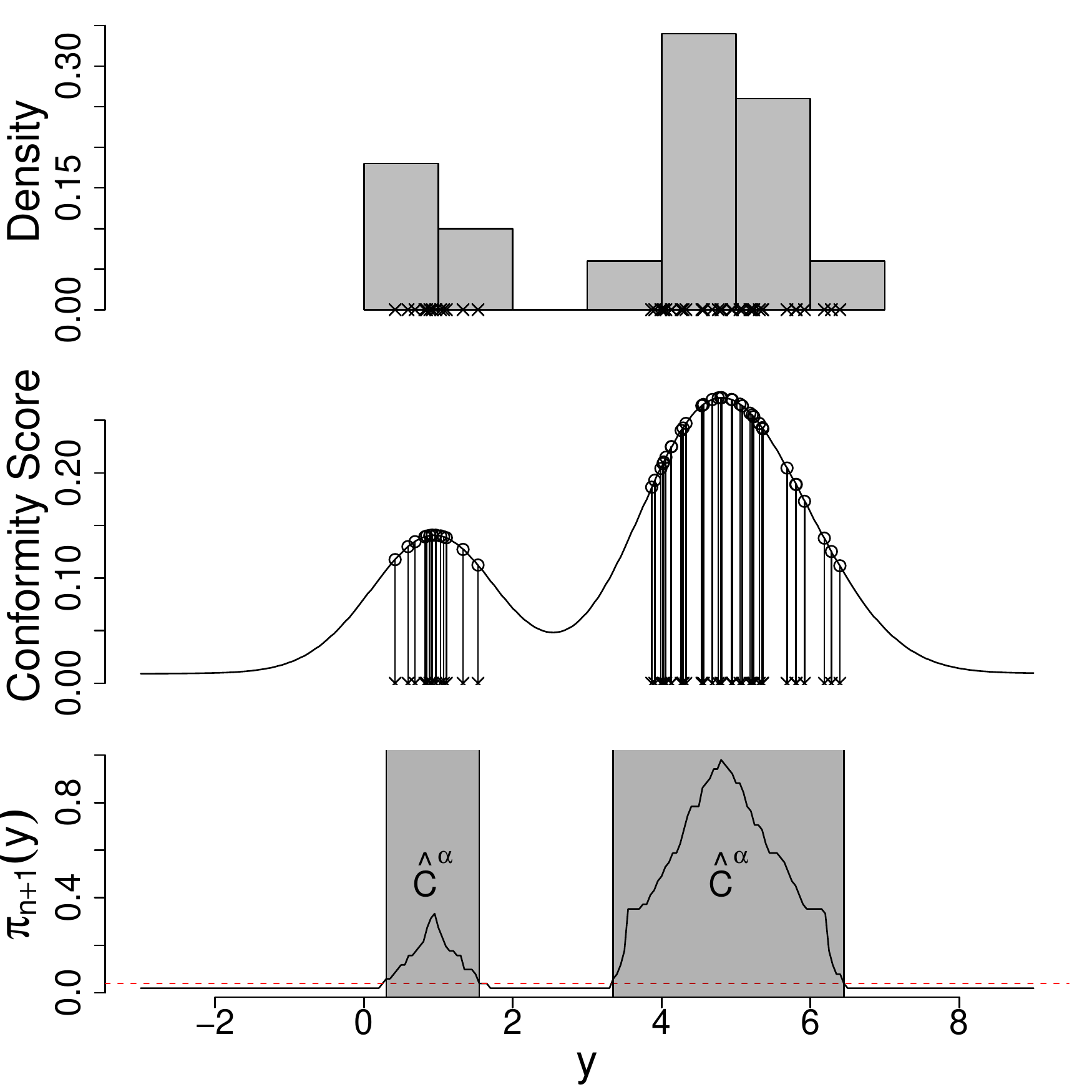}
\includegraphics[width=0.49\textwidth]{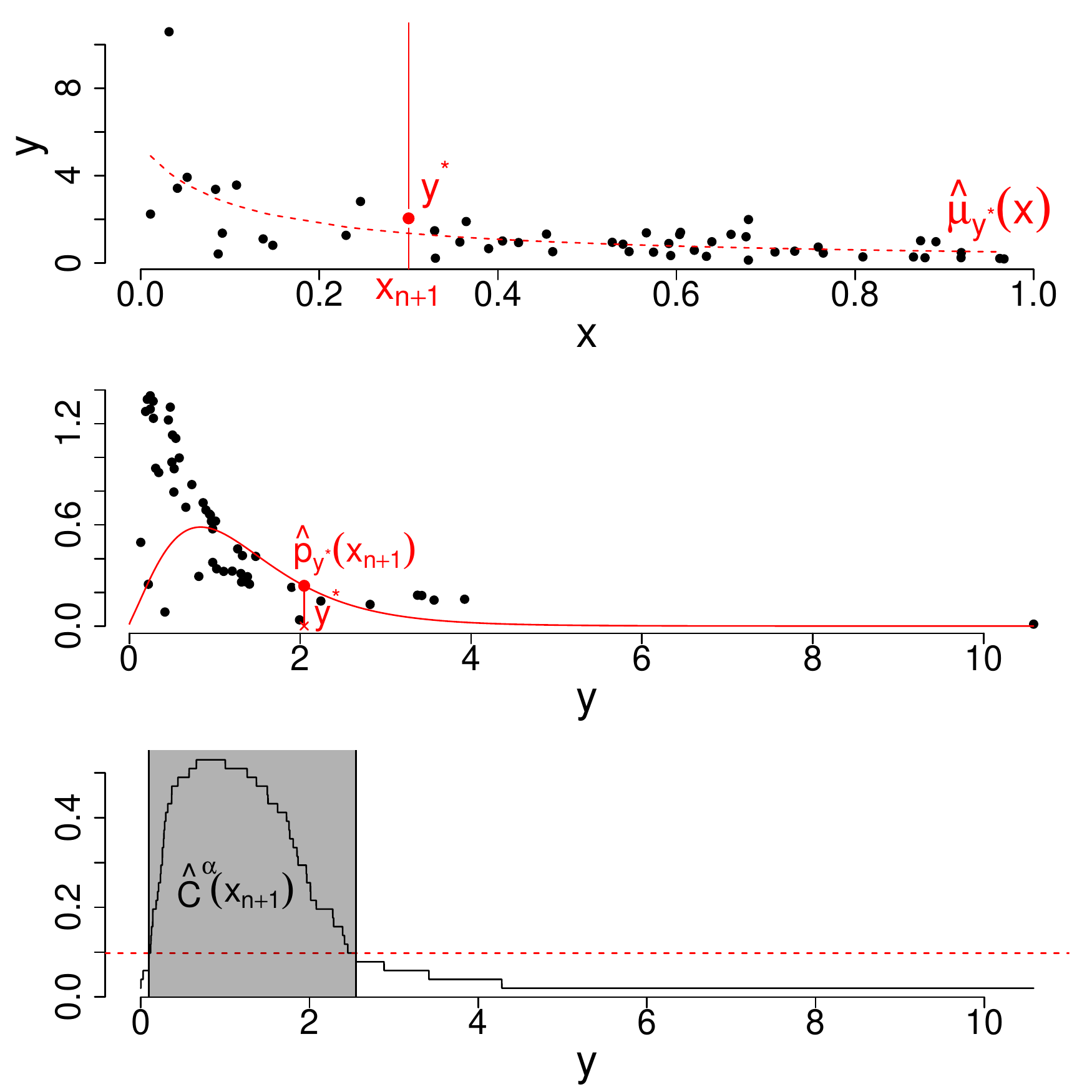}
\end{center}
\caption{ How conformal prediction regions are constructed. 
  The left panel shows conformal prediction for a univariate outcome $y$: 
  the top row shows the data points and corresponding histogram; 
  middle row shows conformity scores for a nonparametric kernel smoother; 
  bottom panel shows $\pi_{n+1,i}$ for each $i$, and the conformal prediction 
  set $\widehat{C}^{(\alpha)}$. 
  The right panel show conformal prediction for a univariate regression with 
  predictor $x$ and outcome $y$: the top row shows the $(x_i,y_i)$ pairs, 
  and a regression estimate $\hat\mu_y(x)$ computed using the additional point 
  $(x_{n+1}, y^*)$; middle panel shows the conformity scores obtained by using 
  the estimated density as the conformity measure; bottom panel shows 
  $\pi(y)$ at $x_{n+1}$, the proportion of density scores lower than the density 
  score at candidate $y$, and the conformal prediction set 
  $\widehat{C}^{(\alpha)}(x_{n+1})$. }
\label{Fig:schematic}
\end{figure}

\subsection{Notions of finite sample validity in regression}

Several notions of finite sample validity exist in the context of 
regression.  Suppose that we have an iid 
sample $(X_i,Y_i)$, $i = 1$, $\ldots$, $n$ where 
the predictor is $X \in \R^m$ and $Y\in\R$ is the outcome.  
We will suppose that $d$ of the $X$'s are main effects, 
where $d \leq m$, and the other $m - d$ terms in $X$ are functions of the 
$d$ main effects.  We will suppose that $X$ has support $\X = [0, 1]^d$.

\begin{defn}[marginal validity]
Let $0 < \alpha < 1$ be a desired error tolerance.  
Let $(X_i,Y_i)$, $i = 1$, $\ldots$, $n$ be an iid 
sample from a continuous distribution $P$.   
The prediction region $\Copthat$ has finite sample 
marginal validity if 
$
  \Prob\left\{Y_{n+1} \in \Copthat(X_{n+1})\right\} \geq 1 - \alpha,
$
where $(X_{n+1},Y_{n+1}) \sim P$.
\end{defn}

Finite sample marginal validity is guaranteed for conformal prediction 
methods by construction.  
Finite sample marginal validity alone may not be desirable when 
variability of the outcome is not constant across the support of the 
predictor.   
A second notion of finite sample validity arises by considering coverage of 
the prediction region conditional on a particular value of the predictor.

\begin{defn}[conditional validity]
Let $0 < \alpha < 1$ be a desired error tolerance.  
Let $(X_i,Y_i)$, $i = 1$, $\ldots$, $n$ be an iid 
sample from a continuous distribution $P$.   
The prediction region $\Copthat$ has finite sample 
conditional validity at $x$ when 
$
  \Prob\left\{Y_{n+1} \in \Copthat(x)|X_{n+1} = x\right\} \geq 1 - \alpha, 
$
where $(X_{n+1},Y_{n+1}) \sim P$.
\end{defn}

However, conditional validity at every $x$ is unattainable if we insist that 
$\Copthat$ is asymptotically of minimal length.  Let $p(y|x)$ be the 
conditional density corresponding to regression of $y$ on $x$.
Define $\salpha_x$ to be the value such 
that $P(Y_{n+1} \in \{y:p(y|x) \geq \salpha_x\}|X = x) = 1 - \alpha$  and define  
$
  C(x) = \left\{y: p(y|x) \geq \salpha_x\right\}
$
to be the optimal conditional prediction region (i.e., the minimal length 
prediction region).  
Let $\nu$ be Lebesgue measure and let $\triangle$ denote the symmetric set 
difference operator.  Then there does not exist a finite sample conditional 
valid prediction region $\Copthat(x)$ such that 
$
  \sup_{x \in \X}\nu\left\{\Copthat(x) \triangle C(x)\right\} 
    \overset{P}{\to} 0
$
holds, a result due to Lemma 1 of \citet{lei2014distribution}.  
A third notion of finite sample validity relaxes the stringent requirement of 
conditional validity at every possible $x$. 

\begin{defn}[local validity]
Let $0 < \alpha < 1$ be a desired error tolerance.  
Let $(X_i,Y_i)$, $i = 1$, $\ldots$, $n$ be an iid 
sample from a continuous distribution $P$.   
Let $\A = \{A_k: k \geq 1\}$ be a partition of $\X$.  
We say that the prediction region $\Copthat$ has finite sample 
local validity when 
$
  \Prob\left\{Y_{n+1} \in \Copthat(X_{n+1})|X_{n+1} \in A_k\right\} 
    \geq 1 - \alpha, 
$
for all $A_k \in \A$ where $(X_{n+1},Y_{n+1}) \sim P$.
\end{defn}

Local validity offers a bridge between marginal validity -- which is 
inappropriate in the presence of heterogeneity -- and conditional validity 
-- which is unattainable when we require that the prediction interval is 
asymptotically minimal \citep{lei2014distribution}.  The nonparametric 
conformal prediction region in Section~\ref{sec:nonparaconf} satisfies 
finite sample local validity and is asymptotically conditional valid 
\citep{lei2014distribution}.  In Section~\ref{sec:paraconf} we show 
that the parametric conformal prediction region also satisfies finite sample 
local validity and is asymptotically conditional valid at rate faster 
than its nonparametric counterpart.

\section{Conformal prediction for parametric models}
\label{sec:paraconf}

We now introduce parametric conformal prediction regions for a wide class 
of parametric families with continuous outcomes. 
We show that this class of parametric families includes continuous regular 
full exponential families, and hence GLMs by extension.
Two parametric conformal prediction regions are provided.  
The first of these guarantees finite-sample local validity using a similar 
binning technique as \citet{lei2014distribution}.  The second 
guarantees finite-sample marginal validity via transformation to a 
``common'' distribution.  

Suppose, for simplicity, that the 
support of the predictors is $\X = [0, 1]^d$.  Let $(X, Y)\sim P$ 
and let $(X_i, Y_i)\sim P$, $i = 1$, $\ldots$, $n$ be an iid 
sample where $P$ is a continuous distribution, 
$Y \in \R$, $X \in \R^m$ with $d \leq m$ main effects. 
We will denote the true conditional density of the regression 
model as $p_\psi(\cdot|x)$ and define $\psi \in \R^r$.  For each $x$, we 
define $\talpha_x$ to be the point such that 
$\Prob\left(\{y:p_\psi(y|x) \geq \talpha_x\}|X=x\right) = 1 - \alpha$ 
for $0 < \alpha < 1$.  The optimal or minimal length prediction region 
at $x$ is $\Copt_P(x) = \{y : p_\psi(y|x) \geq \talpha_x\}$.
Define the moment generating function (MGF) $M_{Y,x,\psi}(t)$ corresponding 
to the conditional density $p_\psi(\cdot|x)$.  The parameter 
space for this parametric family is 
$$
  \aleph_\X = \{\psi: \nabla_{\psi} \log p_\psi(y|x) = O(y),\; 
    M_{Y,x,\psi}(t)\; \text{exists for all}\; t\; 
    \text{in a neighborhood of}\; 0,\; \text{all}\; x \in \X\; \}.
$$

\subsection{Local validity via binning}

Let $\A = \{A_k : k \geq 1\}$ be a partition $\mathcal{X}$ and define the 
rates $r_n = O(\sqrt{\log(n)/n})$ and 
$z_n = O\{(\log(n)/n)^{1/d}\}$. 
The number of elements in the partition $\A$ increases at rate 
$1/r_n$ when $d \leq 2$ and 
$1/z_n$ when $d > 2$.
When we restrict the partition to be formed of equilateral cubes, we let 
the widths of these cubes decrease at rate $r_n$ or $z_n$.  
Let $p_\psi(y|x)$ be the true conditional density function and let 
$\hat{p}_\psi(\cdot|\cdot)$ be the estimated density fit 
using the original data and define $\hat{p}_\psi^{(x,y)}(\cdot|\cdot)$ as the 
estimated density fit to the augmented data 
$\{(X_1,Y_1),\ldots,(X_n,Y_n),(x,y)\}$.  We will use this estimated density 
as our conformity measure.  The local conformal prediction region is then
\begin{equation} \label{confregion}
  \Coptglm(x) = \left\{y : \frac{1}{n_k + 1}\sum_{i=1}^{n + 1}
    \indicator{X_i \in A_k}\indicator{\hat{p}_\psi^{(x,y)}(Y_i|X_i) 
      \leq \hat{p}_\psi^{(x,y)}(y|x)} 
    \geq \widetilde{\alpha}_k\right\},
\end{equation}
where $\widetilde{\alpha}_k = \floor{(n_k + 1)\alpha}/(n_k+1)$ for 
$0 <\alpha <1$.  

A schematic of how this conformal 
prediction region is constructed is displayed in the right-hand side of 
Figure~\ref{Fig:schematic}.  We now establish local validity for 
$\Coptglm$.

\begin{lem} \label{lem:valid}
Let $\A$ be a partition of $\X$ and let $\Coptglm(x)$ be as defined in 
\eqref{confregion}. 
Then $\Coptglm(x)$ is finite sample locally valid for $x \in A_k \in \A$. 
\end{lem}

\begin{proof}
The proof follows from the proof of \cite[Proposition 2]{lei2014distribution}. 
Fix $k$ and let 
$\{i_1$, $\ldots$, $i_{n_k}\} = \{i:1\leq i \leq n$, $X_i\in A_k\}$. 
Let $(X_{n+1}$, $Y_{n+1})\sim P$ be another independent sample. Let 
$i_{n_{k+1}} = n+1$ and $\sigma_{i_j} = \hat{p}_\psi^{(x,y)}(Y_{i_j}|X_j)$ 
for all $1 \leq j \leq n_k + 1$.  Then conditioning on $X_{n+1} \in A_k$ and 
$(i_1$, $\ldots$, $i_{n_k})$, the sequence 
$(\sigma_{i_1}$, $\ldots$, $\sigma_{i_{n_k}}$, $\sigma_{i_{n_k + 1}})$ is 
exchangeable. 
\end{proof}

We now establish that the parametric conformal prediction region $\Coptglm$ 
is asymptotically of minimum length over $\X$.  

\noindent{\bf Assumption 1} (Assumption 1 (a) of \cite{lei2014distribution}).    
Let $\X$ be the support of $X$ and assume that $\X = [0,1]^d$. 
The marginal density of $X$ satisfies $0 < b_1 \leq p_X(x) \leq b_2 < \infty$ 
for all $x \in \X$. 

\noindent{\bf Assumption 2}.    
The conditional density $p_\psi(\cdot|x)$ is Lipschitz in both $x$ and $\psi$, 
there exists some $L_1,L_2 \in \R$ 
such that,
$
  \|\ptrue(\cdot|x_1) - \ptrue(\cdot|x_2)\|_{\infty} 
    \leq L_1\|x_1 - x_2\|,
$
and
$
  \|p_{\psi_1}(\cdot|x) - p_{\psi_2}(\cdot|x)\|_{\infty} 
    \leq L_2\|\psi_1 - \psi_2\|. 
$

\noindent{\bf Assumption 3} (smoothness condition of \citet{lei2014distribution}).
There exists positive constants $\varepsilon, \gamma, c_1$, and $c_2$ such that, 
for all $x \in \X$, and all $\psi \in \aleph_\X$, 
$$
  c_1\varepsilon^\gamma 
    \leq \Prob[\{y:|p_\psi(y|x) - t_x^{(\alpha)}|X=x\}]
    \leq c_2\varepsilon^\gamma,
$$ 
for all $\varepsilon \leq \varepsilon_0$. Moreover, 
$\inf_x t_x^{(\alpha)} \geq t_0 > 0$.

\begin{thm}
Let $Y|X=x$ be a random variable with conditional density 
$p_\psi(\cdot|x)$ and parameter space $\aleph_\X$.  Assume that 
$p_\psi(\cdot|x)$ is twice differentiable in $\psi$ and satisfies 
Assumptions 2 and 3.  Let $(X_1,Y_1)$, $\ldots$, $(X_n,Y_n)$ be independent 
and identically distributed copies of $(X,Y)$.  Let $\psi \in \R^r$.  
Let $\psihat$ be the MLE of $\psi$.  Augment the sample data with a new 
point $(x,y)$, and let $\psihat^{(x,y)}$ be the MLE of $\psi$ with respect 
to the augmented data.  
Let $\hat{p}^{(x,y)}_\psi(\cdot|x)$ be the conditional density with 
$\psihat^{(x,y)}$ plugged in for $\psi$.  
Suppose that Assumption 1 holds.  Let $0 < \alpha < 1$ and $\Coptglm$ be the 
prediction band given by \eqref{confregion}.  Then, for a given $\lambda > 0$, 
there exists a numerical constant $\zeta_{\lambda}$ such that 
\begin{equation} 
  \Prob\left[\sup_{x \in \X} 
    \nu\left\{\Coptglm(x) \triangle C_P^{(\alpha)}(x)\right\} 
      \geq  \zeta_{\lambda}(z_n \vee r_n) \right] 
  = O\left(n^{-\lambda}\right),
\label{mainresult}
\end{equation}
where $r_n = O(\sqrt{\log(n)/n})$.
\label{thm:correctmodel}
\end{thm}

{\bf Remarks}:

\begin{enumerate}

\item The proof of Theorem~\ref{thm:correctmodel} is given in the Appendix. 
The rates $r_n$ and $z_n$ are appreciably faster than that of the 
nonparametric conformal prediction band which has a convergence rate of 
$w_n = O\{\log(n)/n\}^{1/(d+3)}$ when the underlying $p_\psi(y|x)$ has 
parameter space $\aleph_\X$.  The difference between the convergence speed 
of the parametric and nonparametric conformal prediction regions originates 
from the differences in the rates of MLEs and nonparametric techniques and 
the speed at which the bin widths shrink.  

\item The following Lemma governs how 
fast the bins can shrink while still maintaining that $n_k \to \infty$.

\begin{lem}[Lemma 9 in \citet{lei2014distribution}]
Under Assumption 1, there exists constants $C_1$ and $C_2$ such that 
$$
  \Prob\left(\forall k : b_1nw_n^d/2 \leq n_k \leq 3b_2nw_n^d/2\right) 
    \geq 1 - C_1 w_n^{-d}\exp(-C_2nw_n^{d}), 
$$
with $b_1$ and $b_2$ defined in Assumption 1. 
\label{lem:Bern}
\end{lem}

\item The key term in Lemma~\ref{lem:Bern} is the $nw_n^{d}$ term appearing in 
the exponent.  The proof of Theorem~\ref{thm:correctmodel} reveals that the 
rate of convergence of the parametric conformal prediction region is limited 
by Lemma~\ref{lem:Bern}.  This is not the case for the nonparametric conformal 
prediction region where the limiting factor is the rate of convergence of the 
kernel density estimator. 

\item The proof of Theorem~\ref{thm:correctmodel} requires a new concentration 
inequality for MLEs which is stated below as Theorem~\ref{concentration}.  
The proof technique that we employ for Theorem~\ref{thm:correctmodel} involves 
sandwiching $\Coptglm(x)$ by two estimated level sets, formalized as 
Lemma~\ref{lem:sandwich-subexpo} in the Appendix.  The convergence rate of these 
sandwiching level sets depends on $r_n$, $z_n$, and the convergence rate of 
$
  \sup_{x \in \X}
    \|\hat{p}_\psi^{(x,y)}(\cdot|x) - p_\psi(\cdot|x)\|_{\infty},
$ 
see Lemma~\ref{lem:dens} in the Appendix.  We verify that  
$$
  \sup_{x \in \X}
    \|\hat{p}_\psi^{(x,y)}(\cdot|x) - p_\psi(\cdot|x)\|_{\infty}
  = O\left(|\psihat - \psi|\right),
$$ 
where $\psi \in \R^r$ are the parameters in density $p_\psi(\cdot|x)$  
and $\psihat$ be the MLE of $\psi$, hence the need for 
Theorem~\ref{concentration}.
\end{enumerate}

Let $|\cdot|$ denote the $L^1$ norm of a vector or the absolute value of a 
scalar.

\begin{thm}
Let $Y|X=x$ be a random variable with conditional density 
$p_\psi(\cdot|x)$ and parameter space $\aleph_\X$.  Assume that 
$p_\psi(\cdot|x)$ is twice differentiable in $\psi$ and satisfies 
Assumption 2.  Let $(X_1,Y_1)$, $\ldots$, $(X_n,Y_n)$ be independent and 
identically distributed copies of $(X,Y)$.  Let $\psi \in \R^r$.  
Let $\psihat$ be the MLE of $\psi$.  
Then for any $\lambda > 0$, there exists a numerical constant 
$A_{\lambda}$, 
such that  
\begin{equation}
  \Prob\left\{ 
    \sqrt{n}\left|\psihat - \psi\right| \geq A_{\lambda}\sqrt{\log(n)} 
  \right\}
  = O\left(n^{-\lambda}\right).
\label{precise}
\end{equation}
\label{concentration}
\end{thm}

The concentration inequality given in Theorem~\ref{concentration} is a 
generalization of Theorem 3.3 in \cite{miao2010concentration} to 
subexponential random variables with powers of $\log(n)/n$ replacing $r$, 
provided that the underlying density is parameterized with $\aleph_X$.  
Theorem 3.3 in \cite{miao2010concentration} only holds for subgaussian 
random variables.  Our extension is possible because of subexponential 
concentration theory and the score function is of the same order as 
the outcome variable.  These results circumvent the necessity for the use 
of optimal transport theory 
\citep{bobkov1999transport, ledoux1999logsobolev, ledoux2001concentration, 
  djellout2004transport, villani2008optimal} 
to prove Theorem 3.3 in \cite{miao2010concentration}.

\subsection{Transformation conformal}

Transformation parametric conformal prediction was proposed 
by \citet{chernozhukov2019distributional} to construct a
$(1-\alpha)\times 100\%$ prediction region using the lower 
and upper $\alpha/2$ quantiles of the assumed parametric distribution.  
We build on this framework 
to propose transformation based parametric conformal prediction regions  
that are asymptotically minimal length at rate $r_n$.  
Let $m(x,\psi)$ be the number of modes of $p_\psi(\cdot|x)$ and 
suppose that $\sup_{x \in \X,\, \psi \in \aleph_X} m(x, \psi) = m < \infty$.  
With this specification we can construct $C_P(x)$ as a finite union of intervals.  
Let $F(\cdot|x)$ be the conditional distribution function with density 
$p_\psi(\cdot|x)$ and let $\widehat{F}(\cdot|x)$ be the estimate of 
$F(\cdot|x)$ with the MLE $\psihat$ plugged in.  
Let $\widehat{F}^{(x,y)}(\cdot|\cdot)$ be the estimate of $F(\cdot|\cdot)$ 
with the MLE corresponding to the augmented data $\psihat^{(x,y)}$ plugged 
in.  Let $(X_{n+1}, Y_{n+1}) = (x,y)$ be a candidate data point.   

The logic of transformation conformal prediction is as follows.  
We first augment the original data by adding the point $(x,y)$.  
We then estimate model parameters via maximum likelihood estimation with 
respect to this augmented dataset, and denote this estimator as 
$\hat\psi^{(x,y)}$.  The estimator $\hat\psi^{(x,y)}$ is then used to 
estimate the distribution function $F(\cdot|\cdot)$ by 
$\widehat{F}^{(x,y)}(\cdot|\cdot)$ via plugin.  We then transform outcomes 
by computing $\widehat{U}^{(x,y)}_i = \widehat{F}^{(x,y)}(Y_i|X_i)$ 
for $i=1,\ldots,n+1$.  Define 
$
  \tilde \pi(y) = (n+1)^{-1}\sum_{i=1}^{n+1}
    \indicator{\sigma\left(\widehat{U}^{(x,y)}_i\right) 
      \leq \sigma\left(\widehat{U}^{(x,y)}_{n+1}\right)},
$
and the conformal prediction region 
\begin{equation} \label{Ctilde}
  \Copttilde(x) 
    = \left\{y : \tilde \pi(y) \geq \alpha'\right\},     
\end{equation}
where $0 < \alpha' \leq \alpha < 1$.

\begin{lem}
Let $\Copttilde(x)$ be as defined in \eqref{Ctilde}. 
Then $\Copttilde(x)$ is finite sample marginally valid. 
\end{lem}

\begin{proof}
The proof follows from the proof of Lemma~\ref{lem:valid}.
Let $(X_{n+1}$, $Y_{n+1})\sim P$ be another independent sample. The random 
variables $\widehat{U}^{(x,y)}_i$, $1 \leq i \leq n + 1$, are 
exchangeable.  
Then, conditioning on $X_{n+1} = x$, the sequence 
$
  \left(\sigma(\widehat{U}^{(x,y)}_1), \ldots, 
    \sigma(\widehat{U}^{(x,y)}_{n+1})\right)
$ 
is exchangeable.  Marginal validity follows by construction.  
\end{proof}

A particular choice of the conformity measure $\sigma(\cdot)$ 
and $\alpha'$ yields an asymptotically minimal length conformal prediction 
region at rate $r_n$.  First, let 
$\widehat{U}^{(x,y)}_{n+1} = \widehat{F}^{(x,y)}(y|x)$ be the transformed 
augmented outcome.  
We can find a minimal length $1-\alpha$ high density region based on the 
estimated density $\hat{p}_\psi^{(x,y)}(\cdot|x)$ through solution of the 
optimization problem
\begin{equation} 
  \label{general-minproblem}
  \hat A(x) = \text{argmin}_{A(x)} \nu(A(x)), 
    \qquad \text{subject to} \qquad \int_{A(x)} \hat{p}_\psi^{(x,y)}(y'|x) dy' 
      = 1 - \alpha,
\end{equation}  
where the solution $\hat A(x) = \cup_{k=1}^{m'(x)}(\hat a_k(x), \hat b_k(x))$ is 
a union of $m'(x)$ disjoint intervals.  Because $p_\psi(\cdot|x)$ has $m<\infty$ 
modes for every $x\in\X$, the number of intervals in the union is $m'(x) \leq m$.  
Let 
$$ 
\widehat{U}^{(x,y)}_{\text{lwr}, k} = \widehat{F}^{(x,y)}(\hat a_k(x)|x)
  \quad \text{and} \quad 
    \widehat{U}^{(x,y)}_{\text{upr}, k} = \widehat{F}^{(x,y)}(\hat b_k(x)|x)
$$
be the transformed values corresponding to the lower and upper limits of the 
prediction set for each $x\in\X$.  

Then, for each $k$, find intervals in the transformed observed values that 
contain $(a_k(x), b_k(x))$, 
$$
  \widetilde{U}^{(x,y)}_{\text{lwr}, k} = \begin{cases}
    \widehat{U}^{(x,y)}_{\left[\floor{(n+1)\widehat{U}^{(x,y)}_{\text{lwr},k}}\right]} 
      & \text{if } \floor{(n+1)\widehat{U}^{(x,y)}_{\text{lwr},k}} \ge 1\\
    0 & \text{otherwise}
    \end{cases}
$$
and
$
  \widetilde{U}^{(x,y)}_{\text{upr}, k} = 
    \widehat{U}^{(x,y)}_{\left[\ceil{(n+1)\widehat{U}^{(x,y)}_{\text{upr},k}}\right]}
$ 
where $\widehat{U}^{(x,y)}_{[j]}$ is the $j$th order statistic of the transformed values. 
Define the set consisting of the union of these intervals 
$$
  \widehat{B}^{(x,y)} = \cup_{k=1}^{m'(x)}
  \left(\widetilde{U}^{(x,y)}_{\text{lwr}, k}, 
    \widetilde{U}^{(x,y)}_{\text{upr}, k}\right),
$$ 
and define the conformity measure as the indicator that the value $u$ is in the estimated 
high density region, $\sigma(u) = \mathds{1}\{u \in \widehat{B}^{(x,y)} \}$.  
Let the threshold be 
$$
  \alpha' = 1 - 
    \sum_{k=1}^{m'(x)}\left(\frac{\ceil{(n+1)\widetilde{U}^{(x,y)}_{\text{upr}, k}}}{n+1} 
      - \frac{\floor{(n+1)\widetilde{U}^{(x,y)}_{\text{lwr}, k}}}{n+1}\right)
  \leq \alpha .
$$
With these specifications, the conformal prediction region~\eqref{Ctilde} becomes
\begin{equation} \label{transregion}
  \Copttrans(x) = \left\{y: \widehat{U}^{(x,y)}_{n+1} \in \cup_{k=1}^{m'(x)}
  \left(\widetilde{U}^{(x,y)}_{\text{lwr}, k}, 
    \widetilde{U}^{(x,y)}_{\text{upr}, k}\right)
  \right\},
\end{equation}
where $y \in \Copttrans(x)$ implies that 
$\sigma(\widehat{U}^{(x,y)}_{n+1}) = 1$. 

While this procedure may appear complicated, it is motivated by a simple observation.  
When the parametric model is correctly specified, the $\widehat{U}^{(x,y)}_i$ 
random variables have an approximate Uniform$[0,1]$ distribution in finite samples 
and are asymptotically Uniform$[0,1]$.  By conducting the conformal prediction procedure 
in the transformed space, we can leverage distributional information from the 
estimated density $\hat{p}_\psi^{(x,y)}(\cdot|x)$ to mimic a highest density 
region while ensuring nominal marginal coverage.  Even when the model is misspecified, 
this procedure is designed to guarantee nominal marginal coverage.

We now establish that the transformation parametric conformal prediction 
region $\Copttrans$ is asymptotically of minimum length over $\X$ at rate $r_n$.  

\noindent{\bf Assumption 4}.    
Let $Y|X=x$ be a random variable with conditional density $p_\psi(\cdot|x)$ 
and parameter space $\aleph_\X$.  Let $F(y|x)$ be the conditional distribution 
function corresponding to $p_\psi(\cdot|x)$.  Further assume that 
$|\nabla_\psi F(v|x)|$ is bounded for all $v$ and all $\psi \in \aleph_X$.

\begin{thm}
Let $Y|X=x$ be a random variable with conditional density 
$p_\psi(\cdot|x)$ and parameter space $\aleph_\X$.  Assume that 
$p_\psi(\cdot|x)$ has $m$ modes, is twice differentiable in $\psi$, 
and satisfies Assumption 2.  
Assume that the corresponding distribution function satisfies Assumption 4.  
Let $(X_1,Y_1)$, $\ldots$, $(X_n,Y_n)$ be 
independent and identically distributed copies of $(X,Y)$.  
Let $\psi \in \R^r$ and let $\psihat$ be the MLE of $\psi$.  
Augment the sample data with a new point $(x,y)$, and let 
$\psihat^{(x,y)}$ be the MLE of $\psi$ with respect to the augmented data.  
Let $\hat{p}^{(x,y)}_\psi(\cdot|x)$ be the conditional density with 
$\psihat^{(x,y)}$ plugged in for $\psi$.  
Suppose that Assumption 1 holds.  Let $0 < \alpha < 1$ and $\Copttrans$ be 
the prediction band given by \eqref{transregion}.  Then, for a given 
$\lambda > 0$, there exists a numerical constant $\chi_{\lambda}$ such that 
\begin{equation} 
  \Prob\left[\sup_{x \in \X} 
    \nu\left\{\Copttrans(x) \triangle C_P^{(\alpha)}(x)\right\} 
      \geq  \chi_{\lambda}r_n\right] 
  = O\left(n^{-\lambda}\right).
\label{mainresult2}
\end{equation}
\label{thm:correctmodel2}
\end{thm}

{\bf Remarks}: 

\begin{enumerate}

\item In the simple case when $p_\psi(y|x)$ is unimodal, 
the optimization problem~\eqref{general-minproblem} is replaced with,
\begin{equation} \label{minproblem}
  (\hat a(x),\hat b(x)) = \text{argmin}_{(a,b)} \left(b - a\right), 
    \qquad \text{subject to} \qquad \int_a^b \hat{p}_\psi^{(x,y)}(u|x) du 
      = 1 - \alpha,\; \text{and} \; b - a > 0.
\end{equation}  
In this case, we then construct 
$\widehat{U}^{(x,y)}_{\text{lwr}} = \widehat{F}^{(x,y)}(\hat a(x)|x)$ 
and $\widehat{U}^{(x,y)}_{\text{upr}} = \widehat{F}^{(x,y)}(\hat b(x)|x)$, 
and set 
$
  \widetilde{B}^{(x,y)} = (\widehat{U}^{(x,y)}_{\text{lwr}}, 
    \widehat{U}^{(x,y)}_{\text{upr}}).
$
From these quantities, we construct the conformity measure
$
  \sigma(u) = \mathds{1}_{(\widetilde{B}^{(x,y)})}(u),
$
specify that 
$$
  \alpha' = 1 - 
    \left(\frac{\ceil{(n+1)\widehat{U}^{(x,y)}_{\text{upr}}}}{n+1} 
      - \frac{\floor{(n+1)\widehat{U}^{(x,y)}_{\text{lwr}}}}{n+1}\right)
  \leq \alpha,
$$
and build $\Copttilde(x)$ with respect to these specifications as
\begin{equation} \label{transregion2}
  \Copttilde(x) = \left\{y:  
    \widehat{U}^{(x,y)}_{\left[\floor{(n+1)\widehat{U}^{(x,y)}_{\text{lwr}}}\right]} 
    \leq \widehat{U}^{(x,y)}_{n+1} \leq
    \widehat{U}^{(x,y)}_{\left[\ceil{(n+1)\widehat{U}^{(x,y)}_{\text{upr}}}\right]}
  \right\},
\end{equation}
where $y \in \Copttilde(x)$ implies that 
$\sigma(\widehat{U}^{(x,y)}_{n+1}) = 1$. 

\item The proof of Theorem~\ref{thm:correctmodel2} is given in the Appendix.  
Our proof technique is for the $m = 1$ case.  
The proof of Theorem~\ref{thm:correctmodel2} requires the concentration 
inequality for MLEs stated in Theorem~\ref{concentration}.  Under our proof 
technique the convergence rate $r_n = \sqrt{\log(n)/n}$ of $\Copttrans$ 
is sharp.  

\item The rates $r_n$ is appreciably faster than that of the 
nonparametric conformal prediction band and the binned parametric conformal 
prediction region $\Coptglm$ in \eqref{confregion} when the underlying 
density $p_\psi(y|x)$ has parameter space $\aleph_\X$.  The region $\Coptglm$ 
converges more slowly slower than $\Copttrans$ because the former 
prediction region is hampered by $n_k$, the number of observations per bin, 
as the bin widths shrink at rate $w_n$.  
However, our simulation studies show that $\Coptglm$ performs better than 
$\Copttrans$ when modest departures from the assumed model are present,  
due to its finite sample local validity guarantees.

\item The transformation parametric conformal prediction 
region in \citet{chernozhukov2019distributional} has generality beyond 
the context of the maximum likelihood estimation based methodology that 
we propose.  However, explicit convergence rates for their conformal 
prediction rates are not given and we speculate that they would be slower for 
the entire class of models for which their methodology is appropriate.  
In addition, our methodology leverages the underlying true distribution to 
construct the minimal length $(1-\alpha)\times 100\%$ prediction region under 
the assumed model when on the scale of the approximately uniform random 
variables.  This is instead of choosing the end points corresponding to the 
lower and upper $\alpha/2$ quantiles of the transformed approximate uniform 
variables.  Our focus on minimal length prediction regions and relatively fast 
convergence rates addresses the efficiency and utility concerns of conformal 
prediction in the context of regression which were raised in 
\citet{cortes2019concepts}.

\end{enumerate}

\subsection{Example: exponential family distributions and generalized linear models}
\label{sec:glms}

In this Section we show that a broad class of continuous GLMs can be 
parameterized by $\aleph_\X$ so that the conformal prediction methodology 
is applicable. 
The exponential families that we consider have densities 
$p_\theta : \R^p\to\R$ defined by 
\begin{equation}
  p_{\theta}(y) = e^{\inner{y, \theta} - c(\theta)}.
\label{density}
\end{equation}
We refer to $y$ and $\theta$ as the canonical statistic and canonical 
parameter respectively.
Densities of the form $p_{\theta}$ are integrated with respect to a 
\emph{generating measure} $\mu$, where $c(\theta)$ is the cumulant function 
defined as $c(\theta) = \log\int \exp(\inner{y, \theta})\, \mu(d y)$.
The \emph{effective domain} of $c$ is 
$
   \dom c = \left\{ \theta \in \R : c(\theta) < + \infty \right\}.
$
The exponential family is \emph{full} if $\Theta = \dom c$ and is 
\emph{regular} if $\Theta$ is an open set.  
  
We will only consider scalar outcome variables.  However, we can extend our 
setup to multidimensional canonical statistics, the normal distribution being 
an example with canonical statistic $z \mapsto (z,z^2)'$.  Thus we will assume 
that the canonical statistic vector lies on a one-dimensional manifold in the 
event that the canonical statistic is multidimensional.

When the density \eqref{density} corresponds to a generalized linear 
regression model we will re-parameterize  $\theta = f(x'\beta, \phi)$, 
where $f:\R^p\to\R^p$ is continuous in both arguments,  
the vectors $x \in \R^m$ and $\beta \in \R^m$ are a 
predictors and regression coefficients respectively, 
and 
$\phi \in \R^{p-1}$ is a vector of nuisance parameters.  
Let $\psi = (\beta, \phi)$ and set $r = m + p -1$.  We specify that the 
base set (main effects) of predictors have support $\X \in\R^d$, $d \leq m$.  
Here we assume, without loss of generality, that $\E(Y_1|x) = g^{-1}(x'\beta)$ 
where $g:\R\to\R$ is a link function and $Y_1$ is the first component of the 
canonical statistic vector.  As an example of this reparameterization, 
consider the simple linear regression model with homoskedastic normal errors 
with variance given by $\sigma^2$.  In this example $d = m = 1$, $r = 2$, and 
$\phi = \sigma^2$.  The link function $g$ is taken to be the identity 
function and $\mu_x = \E(Y|x) = x'\beta$.  

The reparameterized density corresponding to the generalized linear regression 
model is then 
\begin{equation}
  p_\psi(y|x) = \exp\left[
    \inner{y, f(x'\beta, \phi)} - c\left\{f(x'\beta, \phi)\right\}
  \right]
\label{glm}
\end{equation}
with respect to generating measure $\mu$.  We will further assume that $\X$ is 
bounded and that the exponential family is full with parameter 
space given by
$
  \Theta_\X = \{\beta \in \R^m,\; \phi \in \R^{p-1}: 
    c\left\{f(x'\beta, \phi)\right\} < \infty, 
      \; \text{for all} \; x \in \X\}.
$
Consider the multiple linear regression model with homoskedastic normal 
errors with variance given by $\sigma^2$.  Here we have 
$f(x'\beta, \phi) = (x'\beta, \sigma^2)'$ and 
$
  \Theta_\X = \{(\beta', \sigma^2)': 
    \beta \in \R^m, \; \sigma > 0, x \in \X\}.
$  
We will assume that the link function of the generalized linear regression 
model is known, the model is correctly specified, and that $\hat{\beta}$ is a 
$\sqrt{n}$-consistent estimator of the regression coefficients $\beta$. 
All of the continuous exponential families implemented in the \texttt{glm} 
function in R have densities that can be parameterized as \eqref{glm}.  
Moreover, we have that $\Theta_\X \subseteq \aleph_\X$ since,  
\begin{align*}
  &\frac{\partial}{\partial \psi} \log p_\psi(y|x) 
    = \frac{\partial}{\partial \psi}\left(
      \inner{y, f(x'\beta, \phi)} - c\left\{f(x'\beta, \phi)\right\}
      \right) \\
  &\qquad= \left(\frac{\partial f(x'\beta, \phi)}{\partial \psi}\right)
      \left(y - \frac{\partial}{\partial f(x'\beta,\phi)}c\left\{f(x'\beta, \phi)\right\}\right)
  = O(y),
\end{align*}
and $M_{Y,x,\psi}(t)$ exists for all $t$ in a neighborhood of 0, for all 
$x \in \X$ , and $\psi \in \Theta_\X$.  Therefore our conformal 
prediction regions are appropriate for the GLMs parameterized in this 
Section.  We show that our GLM parameterization satisfies 
Assumptions 2 and 3 in the Supplementary Materials. 

We describe some related work on conformal prediction techniques for 
regression settings in the next sections.  Finite-sample performance of our 
parametric conformal prediction region and these related methods is assessed 
in Section~\ref{sec:simulations}.

\section{Existing conformal prediction techniques in regression}

\subsection{The nonparametric kernel estimator}
\label{sec:nonparaconf}

\citet{lei2014distribution} proposed a nonparametric conformal prediction 
region for continuous outcomes that satisfies finite sample local validity.   
Partitioning of the predictor space 
is performed so that the nonparametric conformal prediction region within 
each partition achieves finite sample marginal validity.
When we restrict the partition to be formed of equilateral cubes, we let the 
widths of these cubes decrease at rate $w_n$.    
Let $n_k = \sum_{i=1}^n 1(X_i \in A_k)$.  
Let $K(\cdot)$ be a non-negative kernel function. The estimated local marginal 
density of $Y$ is
$
  \ptilnon(v|A_k) = (n_kh_n)^{-1}\sum_{i=1}^n \indicator{X_i \in A_k} 
    K\left((Y_i-v)/h_n\right)
$
where $h_n$ is the bandwidth. The density estimate given a new pair
$(x$,$y) \in A_k \times \R$ is
$$
  \ptilnon^{(x,y)}(v|A_k) = \frac{n_k}{n_k + 1}\ptilnon(v|A_k) 
    + \frac{1}{(n_k + 1)h_n}K\left(\frac{v-y}{h_n}\right).
$$
The local conformity rank is then
$$
  \widetilde{\pi}_{n,k}(x,y) = \frac{1}{n_k + 1}\sum_{i=1}^{n+1}1(X_i\in A_k)
    \indicator{\ptilnon^{(x,y)}(Y_i|A_k) \leq \ptilnon^{(x,y)}(Y_{n+1}|A_k)},  
$$
and the conformal prediction band is 
$
  \Coptloc(x) = \left\{y : \widetilde{\pi}_{n,k}(x,y) \geq \alpha\right\},
$
for all $x \in A_k$.  The intuition for constructing the conformal 
prediction region in this manner is that, asymptotically, the width $w_n$ 
shrinks so that the conditional distributions $Y|X=x$ for all $x\in A_k$ 
become very similar to that of $Y|X=x$ for any particular $x\in A_k$.  
\citet{lei2014distribution} provided that rate of convergence for which 
$\Coptloc(x)$ is asymptotically of minimal length.  In the Supplementary 
Materials we show that when the underlying model is a GLM, $\Coptloc(x)$ 
is asymptotically of minimal length at rate $w_n$.  

The nonparametric conformal prediction region is desirable when the underlying 
distribution is unknown or analytically intractable.  However, its convergence 
is slow when compared with its parametric counterpart.  Our analytical 
results and simulation studies show that the parametric conformal prediction 
region is preferable to the nonparametric conformal prediction region when the 
GLM is reasonably close to correctly specified.
In the Supplementary Materials, we verify that the nonparametric conformal 
region of \citet{lei2014distribution} is appropriate for the class of GLMs 
that we consider.

\subsection{Normalized residuals}
\label{sec:residconf}

\citet{lei2018distribution} proposed a prediction region obtained from 
conformal prediction for residuals, which is appropriate when errors are 
symmetric about the mean function $\mu(x)$. 
\citet[Section 5.2]{lei2018distribution} also proposed an extension of 
their conformal prediction procedure that is appropriate when the errors 
about $\mu(x)$ exhibit heterogeneity but remain symmetric.  
This extension involves a dispersion function $\rho(x)$ that 
captures the changing variability across $x$ and is used to weight the 
residuals so that theses weighted residuals have the same magnitude, 
on average, across $x$.  The conformal procedure of 
\citet[Section 5.2]{lei2018distribution}, denoted in this paper as the least 
squares locally weighted (LSLW) prediction region, proceeds as follows. 
When the mean regression estimator $\hat{\mu}(x)$ of $\mu(x)$ is a symmetric 
function of the data points, augment the original data with a new point 
$(x,y)$, and estimate the mean function $\hat{\mu}_y$ and the dispersion 
function $\hat{\rho}_y$ with respect to the augmented data.  

Define the normalized absolute residuals as 
$
  R_{y,i} = |Y_i - \hat{\mu}_y(X_i)|/\hat{\rho}_y(X_i), 
$
for $i = 1,\ldots,n+1$. 
The normalized absolute residual $R_{y,n+1}$ is an example of an 
anti-conformity measure, in which a low value of $R_{y,n+1}$ indicates 
agreement between $y$ and the estimated regression function.
As in \citet{lei2018distribution}, we specify the dispersion function as the 
conditional mean absolute deviation of $(Y - \mu(X))|X=x$ as a function 
of $x$.  Let 
$
  \pi_{\text{LSLW}}(y) 
    = (n+1)^{-1}\sum_{i=1}^{n+1}\indicator{R_{y,i} \leq R_{y,n+1}} 
$
be the proportion of points with normalized residuals smaller than the 
proposed normalized residual $R_{y,n+1}$.  Define the LSLW conformal prediction 
regions as 
$
  \Copt_{\text{LSLW}}(x) 
    = \left\{y \in \R : (n+1)\pi_{\text{LSLW}}(y) 
      \leq \ceil{(1-\alpha)(n+1)}\right\}. 
$
The least squares (LS) conformal prediction region is constructed as in 
$\Copt_{\text{LSLW}}(x)$ with no weighting with respect to residuals, 
$\hat{\rho}_y = 1$.  These conformal prediction regions achieve finite sample 
marginal validity as a consequence of exchangeability of the data and symmetry 
of $\hat{\mu}_y$ in its arguments \citep{lei2018distribution}.  

In our simulation studies we find that this approach to conformal prediction 
performs well when there are slight deviations from the symmetric error 
assumption.  However, in such settings these conformal prediction regions can 
be larger than desired, give larger than desired coverage, or give larger 
prediction errors than other conformal prediction regions.  Large departures 
from the symmetric errors assumption prove problematic for the LSLW 
conformal prediction region.

\section{Simulation study}
\label{sec:simulations}

We consider three simulation settings to evaluate the performance of conformal 
prediction regions under correct model specification and model 
misspecification.  These simulation settings are:
\begin{itemize}
\item[A)] Gamma regression with $\beta = [1.25, -1]'$ and $n = 150$. 
Data are generated from a Gamma regression model, and the 
parametric conformal and highest density prediction regions are correctly 
specified.  A cubic regression model is assumed for the LS and 
LSLW conformal prediction regions.  
\item[B)] Gamma regression with $\beta = [0.5, 1]'$ and $n = 150$. 
Data are generated from a Gamma regression model, and a cubic regression model 
with homoskedastic normal errors is assumed for the highest density prediction 
region and the misspecified parametric, LS, and LSLW conformal prediction regions. 
\item[C)] Simple linear regression with 
 $\beta = [2, 5]'$, and normal errors with constant variance $\sigma^2 = 1$.  
Results are considered for sample sizes 
$n \in \{150, 250, 500\}$.  In this setting the regression model is correctly 
specified for the highest density prediction region and the parametric, 
LS, and LSLW conformal prediction regions.
\end{itemize}

The prediction regions under consideration are the 
transformation based parametric,
binned parametric,
binned nonparametric,
LSLW,
and LS conformal prediction regions and the highest density (HD) prediction 
region under an assumed model. The transformation based parametric, binned 
parametric, and binned nonparametric are fit using the \texttt{conformal.glm} 
R package \citep{eck2018conformalR} with default settings employed which specifies 
line search precision at 0.005.  The LSLW and LS are fit using the 
\texttt{conformalInference} R package \citep{tibshirani2016conformalR} with 
the default settings employed which specifies that the number of grid points 
is 100.  The HD prediction region is fit using using the \texttt{HDInterval} 
R package \citep{meredith2018HDIntervalR}.
Following the bin width asymptotics of \citet{lei2014distribution},  
the number of bins used to form the binned parametric and nonparametric conformal 
prediction regions is 2 when $n = 150$ and 3 when $n = 250, 500$.  
All simulations correspond to univariate regressions and the predictor variables 
were generated as $X \sim U(0,1)$.  
Figure \ref{Fig:examples} shows four example conformal regions 
for simulation setting A, B, and C with $n = 150$.

\begin{figure}[t!]
\centering
\includegraphics[width=0.75\textwidth]{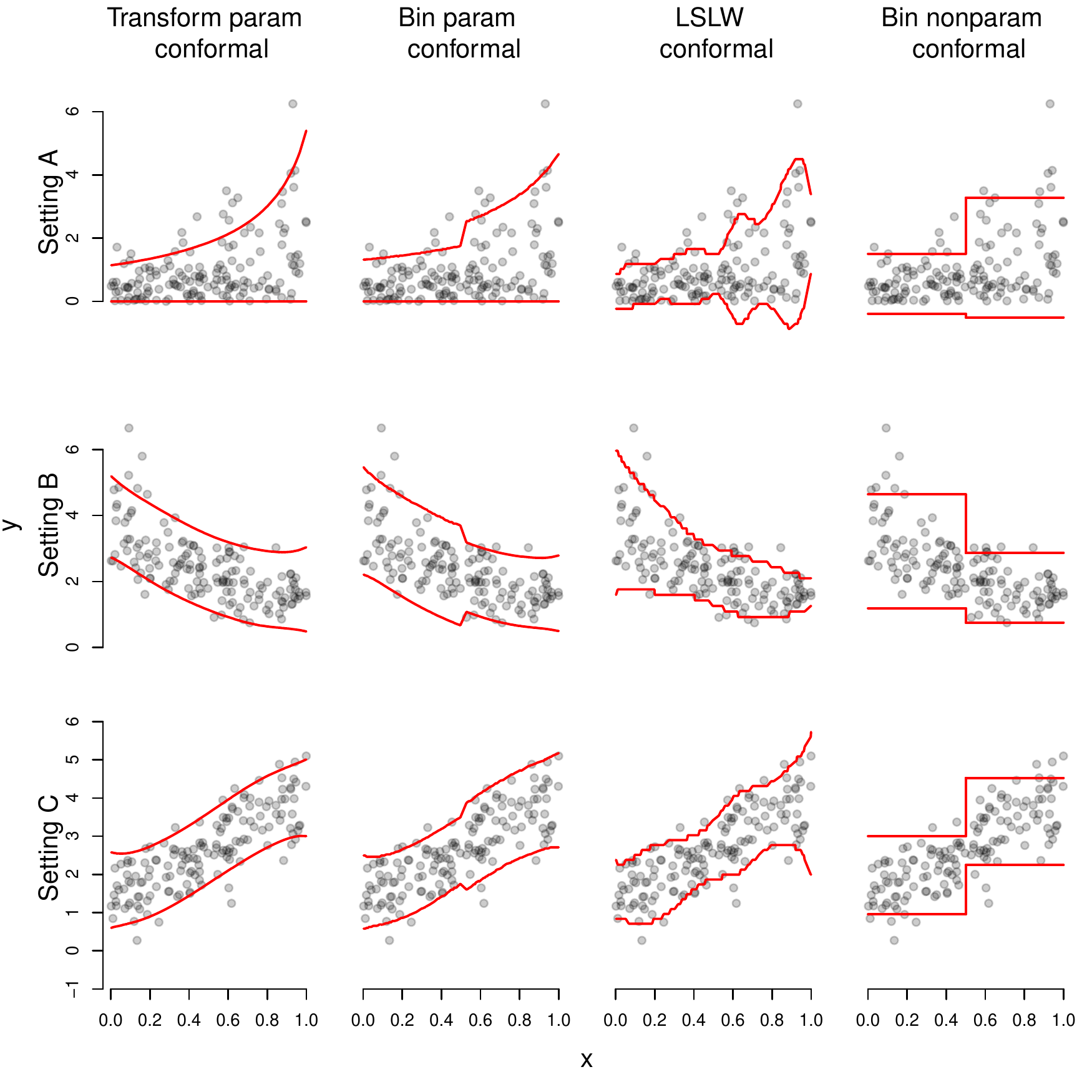}
\caption{Illustration of conformal prediction regions. 
  The top, middle, and bottom rows correspond to simulation setting A with shape 
  parameter equal to 1, simulation setting B with shape parameter equal to 10, 
  and simulation setting C respectively.  
  The first column depicts the transformation based parametric conformal 
  prediction region which is misspecified in row 2.  
  The second column depicts the binned parametric conformal prediction 
  region which is misspecified in row 2.  
  The third column depicts the locally weighted conformal prediction region.  
  The fourth column depicts the nonparametric conformal prediction region. }
\label{Fig:examples}
\end{figure}

Several diagnostic measures are used to compare conformal prediction regions.
These diagnostic measures compare prediction regions by their 
prediction error, volume, and coverage properties.
Our prediction error diagnostic metric will be an average of the squared 
distances of observations outside of the prediction region to the closest 
boundary of the prediction region, averaged over all 
observations. An observation that falls within the prediction region 
has an error of 0.  More formally this prediction error metric is 
$$ 
  \text{prediction error} 
    = n^{-1}\sum_{i=1}^n\indicator{Y_i \not\in \Copt(X_i)}
      \left(\min_{j=1,\ldots,m_i}\left\{\min\{|Y_i - a_{i,j}|, 
        |Y_i - b_{i,j}|\}\right\}\right)^2,
$$
where $a_{i,j}$ and $b_{i,j}$ are, respectively, the lower and upper 
boundaries of possible $j = 1, \ldots, m_i$ disjoint intervals forming 
the prediction region.  

The volume of each prediction region will be estimated by the average of the 
upper boundary minus the lower boundary across observed $\X$, written as
$$ 
  \text{area} = n^{-1}\sum_{i=1}^n\sum_{j=1}^{m_i}(b_{i,j} - a_{i,j}).
$$ 
To assess finite sample marginal validity we calculate the proportion of 
responses that fall within the prediction region.  To assess finite sample 
local validity with respect to binning we first bin the predictor data and 
then, for each bin, we calculate the proportion of responses that fall 
within the prediction region.  The same procedure is used to assess 
finite sample conditional validity, but we use a much finer binning regime 
than what was used to assess finite sample local validity.

\begin{figure} 
\begin{center}
\includegraphics[width=0.49\textwidth]{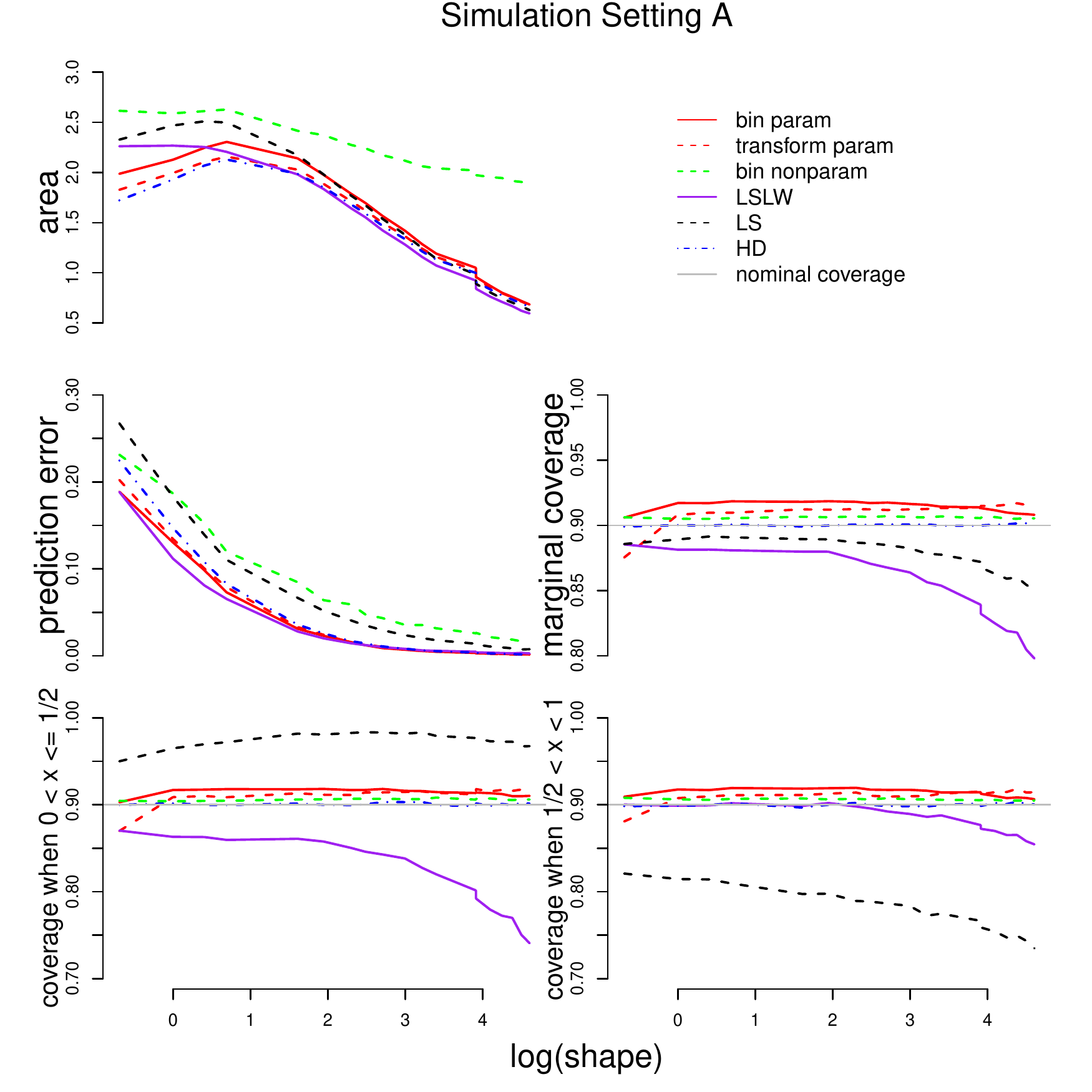}
\includegraphics[width=0.49\textwidth]{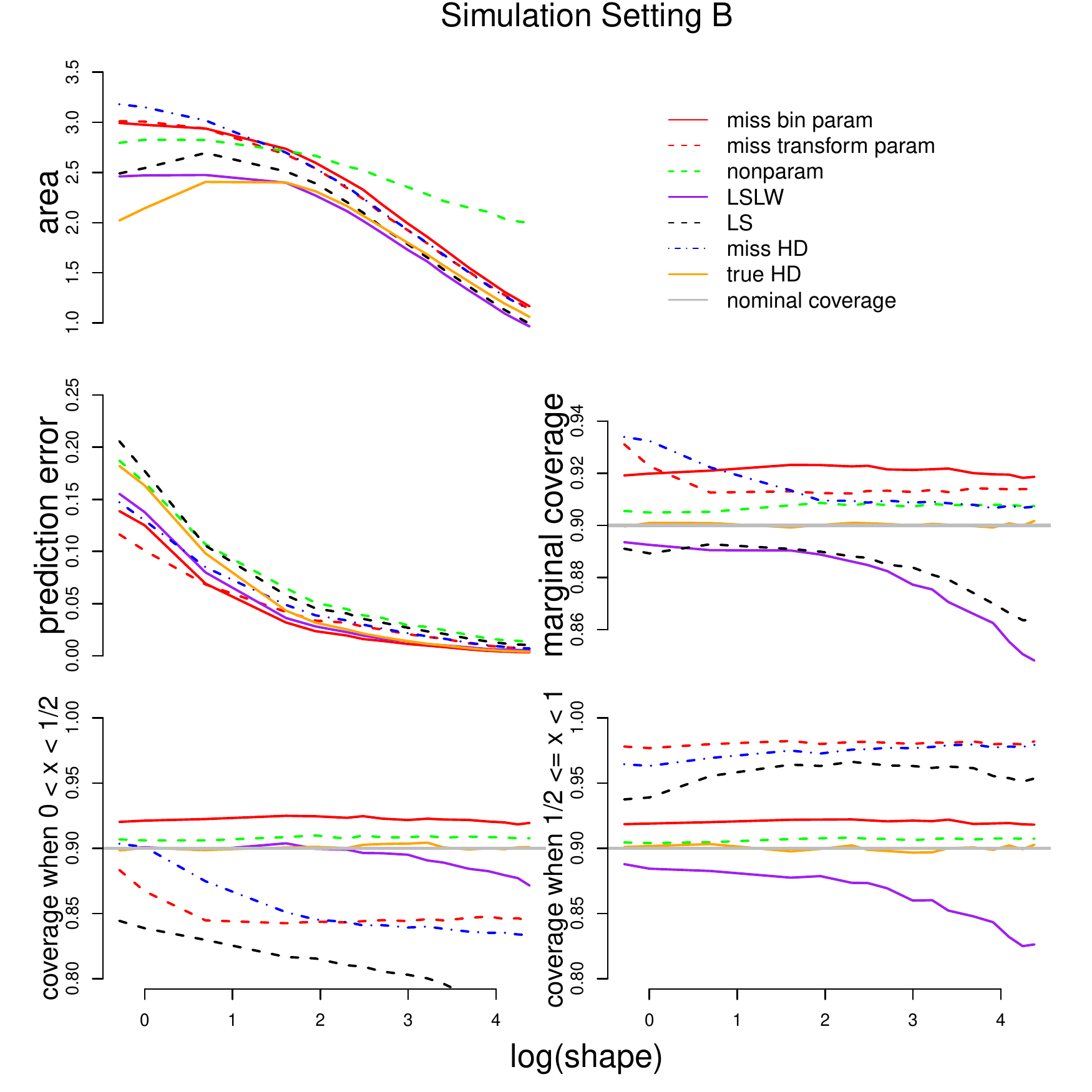}
\end{center}
\caption{
  Area, prediction error, and bin-wise coverage for 
    parametric, 
    nonparametric, 
    LS, 
    LSLW conformal prediction region, 
    and the highest density prediction region 
  for Gamma GLM regression with $n=150$ and $\alpha=0.1$.   Simulation setting A is shown at 
  left, and setting B at right. 
  The average of 250 samples at each shape parameter value in 
  these simulation settings form the lines that are depicted in both plots of 
  this figure.}
\label{Fig:diagnostics}
\end{figure}

In our simulations, we find that the parametric conformal 
prediction regions perform well even when the model is moderately 
misspecified. 
By construction, the binned parametric conformal prediction 
region, along with the binned nonparametric conformal prediction region, 
maintains finite sample local validity with respect to binning as seen in the 
bottom row of both plots in 
Figure~\ref{Fig:diagnostics} and Table~\ref{Tab:regression-results}.  
However, these prediction regions have different shapes conditional on $x$, 
as seen in Figure~\ref{Fig:examples}, and give different prediction errors as 
seen in the top row of both plots in Figure~\ref{Fig:diagnostics} and 
Table~\ref{Tab:regression-results}.  
Both parametric conformal prediction regions adapt naturally to the data when 
the model is correctly specified or when modest deviations from the specified 
model are present.  However, large deviations from model misspecification are 
not handled well as seen in the top row of the right hand side of
Figure~\ref{Fig:diagnostics}, although the transformation and binned parametric 
conformal prediction region give nominal marginal and local coverage 
respectively as seen in the middle row in the right hand side of 
Figure~\ref{Fig:diagnostics}.  
Note that the default precision allowed the transformation based parametric 
conformal prediction region to exhibits undercoverage for small values of 
$x$, and note overcoverage for large values of $x$ in the bottom row of the 
right hand side of Figure~\ref{Fig:diagnostics}.
On the other hand, the nonparametric conformal prediction region does not 
adapt well to data obtained from a Gamma regression model, and the former is 
not data adaptive in our linear regression setting.

\begin{table}[t!]
\scriptsize
\begin{center}
\begin{tabular}{llcccccc}
  & & OLS trans & OLS bin   & bin nonparametric & LS        & LSLW      & HD     \\
  & & conformal & conformal & conformal         & conformal & conformal & region \\ 
  $n = 150$
    & marginal coverage & $0.911$ & $0.919$ & $0.911$ & $0.878$ & $0.883$ & $0.904$ \\
    & local coverage when $0 < x < 1/2$ & $0.912$ & $0.918$ & $0.911$ & $0.878$ & $0.886$ & $0.903$ \\
    & local coverage when $1/2 \leq x < 1$ & $0.909$ & $0.92$ & $0.911$ & $0.878$ & $0.881$ & $0.904$ \\
    & area & $1.88$ & $1.994$ & $2.408$ & $1.775$ & $1.782$ & $1.872$ \\
    & prediction error & $0.01$ & $0.007$ & $0.011$ & $0.012$ & $0.011$ & $0.01$ \\
  \hline
  $n = 250$  
    & marginal coverage & $0.908$ & $0.916$ & $0.909$ & $0.875$ & $0.878$ & $0.902$ \\
    & local coverage when $0 < x < 1/3$ & $0.9$ & $0.916$ & $0.909$ & $0.87$ & $0.882$ & $0.896$ \\
    & local coverage when $1/3 \leq x < 2/3$ & $0.914$ & $0.913$ & $0.909$ & $0.879$ & $0.876$ & $0.907$ \\
    & local coverage when $2/3 \leq x < 1$ & $0.909$ & $0.92$ & $0.909$ & $0.877$ & $0.876$ & $0.904$ \\
    & area & $1.877$ & $1.974$ & $2.139$ & $1.751$ & $1.754$ & $1.864$ \\
    & prediction error & $0.009$ & $0.007$ & $0.01$ & $0.013$ & $0.012$ & $0.01$ \\
  \hline
  $n = 500$
    & marginal coverage & $0.908$ & $0.907$ & $0.904$ & $0.875$ & $0.873$ & $0.902$ \\
    & local coverage when $0 < x < 1/3$ & $0.911$ & $0.908$ & $0.906$ & $0.878$ & $0.875$ & $0.905$ \\
    & local coverage when $1/3 \leq x < 2/3$ & $0.907$ & $0.907$ & $0.904$ & $0.874$ & $0.871$ & $0.902$ \\
    & local coverage when $2/3 \leq x < 1$ & $0.905$ & $0.907$ & $0.904$ & $0.874$ & $0.875$ & $0.899$ \\
    & area & $1.886$ & $1.904$ & $2.098$ & $1.735$ & $1.731$ & $1.86$ \\
    & prediction error & $0.009$ & $0.009$ & $0.011$ & $0.013$ & $0.013$ & $0.01$ 
\end{tabular}
\end{center}
\caption{Diagnostics for conformal prediction regions for linear regression 
  models with normal errors and constant variance.  Local and marginal 
  coverage properties, areas, and prediction errors are presented for the 
    transformation based and binned parametric conformal prediction regions,
    binned nonparametric conformal prediction region,
    LS conformal prediction region, 
    LSLW conformal prediction region, and 
    HD prediction region.} 
\label{Tab:regression-results}
\end{table}

The LS conformal prediction region obtains marginal validity 
\citep{lei2018distribution} but performs poorly when deviations about the 
estimated mean function are either not symmetric, not constant, or both.  
When heterogeneity is present, the LS conformal prediction region 
exhibits undercoverage in regions where variability about the mean function 
is large and over-coverage in regions where variability about the mean 
function is small.  Clear evidence of these features are seen in 
Figure~\ref{Fig:diagnostics} and \ref{Fig:examples}.
This conformal prediction region is very sensitive to model misspecification.  
The LSLW conformal prediction region also guarantees  
marginal validity \citep[Section 5.2]{lei2018distribution}, although our 
simulations reveal that the default software implementation struggles to 
guarantee this in practice.  That being said, the LSLW region appears to 
exhibit flexibility under mild model misspecification as evidenced in  
the second row of Figure \ref{Fig:examples} and our Supplementary Materials.
This region is far less sensitive to model misspecification than the LS 
conformal prediction region, and it performs well under modest model 
misspecification.  However, the LSLW conformal 
prediction region is not appropriate when deviations about an estimated mean 
function are obviously not symmetric, as evidenced in the top row of 
Figure~\ref{Fig:examples}. 
Results from additional simulations corresponding to settings A 
and B are provided in the Supplementary Materials.  The findings from these 
additional simulations are consistent with the conclusions of the simulations 
presented in this Section.

\section{Predicting the risk of diabetes}

\begin{figure}[t]
\begin{center}
\includegraphics[width=0.9\textwidth]{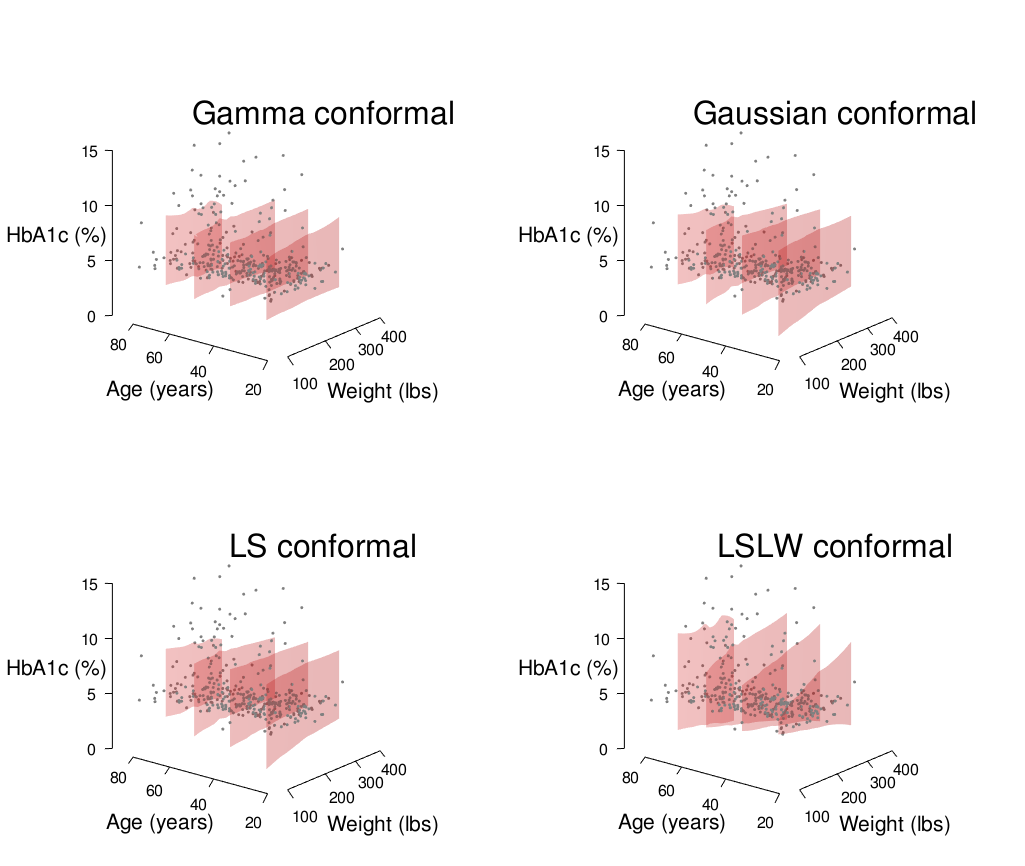}
\end{center}
\caption{ 
  Conformal prediction regions for glycosylated hemoglobin projected onto the 
  age and weight predictor axes.
  Upper and lower bounds of the conformal prediction region are loess smoothed for visual 
  appearance.
}
\label{Fig:3d-conformal-regions}
\end{figure}

Diabetes is a group of metabolic diseases associated with long-term 
damage, dysfunction, and failure of different organs, especially the eyes, 
kidneys, nerves, heart, and blood vessels \citep{american2010diagnosis}.
In 2017 approximately 5 million adult deaths worldwide 
were attributable to diabetes; global healthcare expenditures on people with 
diabetes are estimated USD 850 billion 
\citep{cho2018idf}.  Diabetes remains undiagnosed for an estimated 30\% of 
the people who have the disease \citep{heikes2008diabetes}.  
One way to address the problem of undiagnosed diabetes is to develop simple, 
inexpensive diagnostic tools that can identify people who are at high risk of 
pre-diabetes or diabetes using only readily-available clinical or demographic 
information \citep{heikes2008diabetes}.

We examine the influence of several variables on blood sugar, or glycosylated 
hemoglobin percentage (also known as HbA1c), an important risk factor for 
diabetes.  A glycosylated hemoglobin value of 6.5\% can be used as a cutoff 
for positive diagnosis of diabetes \citep{world2011use}.  
We predict an individual's glycosylated hemoglobin from their height, 
weight, age, and gender, all of which are easy to measure, inexpensive, and do 
not require any laboratory testing.  
The data in this analysis come from a population-based sample of 403 rural 
African-Americans in Virginia \citep{willems1997prevalence}, taken from the 
\texttt{faraway} R package \citep{faraway2016R}.  
We considered a gamma regression model that only includes linear terms for 
each covariate, a linear regression model with homoskedastic normal errors 
and the same linear terms for each covariate, and a linear regression model 
with homoskedastic normal errors that also included quadratic terms for 
each covariate.
Of these considered models, the gamma regression model fit the data 
best: it had the lowest AIC value and gives the best 
predictive predictive performance as measured by the sum of squares prediction error.  

Based on these covariates, conformal prediction regions provide finite sample 
valid prediction regions for glycosylated hemoglobin that may be useful for 
diagnosing diabetes in this study population.
Six conformal prediction regions are considered for predicting glycosylated 
hemoglobin percentage.  These conformal prediction regions are the binned and 
transformation parametric conformal prediction region with a Gamma model fit, 
the binned and transformation parametric conformal prediction region with a 
Gaussian model fit, the LS conformal prediction region, 
and the LSLW conformal prediction region.  
All conformal prediction regions correspond to models that only include 
linear terms for each of the covariates.  
The binned Gamma and Gaussian parametric conformal prediction regions were 
computed with binning across the binary gender factor variable, the predictor 
space is partitioned across genders. 
However, no additional binning structure within the levels of gender was 
employed.


\begin{table}[ht]
\centering
\begin{tabular}{rrrrrrr}
  \hline
 & Gamma trans & Gaussian trans & Gamma bin & Gaussian bin & LSLW      & LS \\
 & conformal   & conformal      & conformal & conformal    & conformal & conformal \\ 
  \hline
marginal coverage & 0.901 & 0.924 & 0.909 & 0.906 & 0.880 & 0.888 \\ 
  volume & 7.656 & 8.560 & 7.349 & 7.730 & 8.574 & 7.103 \\ 
  pred error & 0.653 & 0.425 & 0.931 & 0.849 & 0.803 & 1.012 \\ 
  avg.cond.coverage & 0.856 & 0.888 & 0.874 & 0.849 & 0.863 & 0.827 \\ 
   \hline
\end{tabular}
\label{Tab:diagnostics}
\caption{  Diagnostics for prediction regions.  Marginal coverage 
  ($\alpha = 0.10$), prediction region volume, prediction error, and average 
  conditional coverage make up the first four rows. } 
\end{table}

Diagnostics from the six conformal prediction regions are depicted in 
Table~\ref{Tab:diagnostics}.  The error tolerance for all prediction regions 
was set at $\alpha = 0.10$.  We see that parametric conformal prediction 
regions maintain their advertised finite sample marginal validity for the 
predictions of glycosylated hemoglobin.  
These prediction regions provide a balance between marginal coverage, size, 
prediction error, and average conditional coverage (the average of the 
coverage probabilities taken over small subregions of the predictor space).  
The transformation Gamma conformal prediction region balances these criteria 
particularly nicely. This prediction region is relatively small, has relatively 
small prediction error, and it gives near nominal desired coverage.  This 
finding is expected when the underlying estimated density used as the 
parametric conformity measure is a good approximation of the data generating 
model.

Plots of three dimensional conformal prediction region are displayed in 
Figure~\ref{Fig:3d-conformal-regions}.  These plots are projections of each 
conformal prediction regions to the age (in years), weight (in pounds), and 
glycosylated hemoglobin percentage three dimensional space.

\section{Discussion}

The finite sample validity properties of conformal prediction regions have been 
verified in broader methodological contexts, 
including support vector machines, ridge regression, nearest neighbor 
regression, neural networks, and decision decision trees 
\citep{vovk2005springer, gammerman2007hedging, papa2011neural, 
  vovk2012conditional}.  
While any conformity measure function will achieve finite sample validity, a
careful choice may make the returned conformal prediction regions more useful in 
applications \citep{papa2011neural}.
The developments of conformal prediction in the machine learning literature 
show how empirically successful prediction methods in machine learning can be 
hedged to give valid predictions in finite samples 
\citep{gammerman2007hedging}.  
The developments of conformal prediction in the statistics literature show 
that specification of the conformity measure to incorporate knowledge about 
the data generating process can lead to conformal prediction regions which 
are also asymptotically of minimal length 
\citep{lei2014distribution, dunn2018distribution}.  
Our parametric conformal prediction regions for parametric regression models 
falls within this line of research, which addresses efficiency concerns of 
conformal prediction methodology raised by \citet{cortes2019concepts}.
This line of research shows that when uncertainty about point predictions is 
considered, regression modeling provides smaller prediction regions when we 
require regions to guarantee valid finite sample coverage.  

Because GLMs are widely used by empiricists conducting regression analyses, 
parametric conformal prediction for GLM regression therefore may offer an 
appealing compromise for the applied researcher: when the GLM is correctly 
specified, conformal prediction regions are asymptotically minimal, 
finite sample local and marginal validity holds, and the rate of convergence 
is fast; when the GLM is incorrectly specified, asymptotic minimality is not 
guaranteed, but local finite sample validity still holds by construction.  
Researchers who currently use GLMs to compute prediction intervals for 
the mean regression function may be able to easily 
integrate conformal prediction into their data analysis workflow, since 
specification and fitting of the GLM is unchanged.   
A software package that accompanies this paper implements the parametric and  
nonparametric conformal prediction regions \citep{eck2018conformalR}.

The robustness properties of conformal prediction come at a 
substantial computational cost \citep{vovk2012conditional}.  
In our software implementation of the conformal methodology, 
two line searches are performed to determine the boundaries of the possibly 
disjoint conformal prediction region at every point at which a prediction 
region is desired.  The conformity scores must be recomputed with respect 
to augmented data at every iteration of these line searches.  
This involves refitting the parametric model to augmented data at every 
iteration of the line search to construct the parametric conformal prediction 
region.  Furthermore, when sample sizes are very large, conformal prediction 
may not offer much additional benefit beyond the parametric conditional prediction 
regions available in popular software packages for regression.  However, 
when sample sizes are moderate, conformal prediction may 
substantially outperform traditional methods in terms of finite sample 
marginal, local, and conditional coverage. 

The asymptotic optimality properties of parametric conformal prediction 
regions follow from our concentration inequality for maximum likelihood 
estimation stated in Theorem~\ref{concentration}.  
We conjecture that a similar concentration inequality may hold for random 
variables whose score functions do not have functional subexponential behavior
(see Appendix), and/or a non-existent MGF.  This concentration inequality 
will likely induce convergence rates that are slower than 
$\sqrt{\log(n)/n}$.  
\citet{dunn2018distribution} showed that finite sample 
validity holds in the presence of random effects.  We expect that asymptotic 
optimality properties for parametric conformal prediction regions can be 
extended to their settings and to the class of generalized linear mixed 
models.  

\textbf{Supplementary Materials}: Additional simulation results are 
available in the accompanying Supplementary Materials document 
(appended at the end of the main text).
An accompanying R package is available at 
\url{https://github.com/DEck13/conformal.glm} \citep{eck2018conformalR}. 
A technical report that includes the data and all of the code necessary to 
reproduce the findings, tables, and figures in this manuscript is available at 
\url{https://github.com/DEck13/conformal.glm/tree/master/techreport}.

\textbf{Acknowledgments}: We are grateful to 
Peter M. Aronow, 
Karl Oskar Ekvall, 
Jing Lei,
Aaron J. Molstad, 
Molly Offer-Westort, 
Cyrus Samii, 
Fredrik S\"avje, 
and
Larry Wasserman 
for helpful comments.  
This work was supported by NIH grant NICHD 1DP2HD091799-01.

\appendix

\section*{Appendix: Mathematical details}

\subsection*{Subexponential random variables}

The following definition and two lemmas are taken from 
\citet[Chapter 2]{wainwright2019high}.

\begin{defn}
A mean zero random variable $Y$ is said to be subexponential with parameters 
$(\tau^2, b)$ if 
$
  \E\left\{\exp(tY)\right\} \leq \exp\left(t^2\tau^2/2\right) 
$ 
for all $|t| \leq 1/b$. 
\label{subexpo}
\end{defn}

\begin{lem}
For a mean zero random variable $Y$, the following are equivalent:
(a) $Y$ is subexponential with parameters $(\tau^2, b)$;
(b) There is a positive number $c_0 > 0$ such that 
  $\E\left(e^{tY}\right) < \infty$ for all $|t| < c_0$; 
(c)
  $
    \Prob\left\{Y \geq \E(Y) + t\right\} 
      \leq \max\left(e^{-\frac{t^2}{2\tau^2}}, e^{-\frac{t}{2b}}\right). 
  $
\label{equiv}
\end{lem}

\begin{lem}
Let $Y_i$ be independent mean zero subexponential random variables with 
parameters $(\tau_i^2, b_i)$. Then $\sum_{i=1}^n Y_i$ is subexponential with 
parameters $\left(\sum_{i=1}^n\tau_i^2, b_{\star}\right)$ where 
$b_{\star} = \max_{i}(b_i)$ and 
\begin{equation}
  \Prob\left\{ \left|\frac{1}{n}\sum_{i=1}^n Y_i\right| \geq t \right\}
    \leq 2\exp\left\{-\min\left(
        \frac{nt^2}{2n^{-1}\sum_{i=1}^n\tau_i^2}, 
        \frac{nt}{2b_{\star}}
      \right)\right\}.
\label{sumsubexpo}
\end{equation}
\label{sums}
\end{lem}

Let $l_\psi(Y|x) = \log p_\psi(Y|x)$ be the log likelihood for the random 
variable $Y|X=x$ and let $\nabla_{\psi}l_\psi(Y|x)$ be the corresponding score 
function.  
We will now motivate the construction of the parameter space $\aleph_\X$.  
This parameter space requires a notion of functional subexponential that we 
will now define and examine.

\begin{defn}[functional subexponential]
Let $Y$ be a random variable and suppose that $f:\R\to\R$ is a function such 
that $E\{f(Y)\} = 0$.  We say that $f$ is \emph{functional subexponential} 
with respect to $Y$ with parameters $(\tau^2, b)$ if 
$
  \E\left\{\exp(tf(Y))\right\} \leq \exp\left(t^2\tau^2/2\right) 
$ 
for all $|t| \leq 1/b$.
\end{defn}

\begin{defn}[multivariate functional subexponential]
Let $Y$ be a random variable and suppose that $f:\R\to\R^r$ is a function 
such that $E\{f(Y)\} = 0$.  We say that $f$ is 
\emph{multivariate functional subexponential} with respect to $Y$ with parameters 
$(\tau^2, b)$ if 
$
  \E\left\{\exp(t'f(Y))\right\} \leq \exp\left(|t|^2\tau^2/2\right) 
$ 
for all $|t| \leq 1/b$.
\end{defn}

\begin{lem}
Let $Y$ be a random variable and suppose that $f:\R\to\R^r$ is a function 
such that $E\{f(Y)\} = 0$.  Then $f(Y)$ is multivariate functional subexponential 
with respect to $Y$ if and only the components of $f(Y)$ are functional 
subexponential with respect to $Y$.
\label{lem:fsubexpo}
\end{lem}

\begin{proof}
First suppose that $f(Y)$ is multivariate functional subexponential 
with respect to $Y$.  Let $s_j$ be the $r$ dimensional zeros vector 
with an $s$ in the $j$th component.  Then 
$
  \E\left\{\exp(sf_j(Y))\right\} \leq \exp\left(s^2\tau^2/2\right) 
$ 
for all $|s| \leq 1/b$ where $f_j(Y)$ is the $j$th component of $f(Y)$.  
It follows that the components of $f(Y)$ are functional subexponential 
with respect to $Y$.

Now suppose that the components of $f(Y)$ are functional subexponential 
with respect to $Y$ with parameters $(\tau^2_j, b_j)$, $j = 1,...,r$.  
Let $b = \max(b_j)$ and $\tau = \max(\tau_j)$.  Then 
$
  \E\left\{\exp(sf_j(Y))\right\} \leq \exp\left(s^2\tau^2/2\right) 
$ 
holds for all $|s| \leq 1/b$.  Now pick a vector $t \in \R^r$ such that 
$|t| < 1/b$.  Then by multiple applications of Cauchy Scharwz and the 
supposition that the components of $f(Y)$ are functional subexponential 
we have that 
\begin{align*}
  &\E\left\{\exp(t'f(Y))\right\} 
    = \E\left\{\prod_{j=1}^r\exp(t_jf_j(Y))\right\}
    \leq \prod_{j=1}^r \E\left\{\exp(2^jt_jf_j(Y))\right\}^{1/2^j} \\
  &\qquad \leq \prod_{j=1}^r \E\left\{\exp(2^{2j-1}t_j^2\tau^2)\right\}^{1/2^j}
    \leq \prod_{j=1}^r \E\left\{\exp(2^jt_j^2\tau^2/2)\right\} 
    \leq \exp\left(\frac{|t|^2 2^r\tau^2}{2}\right),
\end{align*}
for all $|t| < 1/b$. Set $\tau' = \tau 2^{r/2}$, then $f(Y)$ is multivariate 
subexponential with respect to $Y$ with parameters $(\tau', b)$.  The 
conclusion follows.
\end{proof}

The following Lemma is the key piece that ties $\aleph_\X$ together with the 
concentration inequality given in Theorem~\ref{concentration}.

\begin{lem} 
Let $Y|X=x$ be a random variable and suppose that the moment generating 
function $M_{Y,x,\psi}(t)$ exists for all $t$ in a neighborhood of $0$.  
Let $p_\psi(y|x)$ be the corresponding probability density function with 
parameter vector $\psi \in \R^r$.  Assume that $p_\psi(y|x)$ is differentiable 
in both $y$ and $\psi$.  Then the score function $\nabla_\psi \log p_\psi(Y|x)$ 
is functional subexponential with respect to $Y|X=x$ provided that 
$\nabla_\psi \log p_\psi(y|x) = O(y)$.
\label{lem:mgfsubexpo}
\end{lem}

\begin{proof}
Let $\nabla_{\psi_j} \log p_\psi(y|x)$ be the $j$th component of 
$\nabla_{\psi} \log p_\psi(y|x)$.  Then, 
\begin{align*}
  &\E\{\exp(t_j\nabla_{\psi_j} \log p_\psi(Y|x))\} 
    = \int \exp(t_j\nabla_{\psi_j} \log p_\psi(y))p_\psi(y|x)dy \\
  &\qquad= \int \exp\left\{t_jy\left(
      \frac{\nabla_{\psi_j} \log p_\psi(y|x)}{y}
    \right)\right\}p_\psi(y|x)dy < \infty,
\end{align*}
for all $t_j$ in a neighborhood about $0$, and all $j = 1$, $\ldots$, $r$.  
The same argument as that in Appendix B of \citet{wainwright2019high} 
implies that the score function $\nabla_{\psi_j} \log p_\psi(Y)$ is 
functional subexponential with respect to $Y|X=x$.  Our conclusion follows 
from Lemma~\ref{lem:fsubexpo}.
\end{proof}

\subsection*{Conentration results}

We can now prove Theorem~\ref{concentration} from Lemma~\ref{lem:mgfsubexpo} 
and our developed notion of functional subexponential random variables.  

\begin{proof}[Proof of Theorem~\ref{concentration}]
From the mean-value theorem, there exists some $\psihat_1$ such that  
$\psihat_{1,j} \in (\psi_j \wedge \psihat_j$, 
$\psi_j \vee \psihat_j)$, $j = 1$, $\ldots$, $r$, which satisfies 
$$
  -\nabla_{\psi}l_n(\psi) = \nabla_{\psi}^2l_n(\psihat_1)(\psihat - \psi).
$$
Rearranging the above yields 
$$
  \sqrt{n}(\psihat - \psi) 
    = -\sqrt{n}\left\{\nabla_{\psi}^2l_n(\psihat_1)\right\}^{-1}
      \nabla_{\psi}l_n(\psi)
$$
where the inverse exists almost surely when $n > r$. We see that 
\begin{align*}
  \left|\sqrt{n}(\psihat - \psi)\right| &= 
    \left|\sqrt{n}\left\{\nabla_{\psi}^2l_n(\psihat_1)\right\}^{-1}
      \nabla_{\psi}l_n(\psi)\right| \\
  &\leq \left|\sqrt{n}\left\{\nabla_{\psi}^2l_n(\psi)\right\}^{-1}
      \nabla_{\psi}l_n(\psi)\right| + |a_n| \\
  &\leq \sqrt{n}\left\|\left\{\nabla_{\psi}^2l_n(\psi)\right\}^{-1}\right\|_1 
    \left|\nabla_{\psi}l_n(\psi)\right| + |a_n|,    
\end{align*}
where $\|\cdot\|_1$ is the induced $k$-norm for a matrix with $k = 1$ and 
$$
  |a_n| = \left|\sqrt{n}\left[
    \left\{\nabla_{\psi}^2l_n(\psihat_1)\right\}^{-1}
      - \left\{\nabla_{\psi}^2l_n(\psi)\right\}^{-1}
  \right]\nabla_{\psi}l_n(\psi)\right|. 
$$
Choose some $0 < a_\lambda < A_{\lambda}$ and let 
$B_\lambda = A_\lambda - a_\lambda$.  Then for $n$ sufficiently large, 
\begin{align*}
  &\Prob\left\{ 
      \sqrt{n}\left|\psihat - \psi\right| \geq A_{\lambda}\sqrt{\log(n)} 
    \right\}
    \leq \Prob\left\{ 
      \sqrt{n}\left\|\left\{\nabla_{\psi}^2l_n(\psi)\right\}^{-1}\right\|_1 
        \left|\nabla_{\psi}l_n(\psi)\right| + |a_n| 
      \geq A_{\lambda}\sqrt{\log(n)} 
    \right\} \\
  &\qquad \leq \Prob\left( 
        \left|\nabla_{\psi}l_n(\psi)\right| 
      \geq \sqrt{\frac{\log(n)}{n}}\frac{B_\lambda}
        {
          \left\|\left\{\nabla_{\psi}^2l_n(\psi)\right\}^{-1}\right\|_1
        } 
    \right) \\
  &\qquad= \Prob\left( 
        \left|n^{-1}\sum_{i=1}^n\nabla_{\psi}l_i(\psi)\right| 
      \geq \sqrt{\frac{\log(n)}{n}}\frac{B_\lambda}
        {
          \left\|n^{-1}\left\{\nabla_{\psi}^2l_n(\psi)\right\}^{-1}\right\|_1 
        } 
    \right),
\end{align*}
where the second inequality follows from the strong law of large numbers with 
respect to $a_n$ and $l_i(\psi)$ is the log likelihood for each observation. 
We can now choose some 
$
  D > \left\|\E\left[
        \left\{\nabla_{\psi}^2l_{n=1}(\psi)\right\}^{-1}
    \right]\right\|_1
$
such that, for $n$ sufficiently large, we have
\begin{align*}
  &\Prob\left( 
        \left|n^{-1}\sum_{i=1}^n\nabla_{\psi}l_i(\psi)\right| 
      \geq \sqrt{\frac{\log(n)}{n}}\frac{B_\lambda}
        {
          \left\|n^{-1}\left\{\nabla_{\psi}^2l_n(\psi)\right\}^{-1}\right\|_1 
        } 
    \right) \\
  &\qquad\leq \Prob\left( 
        \left|n^{-1}\sum_{i=1}^n\nabla_{\psi}l_i(\psi)\right|  
      \geq \sqrt{\frac{\log(n)}{n}}\frac{B_\lambda}{D}
    \right) \\
  &\qquad= \Prob\left( 
        \sum_{j=1}^r
          \left|n^{-1}\sum_{i=1}^n[\nabla_{\psi}l_i(\psi)]_j\right|  
      \geq \sqrt{\frac{\log(n)}{n}}\frac{B_\lambda}{D}
    \right) \\ 
  &\qquad\leq \sum_{j=1}^r\Prob\left( 
        \left|n^{-1}\sum_{i=1}^n[\nabla_{\psi}l_i(\psi)]_j\right| 
      \geq \sqrt{\frac{\log(n)}{n}}\frac{B_\lambda}{Dr}
    \right), 
\end{align*}
where the first inequality follows from the strong law of large numbers 
with respect to
$
  \left\|\E\left[
        \left\{\nabla_{\psi}^2l_{n=1}(\psi)\right\}^{-1}
    \right]\right\|_1,
$
the term 
$
  \left|n^{-1}\sum_{i=1}^n[\nabla_{\psi}l_i(\psi)]_j\right|
$
is the $j$th component of $\nabla_{\psi}l_i(\psi)$, 
and the second inequality follows from sub additivity of probability and the 
fact that a sum of elements is greater than or equal to a number if at least 
one term in the sum is greater than or equal to that number divided by the 
number of elements in the sum.  
Then Lemmas~\ref{sums}, \ref{lem:fsubexpo}, and \ref{lem:mgfsubexpo} give 
\begin{align*}
  \sum_{j=1}^r\Prob\left( 
        \left|n^{-1}\sum_{i=1}^n[\nabla_{\psi}l_i(\psi)]_j\right| 
      \geq \sqrt{\frac{\log(n)}{n}}\frac{B_\lambda}{Dr}
    \right) 
  \leq 2\sum_{j=1}^r n^{-\frac{B_\lambda}{2D\tau_j^2r}}. 
\end{align*}
We can pick $A_{\lambda}$ such that $B_{\lambda}$ can be chosen to satisfy 
$$
  2\sum_{j=1}^r n^{-\frac{B_\lambda}{2D\tau_j^2r}} = O(n^{-\lambda}).
$$
Our conclusion follows. 
\end{proof}

The following Corollary extends the concentration inequality for MLEs 
corresponding to conditional densities $p_\psi(\cdot|x)$ with parameter 
space $\aleph_\X$ in Theorem~\ref{concentration} to the conformal prediction 
framework. 

\begin{cor}
Let $Y|X=x$ be a random variable with conditional density 
$p_\psi(\cdot|x)$ and parameter space $\aleph_\X$.  Assume that 
$p_\psi(\cdot|x)$ is twice differentiable in $\psi$.  
Let $(X_1,Y_1)$, $\ldots$, $(X_n,Y_n)$ be independent and identically 
distributed copies of $(X,Y)$.  Let $\psi \in \R^r$.  Let $\psihat$ be the 
MLE of $\psi$.  Augment the sample data with a new point $(x,y)$, and let 
$\psihat^{(x,y)}$ be the MLE of $\psi$ with respect to the augmented data.  
Then for any $\lambda > 0$, there exists a numerical constant $A_{\lambda}'$, 
such that  
\begin{equation}
  \Prob\left\{ 
    \sqrt{n}\left|\psihat^{(x,y)} - \psi\right| \geq A_{\lambda}'\sqrt{\log(n)} 
  \right\}
  = O\left(n^{-\lambda}\right).
\label{precise-cor}
\end{equation}
\label{concentration-cor}
\end{cor}

\begin{proof}
Let $\psihat$ be the maximum likelihood estimator of $\psi$ under the original 
data.  First note that 
\begin{equation} \label{deriv-cor}
\begin{split}
  \Prob\left\{ 
    \sqrt{n}\left|\psihat^{(x,y)} - \psi\right| \geq A_{\lambda}'\sqrt{\log(n)} 
  \right\}
    &\leq  \Prob\left\{ 
      \sqrt{n}\left|\psihat^{(x,y)} - \psihat\right| + 
      \sqrt{n}\left|\psihat - \psi\right| 
      \geq A_{\lambda}'\sqrt{\log(n)} 
    \right\} \\
   &= \Prob\left\{  
      \sqrt{n}\left|\psihat - \psi\right| 
      \geq A_{\lambda}'\sqrt{\log(n)} - 
      \sqrt{n}\left|\psihat^{(x,y)} - \psihat\right|
    \right\}.
\end{split}
\end{equation}
We will show that the $\sqrt{n}\left|\psihat^{(x,y)} - \psihat\right|$ term in 
\eqref{deriv-cor} vanishes quickly enough to yield \eqref{precise-cor}.  
From the proof Theorem~\ref{concentration} we have that 
$$
  \sqrt{n}(\psihat - \psi) 
    = -\sqrt{n}\left\{\nabla_{\psi}^2l_n(\psihat_1)\right\}^{-1}
      \nabla_{\psi}l_n(\psi).
$$
where $l_n(\cdot)$ is the log likelihood of the original data and $\psihat_1$ 
is such that $\psihat_{1,j} \in (\psi_j \wedge \psihat_j$, 
$\psi_j \vee \psihat_j)$, $j = 1$, $\ldots$, $r$. 
Let $l_n^{(x,y)}(\psi)$ be the likelihood of the augmented data.  A similar 
calculation to that in the proof Theorem~\ref{concentration} yields 
$$
  \sqrt{n}(\psihat^{(x,y)} - \psi) 
    = -\sqrt{n}\left\{\nabla_{\psi}^2l_n^{(x,y)}(\psihat^{(x,y)}_1)\right\}^{-1}
      \nabla_{\psi}l_n^{(x,y)}(\psi)
$$
with $\psihat^{(x,y)}_1$ defined similarly to $\psihat_1$.  
We have that
$$
  \nabla_{\psi}l_n^{(x,y)}(\psi) = \nabla_{\psi}l_n(\psi) 
    + \nabla_\psi \log p_\psi(y|x).  
$$
These derivations yield 
\begin{align*}
  &\sqrt{n}\left|\psihat^{(x,y)} - \psihat\right| = 
    \left|(\psihat^{(x,y)} - \psi) - (\psihat - \psi)\right| \\
  &\qquad= |\sqrt{n}\left\{\nabla_{\psi}^2l_n^{(x,y)}(\psihat^{(x,y)}_1)\right\}^{-1}
      \nabla_{\psi}l_n^{(x,y)}(\psi) - 
    \sqrt{n}\left\{\nabla_{\psi}^2l_n(\psihat_1)\right\}^{-1}
      \nabla_{\psi}l_n(\psi)| \\
  &\qquad= |\sqrt{n}\left\{\nabla_{\psi}^2l_n^{(x,y)}(\psihat^{(x,y)}_1)\right\}^{-1}
      \left[  
        \nabla_{\psi}l_n(\psi) + 
        \nabla_\psi \log p_\psi(y|x) 
      \right] \\
    &\qquad\qquad- 
      \sqrt{n}\left\{\nabla_{\psi}^2l_n(\psihat_1)\right\}^{-1}
        \nabla_{\psi}l_n(\psi)| \\
  &\qquad= |\left[
      \left\{\nabla_{\psi}^2l_n^{(x,y)}(\psihat^{(x,y)}_1)\right\}^{-1} -
      \left\{\nabla_{\psi}^2l_n(\psihat_1)\right\}^{-1}
    \right]\sqrt{n}\nabla_{\psi}l_n(\psi) \\
      &\qquad\qquad+ 
        \sqrt{n}\left\{\nabla_{\psi}^2l_n^{(x,y)}(\psihat^{(x,y)}_1)\right\}^{-1}
        \nabla_\psi \log p_\psi(y|x)| \\
  &\qquad= |\left[
      \left\{\frac{1}{n}\nabla_{\psi}^2l_n^{(x,y)}(\psihat^{(y)}_1)\right\}^{-1} -
      \left\{\frac{1}{n}\nabla_{\psi}^2l_n(\psihat_1)\right\}^{-1}
    \right]\frac{1}{\sqrt{n}}\nabla_{\psi}l_n(\psi) \\
      &\qquad\qquad+ 
        \frac{1}{\sqrt{n}}\left\{\frac{1}{n}
          \nabla_{\psi}^2l_n^{(x,y)}(\psihat^{(x,y)}_1)\right\}^{-1}
        \nabla_\psi \log p_\psi(y|x)|.         
\end{align*}
By the strong law of large numbers and a similar argument to the proof of 
Theorem~\ref{concentration}, we can pick $A_{\lambda}'$ such that, 
for $n$ sufficiently large, we have that
\begin{align*}
  \Prob\left\{  
    \sqrt{n}\left|\psihat - \psi\right| 
    \geq A_{\lambda}'\sqrt{\log(n)} - 
    \sqrt{n}\left|\psihat^{(x,y)} - \psihat\right|
  \right\}
  \leq
  \Prob\left\{  
    \sqrt{n}\left|\psihat - \psi\right| 
    \geq A_{\lambda}\sqrt{\log(n)} 
  \right\}.  
\end{align*}
where $A_{\lambda}$ is defined in Theorem~\ref{concentration}.  
Our conclusion follows from Theorem~\ref{concentration}.
\end{proof}

\subsection*{Proof of Theorem~\ref{thm:correctmodel}}

The concentration inequalities in Theorem~\ref{concentration} and 
Corollary~\ref{concentration-cor} allow us to prove 
Theorem~\ref{thm:correctmodel}.

\begin{lem}
Let $Y|X=x$ be a random variable with conditional density 
$p_\psi(\cdot|x)$ and parameter space $\aleph_\X$.  Assume that 
$p_\psi(\cdot|x)$ is twice differentiable in $\psi$ and satisfies 
Assumption 2.  Let $(X_1,Y_1)$, $\ldots$, $(X_n,Y_n)$ be independent and 
identically distributed copies of $(X,Y)$.  Let $\psi \in \R^r$.  
Let $\psihat$ be the MLE of $\psi$.  Augment the sample data with a new 
point $(x,y)$, and let $\psihat^{(x,y)}$ be the MLE of $\psi$ with respect 
to the augmented data.  
Let $\hat{p}^{(x,y)}_\psi(\cdot|x)$ be the conditional density with 
$\psihat^{(x,y)}$ plugged in for $\psi$.  Given $\lambda > 0$, there is a 
numerical constant $\xi_{\lambda}$ such that 
$$
  \Prob\left\{\sup_{x \in \X}
    \|\hat{p}^{(x,y)}_\psi(\cdot|x) - p_\psi(\cdot|x)\|_{\infty} 
      \geq \xi_{\lambda}r_n\right\} 
  = O\left(n^{-\lambda}\right).
$$  
\label{lem:dens}
\end{lem}

\begin{proof}
From the Lipschitz property of Assumption 2 and the boundedness of $\X$, 
we can pick an $M > 0$ such that 
$$
  \sup_{x\in\X} \|\hat{p}^{(x,y)}_\psi(\cdot|x) - p_\psi(\cdot|x)\|_{\infty} 
    \leq M\|\psihat^{(x,y)} - \psi\|.
$$
Therefore, 
$$
  \Prob\left\{\sup_{x \in \X}
    \|\hat{p}^{(x,y)}_\psi(\cdot|x) - p_\psi(\cdot|x)\|_{\infty} 
      \geq \xi_{\lambda}r_n\right\}
  \leq
  \Prob\left\{M\|\psihat^{(x,y)} - \psi\| \geq \xi_{\lambda}r_n\right\}.
$$
Choose $\xi_{\lambda}$ so that $A_\lambda' = \xi_\lambda/M$ satisfies the rate 
of \eqref{precise-cor} in Corollary~\ref{concentration-cor}.  Our conclusion 
then follows from Corollary~\ref{concentration-cor}.
\end{proof}

From the Lipschitz property in Assumption 2, we can pick $M' > 0$ such that 
$$
  \sup_{x\in A_k} \|\hat{p}^{(x,y)}_\psi(\cdot|x) - p_\psi(\cdot|x)\|_{\infty} 
    \leq M'\left(\|\psihat^{(x,y)} - \psi\| + \|\psihat - \psi\|\right).
$$
Let $a_k$ be the center point of the cube $A_k$. Again, from the Lipschitz 
property in Assumption 2, we can pick $M'' > 0$ such that 
$$
  \sup_{x\in A_k} \|\hat{p}_\psi(\cdot|x) - \hat{p}_\psi(\cdot|a_k)\|_{\infty}
    \leq z_nM''\|\psihat\|_{\infty}.
$$
Set $\lambda > 0$ and pick $A_{\lambda}$ as in Theorem~\ref{concentration} 
and $A_{\lambda}'$ as in Corollary~\ref{concentration-cor} and let
$$
  E_1 = \{\|\psihat - \psi\| \leq r_nM'A_{\lambda}, 
    \|\psihat^{(x,y)} - \psi\| \leq r_nM'A_{\lambda}'\},
$$
where $\Prob(E_1^c) = O(n^{-\lambda})$ by 
Theorem~\ref{concentration} and Corollary~\ref{concentration-cor}.  
On $E_1$ we have that
$$
  \sup_{x\in A_k} 
    \|\hat{p}_\psi^{(x,y)}(\cdot|x) - \hat{p}_\psi(\cdot|x)\|_{\infty} 
      \leq r_nM'(A_{\lambda} + A_{\lambda}'),
$$
and 
$$
  \sup_{x\in A_k} 
    \|\hat{p}_\psi(\cdot|x) - \hat{p}_\psi(\cdot|a_k)\|_{\infty}
      \leq z_nM''(\|\psi\| + r_nM'A_{\lambda}).
$$
For $n$ sufficiently large we can pick $M''' > 0$ such that 
$$
  z_nM''(\|\psi\| + r_nM'A_{\lambda}) \leq z_n M'''
$$
and this implies that 
$$
  \sup_{x\in A_k} 
    \|\hat{p}_\psi(\cdot|x) - \hat{p}_\psi(\cdot|a_k)\|_{\infty} 
      \leq z_n M'''
$$
Now let $\{(X_{i_1}, Y_{i_1}),\ldots,(X_{i_{n_k}}, Y_{i_{n_k}})\}$ be the data 
points that belong to $A_k$ where the indices $(i_1,\ldots,i_{n_k})$ are 
conditioned on.
Let the data point $(Y_{(k,\alpha)}, X_{(k,\alpha)})$ be such 
that $\hat{p}_\psi(Y_{(k,\alpha)}|a_k)$ is the $\floor{n_k\alpha}$ largest 
value among all $\hat{p}_\psi(Y_{i_j}|a_k)$, $1 \leq j \leq n_k$ and define 
the sandwiching sets,
\begin{equation} \label{sandwichsets}
\begin{split}
  &\widehat{L}^-_k = \widehat{L}_k\left(
    \hat{p}_\psi\left(Y_{(k,\alpha)}|a_k\right) + 2r_nM'A_{\lambda}'' + 2z_nM''' 
  \right), \\
  &\widehat{L}^+_k = \widehat{L}_k\left(
    \hat{p}_\psi\left(Y_{(k,\alpha)}|a_k\right) - 2r_nM'A_{\lambda}'' - 2z_nM'''
  \right),
\end{split}
\end{equation}
where the level set $\widehat{L}_k(t) = \{y: \hat{p}_\psi(y|a_k) \geq t\}$ 
and $A_\lambda'' = A_\lambda + A_\lambda'$.  The bin sample size has to 
satisfy $n_k \to \infty$ for the sandwiching sets to be of use, this occurs 
when $z_n$ is as specified.  With this construction, we have
\begin{align*}
  \widehat{L}^-_k &\subseteq \left\{y : \frac{1}{n_k + 1}\sum_{i=1}^{n + 1}
    \indicator{X_i \in A_k}
      \indicator{\hat{p}_\psi(Y_i|a_k) + 2r_nMA_{\lambda}'' + 2z_nM'''
        \leq \hat{p}_\psi(y|a_k)} 
    \geq \widetilde{\alpha}_k\right\} \\
  &\subseteq \Coptglm(x)  \\
  &\subseteq\left\{y : \frac{1}{n_k + 1}\sum_{i=1}^{n + 1}
    \indicator{X_i \in A_k}
      \indicator{\hat{p}_\psi(Y_i|a_k) - 2r_nMA_{\lambda}'' - 2z_nM'''
        \leq \hat{p}_\psi(y|a_k)} 
    \geq \widetilde{\alpha}_k\right\} \\
  &\subseteq \widehat{L}^+_k.       
\end{align*}
Therefore
$$
  \Prob\left\{
    \widehat{L}^-_k \subseteq \Coptglm(x) \subseteq \widehat{L}^+_k
  \right\} = 1 - O(n^{-\lambda}).
$$
We summarize this result in the following Lemma.

\begin{lem}
Let $Y|X=x$ be a random variable with conditional density 
$p_\psi(\cdot|x)$ and parameter space $\aleph_\X$.  Assume that 
$p_\psi(\cdot|x)$ is twice differentiable in $\psi$ and satisfies 
Assumption 2.  Let $(X_1,Y_1)$, $\ldots$, $(X_n,Y_n)$ be independent and 
identically distributed copies of $(X,Y)$.  Let $\psi \in \R^r$.  
Let $\psihat$ be the MLE of $\psi$.  Augment the sample data with a new 
point $(x,y)$, and let $\psihat^{(x,y)}$ be the MLE of $\psi$ with respect 
to the augmented data.  Suppose that Assumption 1 holds.  
Let $0 < \alpha < 1$.   
Let the sets $\widehat{L}^-_k$ and $\widehat{L}^+_k$ be defined as in 
\eqref{sandwichsets}. Then, 
$$
  \Prob\left\{
    \widehat{L}^-_k \subseteq \Coptglm(x) \subseteq \widehat{L}^+_k
  \right\} = 1 - O(n^{-\lambda}).
$$
\label{lem:sandwich-subexpo}
\end{lem}

\begin{proof}
The details of this proof are given in the above paragraph.
\end{proof}

We now have enough technical tools to finish the proof of 
Theorem~\ref{thm:correctmodel}.

\begin{proof}[Proof of Theorem~\ref{thm:correctmodel}]
Define $L^l_k(t) = \{y:\hat{p}_\psi(y|a_k) \leq t\}$ and 
$\thatalpha_k = \hat{p}_\psi\left(Y_{(k,\alpha)}|a_k\right)$ and
$\alpha_k = \floor{\alpha n_k}/n_k$.  We have that 
\begin{align*}
  \widehat{P}\left\{\widehat{L}_k^l(\thatalpha_k)|A_k\right\}
    &= \frac{1}{n_k}\sum_{i=1}^n \indicator{X_i \in A_k}
      \indicator{Y_i \in \widehat{L}_k^l(\thatalpha_k)} \\
  &= \frac{1}{n_k}\sum_{i=1}^n \indicator{X_i \in A_k}
    \indicator{\hat{p}_\psi(Y_i|a_k) 
      \leq \hat{p}_\psi\left(Y_{(k,\alpha)}|a_k\right)}, 
\end{align*}
and this implies that 
$
  \thatalpha_k = \inf[t\geq0: 
    \widehat{P}\left\{\widehat{L}_k^l(t)|A_k\right\} \geq \alpha_k]. 
$
Let 
$$
  R_{n,k} = \sup_{x\in A_k} 
    \|\hat{p}_\psi(\cdot|x) - \hat{p}_\psi(\cdot|a_k)\|_{\infty},
$$ 
and let $V_n(x)$ be as in \citet[Lemma 6]{lei2014distribution} with 
$z_n$ in place of $w_n$ and $p_\psi(\cdot|a_k)$ replacing their $p(\cdot|A_k)$ 
in the setup and proof of \citet[Lemma 6]{lei2014distribution}.  
Define the event 
$$
  E_2 = \{\|\psihat - \psi\| \leq r_nM'A_{\lambda}, 
    \|\psihat^{(x,y)} - \psi\| \leq r_nM'A_{\lambda}',
    R_{n,k} \leq z_n M''',
    V_n(x) \leq z_n\xi_{\lambda}\},
$$
where $z_n$ and $\xi_{\lambda}$ are respectively $r_n$ and $\xi_{2,\lambda}$ 
in \citet[Lemma 6]{lei2014distribution}.  Note that 
$\Prob(E_2^c) = O(n^{-\lambda})$.  Let 
$v_1 = z_nM''' + r_nM'(A_\lambda + A_\lambda')$ and $v_2 = z_n\xi_\lambda$ 
and note that, for sufficiently large $n$, these choices satisfy 
\citet[Lemma 8]{lei2014distribution} which then gives 
$$
  |\thatalpha_k - \talpha| \leq v_1 + c_1^{-1}v_2, 
$$
and
$$
  \Prob\left[
    \sup_x \nu\left\{\widehat{L}_k(\thatalpha_k)\triangle \Copt_P(x)\right\} 
      \geq \xi_1v_1 + \xi_2v_2
  \right] = O(n^{-\lambda}),
$$
where $c_1$, $\xi_1$, and $\xi_2$ are all constants that are independent 
of $n$.  

Let $\ttildealpha_k = \thatalpha_k + 2r_nM'A_{\lambda}'' + 2z_nM'''$. Then, for 
$v_3 = 2r_nM'A_{\lambda}'' + 2z_nM'''$  and constants $\xi_j'$, $j = 1,2,3$, 
\citet[Lemma 8]{lei2014distribution} gives us
$$
  \Prob\left[
    \sup_x \nu\left\{\widehat{L}_k^-\triangle \Copt_P(x)\right\} 
      \geq \xi_1'v_1 + \xi_2v_2' + \xi_3'v_3
  \right] = O(n^{-\lambda}).
$$
Similarly, for $v_3' = - 2r_nM'A_{\lambda}'' - 2z_nM'''$  and constants 
$\xi_j''$, $j = 1,2,3$, \citet[Lemma 8]{lei2014distribution} gives us
$$
  \Prob\left[
    \sup_x \nu\left\{\widehat{L}_k^+\triangle \Copt_P(x)\right\} 
      \geq \xi_1''v_1 + \xi_2v_2'' + \xi_3''v_3
  \right] = O(n^{-\lambda}).
$$
Then, conditional on $E_2$, we have that 
$
  \left\{\widehat{L}^-_k \subseteq \Coptglm(x) \subseteq \widehat{L}^+_k\right\},
$
and this implies that 
$$
  \nu\left\{\Coptglm(x) \triangle C_P^{(\alpha)}(x)\right\}
    \leq \nu\left\{\widehat{L}^-_k \triangle C_P^{(\alpha)}(x)\right\}
    + \nu\left\{\widehat{L}^+_k \triangle C_P^{(\alpha)}(x)\right\}.
$$
Therefore our conclusion holds for some $\zeta_\lambda'$ at rate 
$r_n \vee z_n$.
\end{proof}

\subsection*{Preliminaries for the proof of Theorem~\ref{thm:correctmodel2}}

We will prove Theorem~\ref{thm:correctmodel2} under the specification that 
the number of modes of $p_\psi(\cdot|x)$ is equal to 1 for all $x \in \X$, 
so that \eqref{transregion} and \eqref{transregion2} are equivalent.
The general case follows a similar argument.
With an abuse of notation, let $(a(x), b(x))$ be the solution to 
\begin{equation} \label{minproblemtrue}
  \text{argmin}_{(a(x),b(x))} \left(b(x) - a(x)\right), 
      \; \text{subject to} \; \int_{a(x)}^{b(x)} 
        p_\psi(u|x) du = 1 - \alpha, \; \text{and} \; b(x) - a(x) > 0.
\end{equation}
We now provide a probability bound on the distance between the solutions 
to \eqref{minproblem} and \eqref{minproblemtrue}.  This bound is of 
fundamental importance for proving Theorem~\ref{thm:correctmodel2}.

\begin{lem}
Let $Y|X=x$ be a random variable with conditional density 
$p_\psi(\cdot|x)$ and parameter space $\aleph_\X$.  Assume that 
$p_\psi(\cdot|x)$ is twice differentiable in $\psi$ and satisfies 
Assumption 2.  Let $(X_1,Y_1)$, $\ldots$, $(X_n,Y_n)$ be independent and 
identically distributed copies of $(X,Y)$.  Let $\psi \in \R^r$.  
Let $\psihat$ be the MLE of $\psi$.  Augment the sample data with a new 
point $(x,y)$, and let $\psihat^{(x,y)}$ be the MLE of $\psi$ with respect 
to the augmented data.  Suppose that Assumption 1 holds.  
Let $0 < \alpha < 1$.   
Let $(\hat a(x),\hat b(x))$ be the solution to the optimization problem 
\eqref{minproblem} and let $(a(x), b(x))$ be the solution to the 
optimization problem \eqref{minproblemtrue}.  
Let $\lambda > 0$ be as is defined in Lemma~\ref{lem:dens}.  Then we can pick 
a number $\xi_{\lambda}'$ such that
$$
  \Prob\left(\sup_{x\in\X}|(\hat a(x), \hat b(x)) - (a(x), b(x))| 
    \geq \xi_{\lambda}'r_n\right) = O(n^{-\lambda}).
$$
\label{lem:optim}
\end{lem}

\begin{proof}
Let $E$ be the event that  
$\|\hat{p}_\psi^{(x,y)}(\cdot|x) - p_\psi(\cdot|x)\|_\infty \leq \xi_\lambda r_n$ 
with $\xi_\lambda$ defined in Lemma~\ref{lem:dens}.  
Note that $\Prob(E^c) = O\left(n^{-\lambda}\right)$.  
As a reminder, WLOG we are assuming that the class of densities $p_\psi$ are 
unimodal over the parameter space $\aleph_\X$.  In what follows, 
we will assume that $a(x) < b(x)$ is an optimization constraint.  
On $E$ we have that,
\begin{align*}
  &\sup_{x \in \X}\left\{
    [\text{argmin}_{(a(x),b(x))} \left(b(x) - a(x)\right), 
    \; \text{subject to} \; \int_{a(x)}^{b(x)} 
      \hat{p}_\psi^{(x,y)}(u|x) du = 1 - \alpha] \right. \\ 
    &\qquad- \left.
      [\text{argmin}_{(a(x),b(x))} \left(b(x) - a(x)\right), 
      \; \text{subject to} \; \int_{a(x)}^{b(x)} 
        p_\psi(u|x) du = 1 - \alpha] 
  \right\} \\
  &\leq\sup_{x \in \X}\left\{
    [\text{argmin}_{(a(x),b(x))} \left(b(x) - a(x)\right), 
    \; \text{subject to} \; 
      \int_{a(x)}^{b(x)} \{p_\psi(u|x) - \xi_{\lambda}r_n\} du = 1 - \alpha] 
        \right. \\
    &\qquad- \left.
      [\text{argmin}_{(a(x),b(x))} \left(b(x) - a(x)\right), 
      \; \text{subject to} \; 
        \int_{a(x)}^{b(x)} \{p_\psi(u|x) + \xi_{\lambda}r_n\} du = 1 - \alpha] 
  \right\} \\
  &= \sup_{x \in \X}\left\{
    [\text{argmin}_{(a(x),b(x))} \left(b(x) - a(x)\right), 
    \; \text{subject to} \; \int_{a(x)}^{b(x)} p_\psi(u|x) du 
      -(b(x) - a(x))\xi_{\lambda}r_n = 1 - \alpha] \right. \\ 
      &\qquad- \left.
        [\text{argmin}_{(a(x),b(x))} \left(b(x) - a(x)\right), 
        \; \text{subject to} \; \int_{a(x)}^{b(x)} p_\psi(u|x) du 
        + (b(x) - a(x))\xi_{\lambda}r_n = 1 - \alpha] 
  \right\}.
\end{align*}
Since the optimization problems above correspond to highest density 
regions for a uniformly bounded class of densities, it must be the 
case that 
$\inf_{x \in \X}a(x) > -\infty$, 
$\sup_{x \in \X}b(x) < \infty$.
For all $\alpha > 0$ in the optimization problems above, 
there exists an $N$ and $M$ such that 
$$
  \int_{a(x)}^{b(x)} p_\psi(u|x) du + (b(x) - a(x))\xi_{\lambda}r_n 
    \leq \int_{a(x)}^{b(x)} p_\psi(u|x) du + M\xi_{\lambda}r_n 
$$ 
and
$$
  \int_{a(x)}^{b(x)} p_\psi(u|x) du - (b(x) - a(x))\xi_{\lambda}r_n 
    \geq \int_{a(x)}^{b(x)} p_\psi(u|x) du - M\xi_{\lambda}r_n 
$$
for all $n > N$.  This yields that
\begin{align*}
  &\sup_{x \in \X}\left\{
    [\text{argmin}_{(a(x),b(x))} \left(b(x) - a(x)\right), 
    \; \text{subject to} \; \int_{a(x)}^{b(x)} p_\psi(u|x) du 
      -(b(x) - a(x))\xi_{\lambda}r_n = 1 - \alpha] \right. \\ 
      &\qquad- \left.
        [\text{argmin}_{(a(x),b(x))} \left(b(x) - a(x)\right), 
        \; \text{subject to} \; \int_{a(x)}^{b(x)} p_\psi(u|x) du 
        + (b(x) - a(x))\xi_{\lambda}r_n = 1 - \alpha] 
  \right\} \\
  &\leq \sup_{x \in \X}\left\{
    [\text{argmin}_{(a(x),b(x))} \left(b(x) - a(x)\right), 
    \; \text{subject to} \; \int_{a(x)}^{b(x)} p_\psi(u|x) du 
      - M\xi_{\lambda}r_n = 1 - \alpha] \right. \\ 
      &\qquad- \left.
        [\text{argmin}_{(a(x),b(x))} \left(b(x) - a(x)\right), 
        \; \text{subject to} \; \int_{a(x)}^{b(x)} p_\psi(u|x) du 
        + M\xi_{\lambda}r_n = 1 - \alpha] 
  \right\} \\
  &= \sup_{x \in \X}\left\{
    [\text{argmin}_{(a(x),b(x))} \left(b(x) - a(x)\right), 
    \; \text{subject to} \; \int_{a(x)}^{b(x)} p_\psi(u|x) du 
      = 1 - \alpha + M\xi_{\lambda}r_n] \right. \\ 
      &\qquad- \left.
        [\text{argmin}_{(a(x),b(x))} \left(b(x) - a(x)\right), 
        \; \text{subject to} \; \int_{a(x)}^{b(x)} p_\psi(u|x) du 
        = 1 - \alpha - M\xi_{\lambda}r_n] 
  \right\}.
\end{align*}
Let $a_{n,1}(x)$ and $b_{n,1}(x)$ be the solution to 
$$
  \text{argmin}_{(a(x),b(x))} \left(b(x) - a(x)\right), 
    \; \text{subject to} \; \int_{a(x)}^{b(x)} p_\psi(u|x) du 
      = 1 - \alpha + M\xi_{\lambda}r_n, 
$$
and let $a_{n,2}(x)$ and $b_{n,2}(x)$ be the solution to 
$$
  \text{argmin}_{(a(x),b(x))} \left(b(x) - a(x)\right), 
    \; \text{subject to} \; \int_{a(x)}^{b(x)} p_\psi(u|x) du 
      = 1 - \alpha - M\xi_{\lambda}r_n.
$$
Since $\ptrue$ is unimodal, we have that $a_{n,1}(x) < a_{n,2}(x)$ 
and $b_{n,2}(x) < b_{n,1}(x)$.  The combination of this and the setup of 
the optimization problems implies that 
\begin{align*}
  &\sup_{x\in\X}\left\{(a_{n,2}(x) - a_{n,1}(x))p_\psi(a_{n,1}(x)|x) 
    + (b_{n,1}(x) - b_{n,2}(x))p_\psi(b_{n,1}(x)|x)\right\} \\
    &\qquad= \sup_{x\in\X}\left[\left\{(a_{n,2}(x) - a_{n,1}(x)) 
      + (b_{n,1}(x) - b_{n,2}(x))\right\}p_\psi(a_{n,1}(x)|x)\right] \\
    &\qquad\leq 2M\xi_{\lambda}r_n.
\end{align*}
The number $p_\psi(a_{n,1}(x)|x)$ is increasing in $n$ for all $x \in \X$ 
since $1 - \alpha + M\xi_{\lambda}r_n$ is decreasing in $n$.  We can pick a 
number $m$ such that $\inf_{x\in\X}p_\psi(a_{n,1}(x)|x) > m$.  Thus, 
$$
  \sup_{x\in\X}\left[\left\{(a_{n,2}(x) - a_{n,1}(x)) 
      + (b_{n,1}(x) - b_{n,2}(x))\right\}\right] 
    \leq \frac{2M\xi_{\lambda}r_n}{m}.
$$
Let $\xi_{\lambda}' = 2M\xi_{\lambda}/m$.  
This implies that  
$$
  \sup_{x\in\X}(a_{n,2}(x) - a_{n,1}(x)) \leq \xi_{\lambda}'r_n 
    \qquad \text{and} \qquad
  \sup_{x\in\X}(b_{n,1}(x) - a_{n,2}(x)) \leq \xi_{\lambda}'r_n.
$$
Therefore we have that 
$|(\hat a(x), \hat b(x)) - (a(x), b(x))| \leq \xi_{\lambda}'r_n$.  
Our conclusion follows.
\end{proof}

\begin{lem}
Let $Y|X=x$ be a random variable with conditional density 
$p_\psi(\cdot|x)$ and parameter space $\aleph_\X$.  Assume that 
$p_\psi(\cdot|x)$ is twice differentiable in $\psi$ and satisfies 
Assumption 2.  Let $F(y|x)$ be the conditional distribution function 
corresponding to $p_\psi(\cdot|x)$ and assume that $F(y|x)$ satisfies 
Assumption 4.    
Let $(X_1,Y_1)$, $\ldots$, $(X_n,Y_n)$ be independent and 
identically distributed copies of $(X,Y)$.  Let $\psi \in \R^r$.  
Let $\psihat$ be the MLE of $\psi$.  Augment the sample data with a new 
point $(x,y)$.  Suppose that Assumption 1 holds.  
Let $\psihat^{(x,y)}$ be the MLE of $\psi$ with respect to the augmented data. 
Let $\hat{F}^{(x,y)}(\cdot|x)$ be the estimated conditional distribution 
$F(y|x)$ with $\psihat^{(x,y)}$ plugged in.  
Given $\lambda > 0$, there is a numerical constant 
$\chi_{\lambda}'$ such that 
$$
  \Prob\left\{\sup_{x \in \X}\|\hat{F}^{(x,y)}(\cdot|x) - F(\cdot|x)\|_{\infty} 
    \geq \chi_{\lambda}'r_n\right\} 
  = O\left(n^{-\lambda}\right).
$$  
\label{lem:dist}
\end{lem}

\begin{proof}
Let $E$ be the event that $|\psihat^{(x,y)} - \psi| < A_{\lambda}'r_n$ with 
$\lambda$ and $A_{\lambda}$ defined in Corollary~\ref{concentration-cor} and 
note that $\Prob(E^c) = O\left(n^{-\lambda}\right)$ by 
Corollary~\ref{concentration-cor}.
A Taylor expansion of $\hat{F}^{(x,y)}(v|x)$ around $F(v|x)$ yields 
\begin{align*}
  |\hat{F}^{(x,y)}(v|x) - F(v|x)| &= 
    |\nabla_\psi F(v|x)'(\psihat^{(x,y)} - \psi) + o(|\psihat^{(x,y)} - \psi|)| \\
  &= |\nabla_\psi F(v|x)'(\psihat^{(x,y)} - \psi) + o(r_n)|, 
\end{align*}
on $E$.  Boundedness of $|\nabla_\psi F(v|x)|$ for all $v$ and all 
$\psi \in \aleph_\X$ implies that  
$|\nabla_\psi F(v|x)'(\psihat^{(x,y)} - \psi)| = O(r_n)$.  
This then implies that there exists an $N$ such that 
\begin{align*}
  |\nabla_\psi F(v|x)'(\psihat^{(x,y)} - \psi)| + o(r_n)| 
    &\leq |\nabla_\psi F(v|x)||\psihat^{(x,y)} - \psi| + |o(r_n)| \\
    &\leq |\nabla_\psi F(v|x)|A_{\lambda}'r_n, 
\end{align*}
for all $n > N$ on $E$.  Our conclusion follows for some 
$\chi_{\lambda}' > 0$.
\end{proof}

We now build on Lemma~\ref{lem:dist}.

\begin{lem}
Define $U_{\text{lwr}}(x) = F(a(x)|x)$ and 
$U_{\text{upr}}(x) = F(b(x)|x)$ 
where $(a(x), b(x))$ as defined in \eqref{minproblemtrue}.  
Define $\lambda > 0$ as in Lemma~\ref{lem:optim} and Lemma~\ref{lem:dist}.  
Then there exists $\eta_{\text{lwr}}, \eta_{\text{upr}} > 0$ such that 
$$
  \Prob\left(\sup_{x\in\X}|\widehat{U}^{(x,y)}_{\text{lwr}} - U_{\text{lwr}}(x)| 
    \geq \eta_\lambda^{\text{lwr}}r_n\right) = O(n^{-\lambda}),
      \quad \text{and} \quad 
  \Prob\left(\sup_{x\in\X}|\widehat{U}^{(x,y)}_{\text{upr}} - U_{\text{upr}}(x)| 
    \geq \eta_\lambda^{\text{upr}}r_n\right) = O(n^{-\lambda}).
$$
\label{lem:percentiles}
\end{lem}

\begin{proof}
Set $\lambda > 0$ and let $E$ be the event that 
$
  \sup_{x\in\X}|(\hat a(x), \hat b(x)) - (a(x), b(x))| 
    \leq \xi_{\lambda}'r_n,
$
and \\
$ 
  \sup_{x \in \X}\|\hat{F}^{(x,y)}(\cdot|x) - F(\cdot|x)\|_{\infty} 
    \leq \chi_{\lambda}'r_n, 
$
where $\xi_{\lambda}'$ is defined as in Lemma~\ref{lem:optim} and 
$\chi_{\lambda}'$ is defined as in Lemma~\ref{lem:dist}, and note that 
these two Lemmas give $\Prob(E^c) = O(n^{-\lambda})$.  
Then,
\begin{align*}
  |\widehat{U}^{(x,y)}_{\text{lwr}} - U_{\text{lwr}}(x)| 
    &= |\widehat F^{(x,y)}(\hat a(x)|x) \pm F(\hat a(x)|x) - F(a(x)|x)| \\
    &\leq |\widehat F^{(x,y)}(\hat a(x)|x) - F(\hat a(x)|x)| + 
      |F(\hat a(x)|x) - F(a(x)|x)| \\
    &\leq r_n\chi' 
      + \left|\frac{F(\hat a(x)|x) - F(a(x)|x)}{\hat a(x) - a(x)}\right|
        |\hat a(x) - a(x)|.
\end{align*}
Arbitrarily choose $\varepsilon > 0$, then for $n$ sufficiently large, 
we have that
$$
  \left|\frac{F(\hat a(x)|x) - F(a(x)|x)}{\hat a(x) - a(x)}\right|
    |\hat a(x) - a(x)| \leq (p_\psi(a(x)|x) + \varepsilon)\xi_{\lambda}'r_n.
$$
We can pick $L \in \R$ such that 
$\sup_{x \in \X}p_\psi(a(x)|x) \leq L$.  Putting all of this together, 
we have 
$$
  |\widehat{U}^{(x,y)}_{\text{lwr}} - U_{\text{lwr}}(x)| 
    \leq (\chi_\lambda' + (L + \varepsilon)\xi_{\lambda}')r_n,
$$
for $n$ sufficently large.  Set 
$\eta_\lambda^{\text{lwr}} = (\chi_\lambda' + (L + \varepsilon)\xi_{\lambda}')$ 
and note that the right hand side of the above displayed equation does not 
depend on $x$.  The same argument gives 
$
  |\widehat{U}^{(x,y)}_{\text{upr}} - U_{\text{upr}}(x)| 
    \leq \eta_\lambda^{\text{upr}}r_n,
$
for $n$ sufficently large.  Our conclusion follows.
\end{proof}

\subsection*{Proof of Theorem 3}

\begin{proof}[proof of Theorem~\ref{thm:correctmodel2}]
Let $\lambda > 0$ be as defined in Lemma~\ref{lem:percentiles}. 
Define $\xi_{\lambda}'$ as in Lemma~\ref{lem:optim}, 
$\chi_{\lambda}'$ as in Lemma~\ref{lem:dist},
and both $\eta_{\lambda}^{\text{lwr}}$ and $\eta_{\lambda}^{\text{upr}}$ 
as in Lemma~\ref{lem:percentiles}.
Let $E$ be the event such that 
$
  \sup_{x\in\X}|(\hat a(x), \hat b(x)) - (a(x), b(x))| 
    \leq \xi_{\lambda}'r_n,
$
and 
$ 
  \sup_{x \in \X}\|\hat{F}^{(x,y)}(\cdot|x) - F(\cdot|x)\|_{\infty} 
    \leq \chi_{\lambda}'r_n, 
$
and 
$
  \sup_{x\in\X}\left(|\widehat{U}^{(x,y)}_{\text{lwr}} - U_{\text{lwr}}(x)|\right)
    \leq \eta_\lambda^{\text{lwr}}r_n,
$
and \\
$
  \sup_{x\in\X}\left(|\widehat{U}^{(x,y)}_{\text{upr}} - U_{\text{upr}}(x)|\right)
    \leq \eta_\lambda^{\text{upr}}r_n.
$    
Note that Lemmas~\ref{lem:optim}-\ref{lem:percentiles} implies that  
$\Prob(E^c) = O(n^{-\lambda})$.  Condition on $E$ and note that
$$
  (n+1)U_{\text{lwr}}^{(x,y)} 
    \geq (n+1)(U_{\text{lwr}}(x) - \eta_{\lambda}^{\text{lwr}}r_n),
$$
and
$$
  (n+1)U_{\text{upr}}^{(x,y)} 
    \leq (n+1)(U_{\text{upr}}(x) + \eta_{\lambda}^{\text{upr}}r_n).
$$  
With these specifications, we have that
\begin{align*}
  \widehat{U}^{(x,y)}_{\left[\floor{(n+1)\widehat{U}^{(x,y)}_{\text{lwr}}}\right]}
    &\geq \widehat{U}^{(x,y)}_{\left[\floor{
      (n+1)(U_{\text{lwr}}(x) - \eta_\lambda^{\text{lwr}}r_n)
    }\right]} \\
    &= \widehat{F}^{(x,y)}(Y_{\left[\floor{
        (n+1)(U_{\text{lwr}}(x) - \eta_\lambda^{\text{lwr}}r_n)
      }\right]} | X_{\left[\floor{
        (n+1)(U_{\text{lwr}}(x) - \eta_\lambda^{\text{lwr}}r_n)
      }\right]}) \\
    &\geq F(Y_{\left[\floor{
        (n+1)(U_{\text{lwr}}(x) - \eta_\lambda^{\text{lwr}}r_n)
      }\right]} | X_{\left[\floor{
        (n+1)(U_{\text{lwr}}(x) - \eta_\lambda^{\text{lwr}}r_n)
      }\right]}) - \chi_{\lambda}'r_n \\
    &\geq U_{\left[\floor{
        (n+1)(U_{\text{lwr}}(x) - \eta_\lambda^{\text{lwr}}r_n)
      }\right]} - \chi_{\lambda}'r_n - |O(n^{-1/2})|,
\end{align*}
where  
$
  U_{\left[\floor{
        (n+1)(U_{\text{lwr}}(x) - \eta_\lambda^{\text{lwr}}r_n)
      }\right]}
$
is the 
$\left[\floor{ (n+1)(U_{\text{lwr}}(x) - \eta_\lambda^{\text{lwr}}r_n) }\right]$th 
quantile of the sample of uniform random variables, $F(Y_i|X_i)$, $i=1,...,n+1$ 
and the last inequality follows from \citet[Theorem 13]{ferguson1996course}.  
Now, let 
$$
  z_n = \frac{\floor{
        (n+1)(U_{\text{lwr}}(x) - \eta_\lambda^{\text{lwr}}r_n)}}
        {n+1},
$$
and note that 
$$
  U_{\left[\floor{
        (n+1)(U_{\text{lwr}}(x) - \eta_\lambda^{\text{lwr}}r_n)
      }\right]}
   \geq z_n - n^{-1/2}|B_n(z_n)| - |O(n^{-1}\log(n))|, 
$$
where $\{B_n(u): 0\leq u\leq 1\}$ is a sequence of Brownian bridges. 
This inequality follows from \citet[Theorem 1]{csorgo1978strong} 
and \citet{komlos1975approximation} and a bit of algebra.  
We can pick $\chi_{\lambda}''$ such that, for $n$ sufficiently large, 
$$
  z_n - n^{-1/2}|B_n(z_n)| - |O(n^{-1}\log(n))| - \chi_{\lambda}'r_n - |O(n^{-1/2})|
    \geq U_{\text{lwr}}(x) - \chi_{\lambda}''r_n.
$$

A similar argument gives, 
\begin{align*}
  \widehat{U}^{(x,y)}_{\left[\floor{(n+1)\widehat{U}^{(x,y)}_{\text{upr}}}\right]} 
    &\leq U_{\text{upr}}(x) + \chi_{\lambda}''r_n, 
\end{align*}
for sufficiently large $n$.
We also have that 
$|\widehat{U}^{(x,y)}_{n+1} - F(y|x)| \leq \chi_{\lambda}'r_n$.

Putting all of this together, we have that 
$$
  \Copttrans(x) \subseteq \{ y: U_{\text{lwr}}(x) - r_n(\chi_{\lambda}'+\chi_{\lambda}'')
      \leq F(y|x) \leq   
    U_{\text{upr}}(x) + r_n(\chi_{\lambda}'+\chi_{\lambda}'')\},
$$
for $n$ sufficiently large.  Let $F_x^{-1}(\cdot)$ be the inverse of 
$F(\cdot|x)$, then
\begin{align*}
  &\{ y: U_{\text{lwr}}(x) - r_n(\chi_{\lambda}'+\chi_{\lambda}'') 
      \leq F(y|x) \leq   
    U_{\text{upr}}(x) + r_n(\chi_{\lambda}'+\chi_{\lambda}'')\} \\
  &\qquad=
    \left\{ y: F_x^{-1}\left(F(a(x)|x) - r_n(\chi_{\lambda}'+\chi_{\lambda}'')\right) 
      \leq y \leq   
    F_x^{-1}\left(F(b(x)|x) + r_n(\chi_{\lambda}'+\chi_{\lambda}'')\right)\right\}.
\end{align*}
For $n$ sufficiently large, we can pick $K > 0$ such that  
\begin{align*}
  &\left\{ y: F_x^{-1}\left(F(a(x)|x) - r_n(\chi_{\lambda}'+\chi_{\lambda}'')\right) 
      \leq y \leq   
    F_x^{-1}\left(F(b(x)|x) + r_n(\chi_{\lambda}'+\chi_{\lambda}'')\right)\right\} \\
  &\qquad\subseteq
    \left\{ y: a(x) - r_n(\chi_{\lambda}'+\chi_{\lambda}'')K
      \leq y \leq   
    b(x) + r_n(\chi_{\lambda}'+\chi_{\lambda}'')K\right\}.
\end{align*}
Note that 
$$
  \nu\left(C_P(x) \triangle \{ a(x) - r_n(\chi_{\lambda}'+\chi_{\lambda}'')K
      \leq y \leq   
    b(x) + r_n(\chi_{\lambda}'+\chi_{\lambda}'')K\}\right) 
    = 2r_n(\chi_{\lambda}'+\chi_{\lambda}'')K.
$$

A similar argument gives
$$
  \Copttrans(x) \supseteq \left\{ y: a(x) - r_n(\chi_{\lambda}'+\chi_{\lambda}'')K'
      \leq y \leq   
    b(x) + r_n(\chi_{\lambda}'+\chi_{\lambda}'')K'\right\},
$$
for some $K' > 0$ and $n$ sufficiently large, where
$$
  \nu\left(C_P(x) \triangle \left\{ y: 
    a(x) - r_n(\chi_{\lambda}'+\chi_{\lambda}'')K'
      \leq y \leq   
    b(x) + r_n(\chi_{\lambda}'+\chi_{\lambda}'')K'\right\}\right) 
  = 2r_n(\chi_{\lambda}'+\chi_{\lambda}'')K'.
$$

Putting everything together yields
$$
  \nu\left(\Copttrans(x) \triangle C_P(x)\right) 
    \leq 4r_n(\chi_{\lambda}'+\chi_{\lambda}'')\max(K,K').
$$
Set $\chi_\lambda = 4(\chi_{\lambda}'+\chi_{\lambda}'')\max(K,K')$.  Therefore
$$
  \nu\left(\Copttrans(x) \triangle C_P(x)\right) 
    \leq \chi_{\lambda}r_n.
$$
Our conclusion follows.
\end{proof}

\newpage
\section*{Supplementary Materials for 
  ``Efficient and minimal length parametric conformal prediction regions''}

We provide additional technical and numeric 
results which demonstrate the advantages and disadvantages 
of our parametric conformal prediction region 
compared to the nonparametric conformal prediction region 
\citep{lei2014distribution}, the LS conformal prediction 
region \citep{lei2018distribution} obtained from conformalized residual 
scores, the LSLW conformal prediction region 
\citep[Section 5.2]{lei2018distribution} obtained from conformalized locally 
weighted residual scores, and the highest density prediction region.  In 
analyses with model misspecification, the parametric, LS, and LSLW 
conformal prediction regions and highest density 
prediction region are constructed under the misspecified model.  
The binning used to construct the parametric and nonparametric conformal 
prediction regions follows the bin width asymptotics of 
\citet{lei2014distribution}.  Comparative diagnostics are discussed in 
Section 5 of the main text.

We expand upon the three simulation settings in the sensitivity analysis in 
Section 5.  These simulation settings are:  
\begin{itemize}
\item[A)] Gamma regression with $\beta = [1.25, -1]'$ and 
$n \in \{150, 250, 500\}$. 
In this setting the Gamma density is correctly specified for the parametric 
conformal and the highest density prediction regions.  A cubic regression 
model is assumed for the LS and LSLW 
conformal prediction regions.  
\item[B)] Gamma regression with $\beta = [0.5, 1]'$ and 
$n \in \{150, 250, 500\}$. 
In this setting the cubic regression model is assumed for the highest density 
prediction region and the misspecified parametric conformal, LS, 
and LSLW conformal prediction regions. 
\item[C)] Simple linear regression with normal errors and constant variance.  
We set $\beta = [2, 5]'$, and the normal errors to have variance 
$\sigma^2 = 1$.  Results are considered for sample sizes 
$n \in \{150, 250, 500\}$.  In this setting the regression model is correctly 
specified for the highest density prediction region and the parametric, 
LS, and LSLW conformal prediction regions.  
\end{itemize}

In the next Section we show that all of the assumptions required for the 
asymptotic efficiency of $\Coptloc$ (see Section 2.3.1 in the 
manuscript) hold for GLMs.

\section{Conformal prediction in the GLM setting}

In this Section we establish that GLMs which can be parameterized as in 
(10) in the main text satisfy Assumptions 2-4 in the main text, and that the 
kernel density based conformal prediction region $\Coptloc(x)$ is 
appropriate for this class of GLMs.  There is considerable overlap in the 
conditions required to achieve both goals, Assumptions 2-3 in the main 
text are established in the proof that our class of GLMs satisfies the 
conditions required for the asymptotic properties of $\Coptloc(x)$ to hold.

We verify that $\Coptloc(x)$ is asymptotically minimal at rate 
$w_n = \left\{\log(n)/n\right\}^{1/(d + 3)}$ when the underlying 
distribution is a GLM.  This result is obtained by showing that our GLM 
parameterization satisfies the conditions of Theorem 1 in 
\citet{lei2014distribution}.
We now provide the conditions in \citet{lei2014distribution} that are 
required for $\Coptloc(x)$ to be asymptotically of minimal length at 
rate $w_n$. The first of these conditions are Assumptions 1-3 in the 
main text.  We additionally require Assumption 5 below.  \\

\noindent{\bf Assumption 5} (Assumptions 1 (b) and (c) of \cite{lei2014distribution}).    
(a) For all $x$, $\ptrue(\cdot|x) \in \Sigma(\eta, L)$, 
  where $\Sigma(\eta, L)$ is a H{\"o}lder class of functions defined in 
  \citet[Appendix A.1]{lei2014distribution}. 
  Correspondingly, the kernel $K$ used to construct the nonparametric 
  conformal prediction band is a $\beta$-valid kernel; 
(b) For any $0 \leq s < \floor{\eta}$, $\ptrue^{(s)}(y|x)$ is 
  continuous and uniformly bounded by $L'$ for all $x$ and $y$. \\

See \citet[Appendix A.1]{lei2014distribution} for a definition of $\beta$-valid 
kernels, the H{\"o}lder class of functions for which $\ptrue(\cdot|x)$ is 
assumed to belong to is defined below.  Additional explanations of these conditions 
are given in \citet{lei2013distribution} and \citet{lei2014distribution}. \\

\begin{defn}[H{\"o}lder class] Given $L>0$ and $\beta>0$, let $l=\floor{\beta}$ 
be the largest integer strictly less than $\beta$. The H{\"o}lder class 
$\Sigma(\beta,L)$ is the family of functions $f:\R\to\R$ whose derivative 
$f^{(l)}$ satisfies
$$
  |f^{(l)}(y)-f^{(l)}(y')| \leq |y-y'|^{\beta-l}, 
    \qquad \text{for all} \qquad y, y'.
$$
\end{defn}

We now show that $\Coptloc(x)$ is asymptotically minimal at rate $w_n$ when 
the underlying model is a GLM, where $\beta$ in Assumption 5 is chosen so that 
$0 < \beta < 2$.

\begin{prop}
Let $(Y_1$, $X_1)$, $\ldots$, $(Y_n$, $X_n)$, $Y_i \in \R^r$, $X_i \in \R^m$, 
be an independent and identically distributed sample of random variables 
with conditional density \eqref{glm}
and parameter space 
$$
  \Theta_\X = \{\beta \in \R^m,\; \phi \in \R^{r-1}: 
    c\left\{f(x'\beta, \phi)\right\} < \infty, 
      \; \text{for all} \; x \in \X\},
$$ 
Assume that $\E(Y_1|x) = g^{-1}(x'\beta)$ and that the canonical statistic 
vector is a one dimensional manifold.  Suppose that Assumption 1 in the 
manuscript holds.  Let $0 < \alpha < 1$. 
Then, for a given $\lambda > 0$, there exists a numerical constant 
$\zeta_\lambda$ such that, 
\begin{equation} 
  \Prob\left[ \sup_{x\in\X} 
    \nu\left\{\Coptloc(x)\triangle\Copt_P(x)\right\} 
      \geq \zeta_\lambda w_n
  \right] = O\left(n^{-\lambda}\right). 
\label{converge}
\end{equation}
\label{expoconf}
\end{prop}

\begin{proof}
The one dimensional manifold assumption is to ensure that all changes in the 
density function with respect to the canonical statistic can be expressed 
through any one component of the canonical statistic vector.  
The result \eqref{converge} follows from Theorem 1 in 
\cite{lei2014distribution}, provided that Assumptions 1-3 from the main text 
hold and that Assumption 5 holds.  
Assumption 1 is an assumed condition of this proposition.  
Exponential families are smooth functions over their parameter space.  
This implies that Assumption 2 from the main text holds for exponential 
families since $\X$ is bounded.  

We now show that Assumption 5 holds.  WLOG, let the exponential family be 
regular full with density $f_\theta(y) = e^{y\theta - c(\theta)}$.  
Then, for every integer $s \leq \floor{\eta}$, we have that 
$f_\theta^{(s)}(y) = \theta^s f_\theta(y)$.  It follows that 
$f_\theta^{(s)}(y)$ is continuous and uniformly bounded, and this establishes 
part (b)of Assumption 5.  We now establish part (a) of Assumption 5 with 
$0 < \beta < 2$, with the only case left to prove is when $1 < \beta < 2$ 
and $s = 1$.  In this case, we can write 
$$
  |f_\theta^{(1)}(y)-f_\theta^{(1)}(y')| 
    = \theta f_\theta(y)|1 - \exp\{\theta(y'-y)\}|.
$$ 
We can bound $\theta f_\theta(y)$ by some $M_1 > 0$ for all $y$ and 
$1 < \beta < 2$, and we can bound $|1 - \exp\{\theta(y'-y)\}|$ by $M_2$ 
for all $y,y'$.  Thus, 
$$
  |f_\theta^{(1)}(y)-f_\theta^{(1)}(y')| 
    \leq M_1M_2(y' - y)^{\beta - 1}
$$ 
for all $y,y'$ such that $|y - y'| \geq 1$.  We can choose $\delta > 0$ such 
that $\exp\{\theta(y'-y) \leq 1 + 2\theta|y' - y|$ for all 
$y,y'$ such that $|y' - y| \leq \delta$.  Thus, for all $y,y'$ such that 
$|y' - y| \leq \delta$, we have that 
$$
  |f_\theta^{(1)}(y)-f_\theta^{(1)}(y')| 
    \leq 2M_1(y' - y)^{\beta - 1}.
$$ 
Finally, for all $y,y'$ such that $\delta \leq |y-y'| < 1$, we have that 
$$
  |f_\theta^{(1)}(y)-f_\theta^{(1)}(y')| 
    \leq 2M_1M_2\left(\frac{y' - y}{\delta}\right)^{\beta - 1}.
$$ 
Part (a) of Assumption 5 holds with $L = 2M_1M_2\delta^{-2}$.

We now verify that Assumption 3 from the main text holds when 
$\ptrue(y|x)$ is an exponential family with density \eqref{glm}.
We will assume WLOG that generating measure $\mu$ is Lebesgue measure.    
Define $\talpha_x$ such that 
$
  \Prob\left\{y: \ptrue(y|x) \geq \talpha_x\right\} = \alpha 
$ 
and let $\Aaxe = \{y: |\ptrue(y|x)-\talpha_x| \leq \varepsilon\}$.  
When $\ptrue(y|x)$ is an exponential family with density \eqref{glm} 
we can choose a $\varepsilon_o > 0$ such that, for all 
$\varepsilon < \varepsilon_o$, we can pick $a_1$, $a_2 \in \R$ for which 
$$
  a_1\varepsilon \leq \mu(\Aaxe) \leq a_2\varepsilon.
$$
Let $M(x) = \sup_{y \in \Aaxe(x)}\{\ptrue(y|x)\}$.  We can pick 
$M_3 < \infty$ such that $\sup_{x\in\X} M(x) \leq M_3$. Then we have, 
\begin{align*}
  &\Prob\left[\left\{
      y:|\ptrue(y|x)-\talpha_x| < \varepsilon
    \right\}| X=x\right] \\
  &\qquad=  \int_{\Aaxe} \exp\left[
      \inner{y, f(x'\beta, \phi)} - c\left\{f(x'\beta, \phi)\right\}
    \right]\mu(dy) \\
  &\qquad\leq M(x)\mu(\Aaxe) \leq a_2M_3\varepsilon.    
\end{align*}
The upper bound in Assumption 3 follows from the choice of $c_2 = a_2M_3$.  
Now pick $\varepsilon_o$ such that $\inf_x \talpha_x \geq \varepsilon_o$.  
This choice of $\varepsilon_o$ is possible since $\alpha < 1$.  
Let $m(x) = \inf_{y \in \Aaxe(x)}\{\ptrue(y|x)\}$. 
We can pick $M_4 > 0$ such that $\inf_{x\in\X}m(x) \geq M_4$.  Then we have, 
\begin{align*}
  &\Prob\left[\left\{
      y:|\ptrue(y|x)-\talpha_x| < \varepsilon
    \right\}| X=x\right] \\
  &\qquad=  \int_{\Aaxe} \exp\left[
      \inner{y, f(x'\beta, \phi)} - c\left\{f(x'\beta, \phi)\right\}
    \right]\nu(dy) \\
  &\qquad\geq m(x)\mu(\Aaxe)  \geq a_1M_4\varepsilon.
\end{align*}
The lower bound in Assumption 3 follows from the choice of $c_1 = a_1M_4$.  
Therefore Assumption 3 holds when $\ptrue(y|x)$ is an exponential family with 
density \eqref{glm}.
\end{proof}

We now verify that Assumption 5 in the main text holds, we show that  
If $|\nabla_\psi F(v|x)|$ is bounded for all $v$ and all $\psi \in \Theta_\X$.

The gradient of $\log \ptrue(y|x)$ is then 
$
  \nabla_{\psi} B(x)
    \left\{y - \nabla_{\theta}c\left(\theta\right)\right\}
$
where 
$$
  B(x) = \left(\begin{array}{cc}
    x & 0 \\
    0 & I 
  \end{array}\right)\left\{\nabla_{\psi}f(x'\beta, \phi)\right\}',
$$
and $\theta = f(x'\beta, \phi)$.  
From \citet[Theorem 11]{trench1978functions}, we have that
\begin{align*}
  \nabla_\psi F(v|x) = \nabla_\psi\int_{-\infty}^v p_{\psi}(u|x) du 
    &= \int_{-\infty}^v \nabla_\psi p_{\psi}(u|x) du \\      
    &= B(x)\int_{-\infty}^v\left\{u - \nabla_\theta c(\theta)\right\} 
      p_{\psi}(u|x) du. 
\end{align*}
The quantity $B(x)$ is bounded in $x$ since $\X$ is bounded.  Existence of 
moments of all orders for the random variables with conditional density 
$p_{\psi}(u|x)$ with parameter space $\Theta_\X$ implies that 
$$
  \int_{-\infty}^v\left\{u - \nabla_\theta c(\theta)\right\} 
      p_{\psi}(u|x) du
$$
is bounded for all $v$ and all $x \in \X$.  Therefore 
$|\nabla_\psi F(v|x)|$ is bounded for all $v$ and all $x \in \X$. 
Therefore Assumption 5 holds for GLMs with densities \eqref{glm}.

\section{Summary of simulation results}
 
We provide additional numerical evidence that parametric conformal 
prediction regions perform well even when the model is misspecified.   
The guarantee of finite-sample marginal and local validity in the presence 
of model misspecification are noted benefits of both the parametric and 
nonparametric conformal prediction regions.  
However, the parametric and nonparametric conformal prediction regions are 
visually very different and give different prediction errors at small to 
moderate sample sizes as seen in Section~\ref{sec:plotsofregions}.  
We see that the both parametric conformal 
prediction regions adapt naturally to the data when the model is correctly 
specified or modest deviations from the specified model are present.  
Large deviations from model misspecification are not handled well as seen 
in the top 2 rows of 
Figures~\ref{conformal-plots-A-150}-\ref{conformal-plots-A-500}, although 
the parametric conformal prediction region does give nominal marginal and 
local coverage with respect to binning in finite-samples.  
When the model is correctly specified, the parametric conformal prediction 
region increasingly resembles the highest density oracle prediction region  
as the sample size increases.
On the other hand, the nonparametric conformal prediction region does not 
adapt well to data obtained from a Gamma regression model or data obtained 
from a linear regression model with a steep mean function (steepness 
is relative to the variability about the mean function) in small to moderate 
sample sizes.

The LS conformal prediction region obtains marginal validity 
\citep{lei2018distribution} but performs poorly when deviations about the 
estimated mean function are either not symmetric, not constant, or both.
When heterogeneity is present, the LS conformal prediction region 
exhibits undercoverage in regions where variability about the mean function 
is large and overcoverage in regions where variability about the mean function 
is small.  This conformal prediction region is very sensitive to model 
misspecification.  
The LSLS conformal prediction region also obtains 
marginal validity \citep[Section 5.2]{lei2018distribution} and it is far less 
sensitive to model misspecification than the LS conformal prediction region  
and it performs well under most model misspecification.  
However, the LSLW conformal prediction region is not 
appropriate when deviations about an estimated mean function are obviously not 
symmetric, as evidenced in Section~\ref{sec:gammaplots}. 
Note that all conformal prediction regions are obtained using the software 
defaults.  Therefore, observed deviations from guaranteed coverage properties 
can result from the precision (or lack thereof) of the default settings.

\begin{figure}[h!]
\begin{center}
\includegraphics[width=0.450\textwidth]{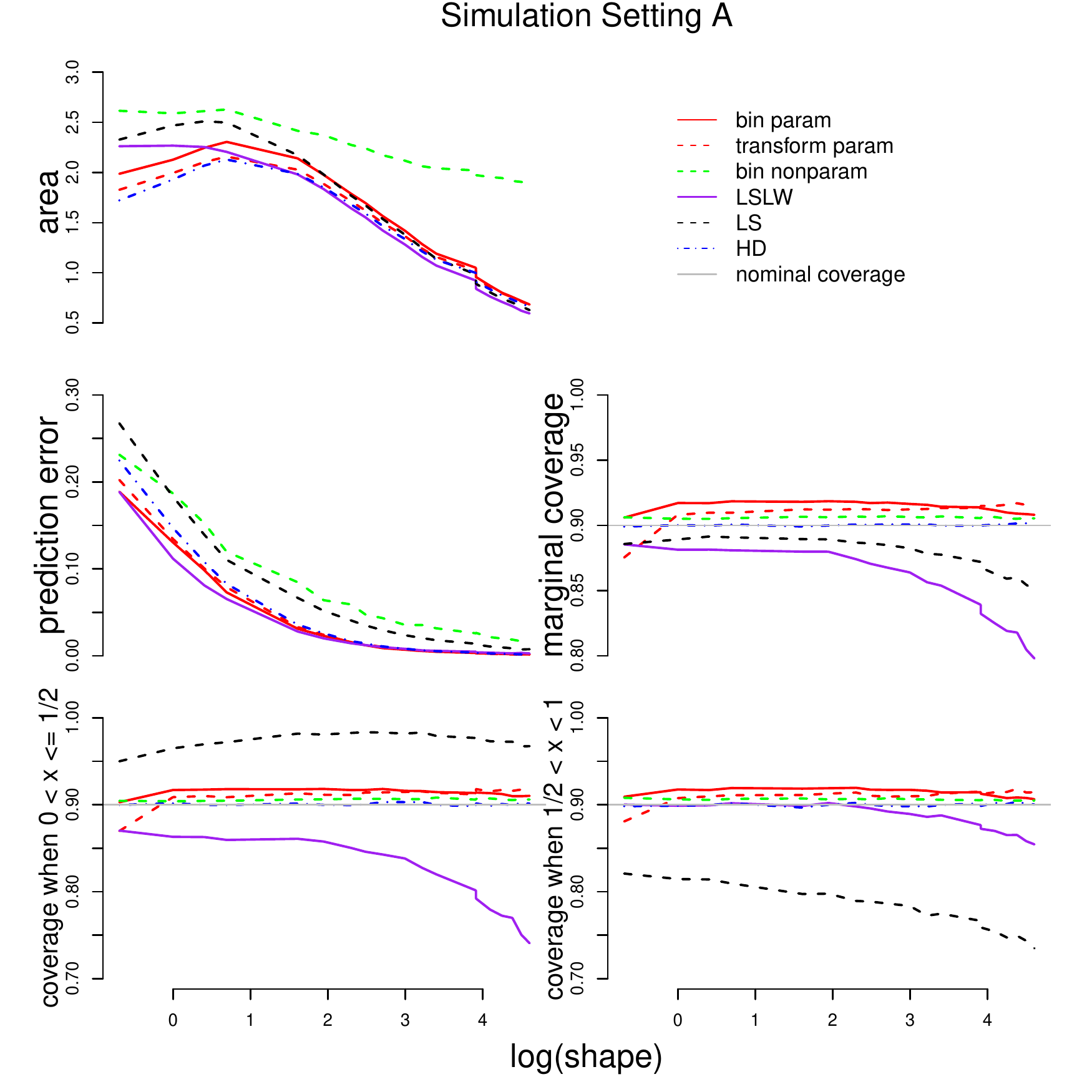}
\includegraphics[width=0.450\textwidth]{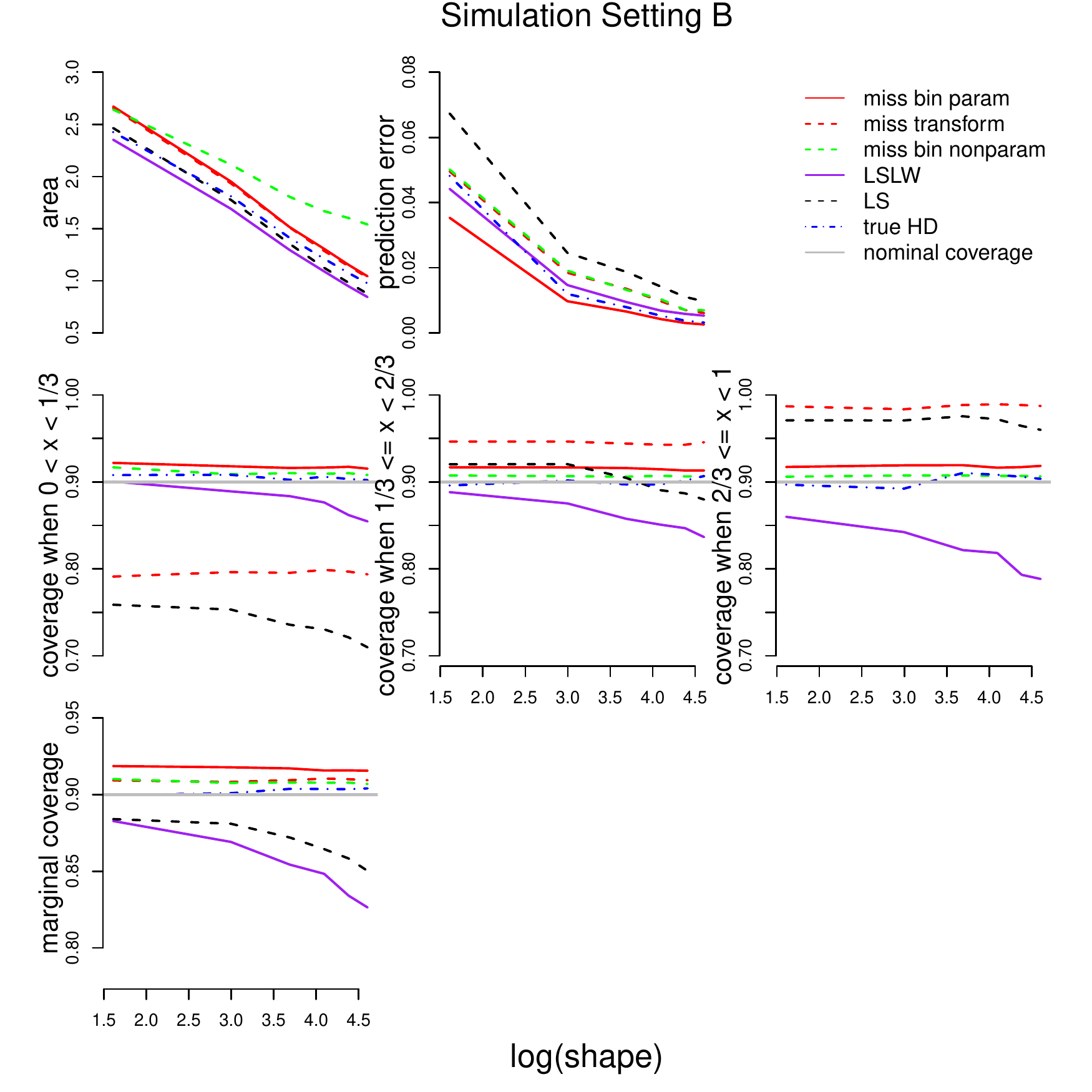}
\end{center}
\caption{Area, prediction error, and bin-wise coverage for 
    parametric, 
    nonparametric, 
    LS, 
    LSLW conformal prediction region, 
    and the highest density prediction region 
  for gamma GLM regression with $n=250$.   Simulation setting A is shown at 
  left, and setting B at right.
  The average of 50 Monte Carlo samples at each shape parameter value in 
  these simulation settings form the lines that are depicted in this figure.}
\label{Fig:diagnostics-250}
\end{figure}

\newpage
\begin{figure}[h!]
\begin{center}
\includegraphics[width=0.450\textwidth]{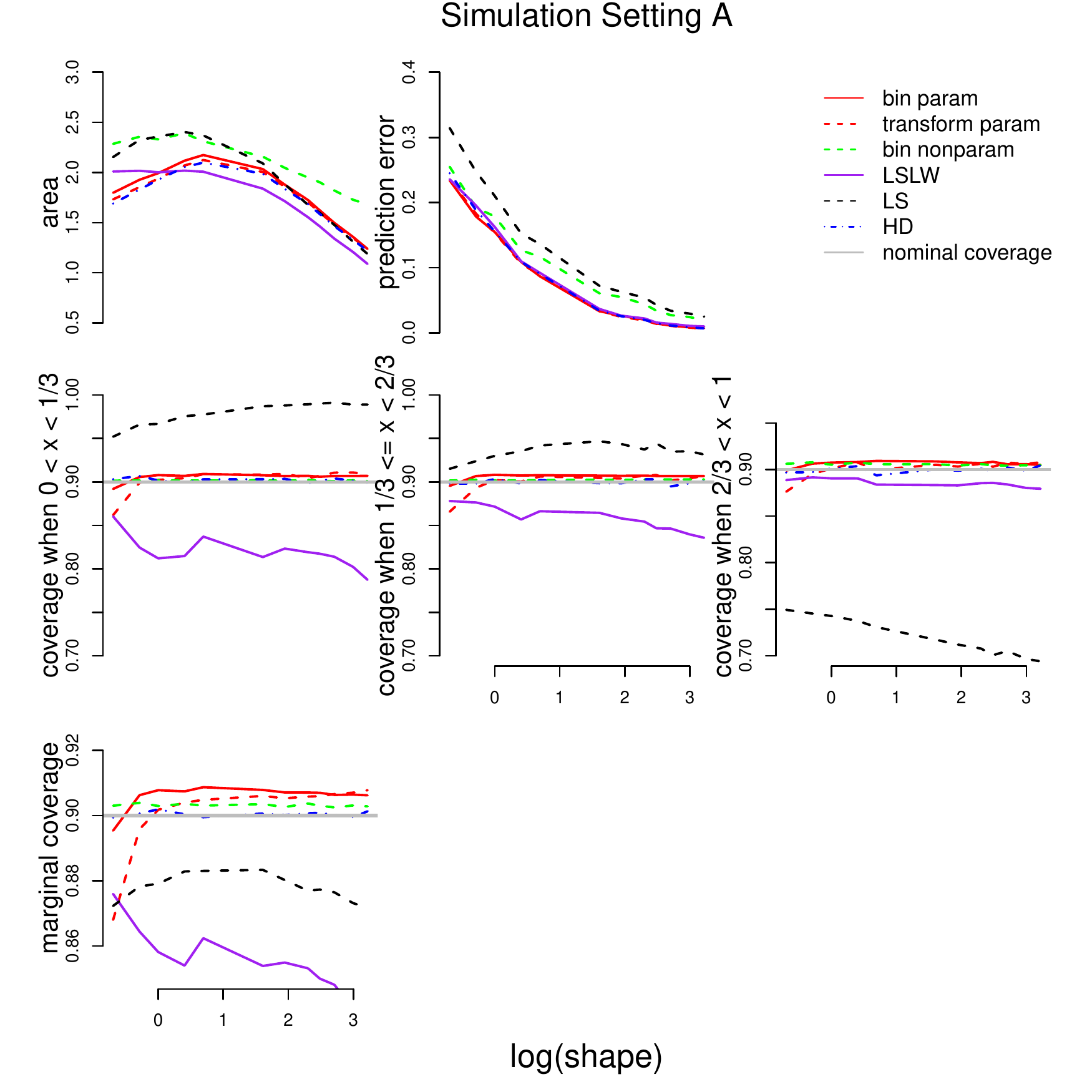}
\includegraphics[width=0.450\textwidth]{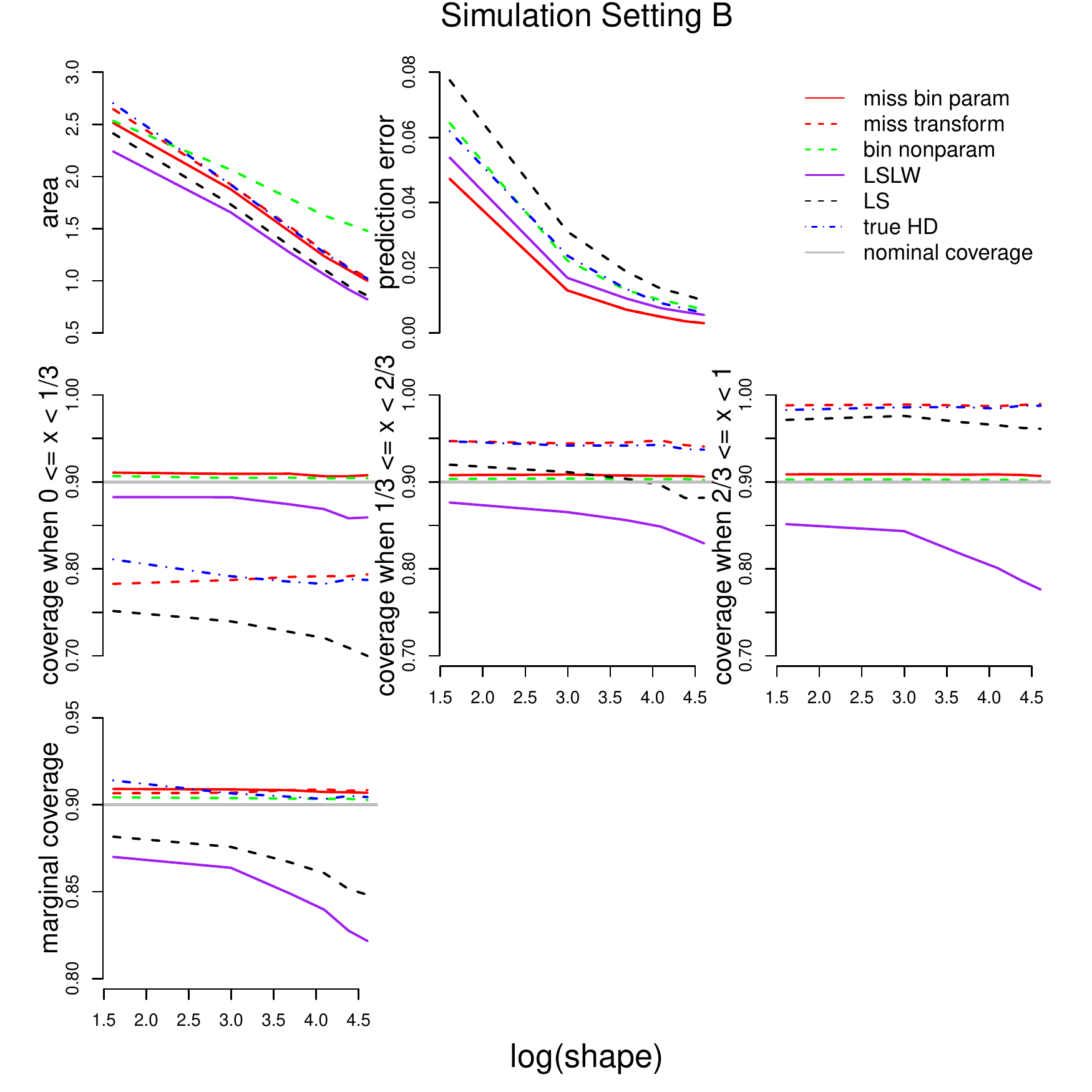}
\end{center}
\caption{Area, prediction error, and bin-wise coverage for 
    parametric, 
    nonparametric, 
    LS, 
    LSLW conformal prediction region, 
    and the highest density prediction region 
  for gamma GLM regression with $n=500$.   Simulation setting A is shown at 
  left, and setting B at right.
  The average of 50 Monte Carlo samples at each shape parameter value in 
  these simulation settings form the lines that are depicted in this figure.}
\label{Fig:diagnostics-500}
\end{figure}

\newpage
\begin{figure}[h!]
\begin{center}
\includegraphics[width=0.70\textwidth]{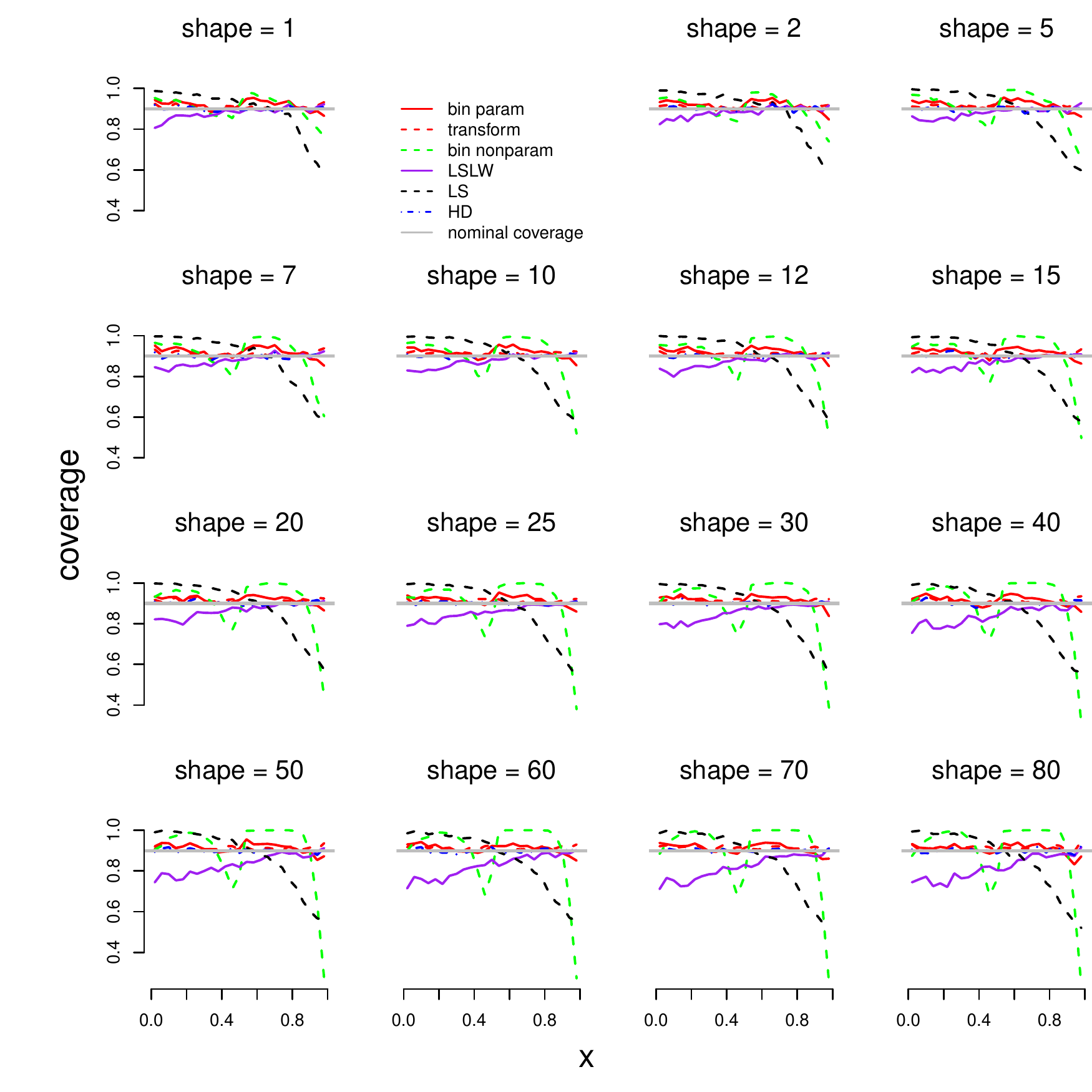}
\end{center}
\caption{Plot of the estimated coverage probabilities of prediction regions 
  across $x$ and shape parameter values when the model is correctly 
  specified, $n = 150$, and the number of bins is equal to $2$.
  The average of 250 Monte Carlo samples at each shape parameter value in 
  these simulation settings form the lines that are depicted in this figure.}
\label{Fig:gamma-inx-150}
\end{figure}

\newpage
\begin{figure}[h!]
\begin{center}
\includegraphics[width=0.70\textwidth]{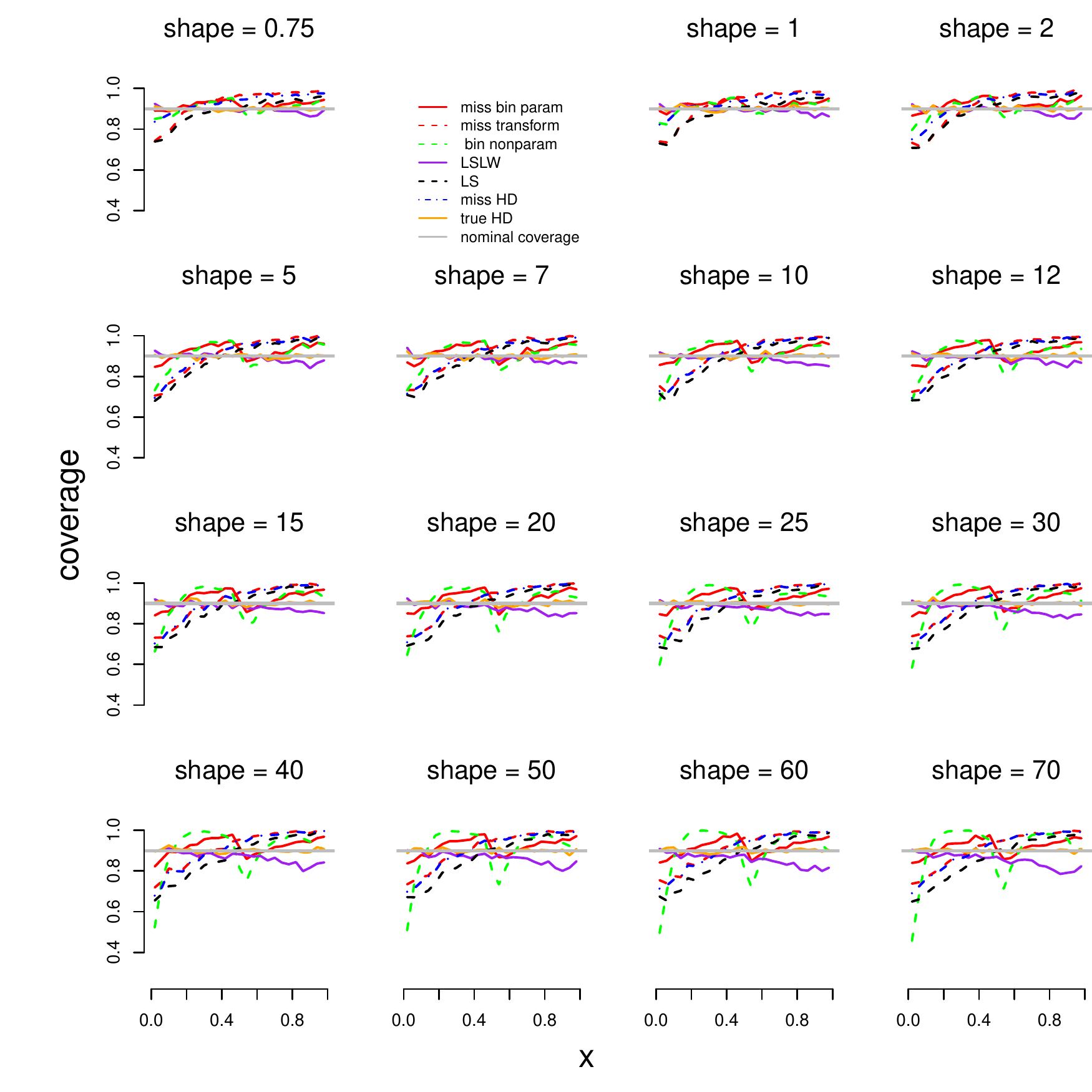}
\end{center}
\caption{Plot of the estimated coverage probabilities of prediction regions 
  across $x$ and shape parameter values when the model is misspecified, 
  $n = 150$, and the number of bins is equal to $2$.
  The average of 250 Monte Carlo samples at each shape parameter value in 
  these simulation settings form the lines that are depicted in this figure.}
\label{Fig:misspec-inx-150}
\end{figure}

\newpage
\section{Plots of conformal prediction regions}
\label{sec:plotsofregions}

\subsection{Plots when Gamma model is correctly specified}
\label{sec:gammaplots}

\begin{figure}[h!]
\begin{center}
\includegraphics{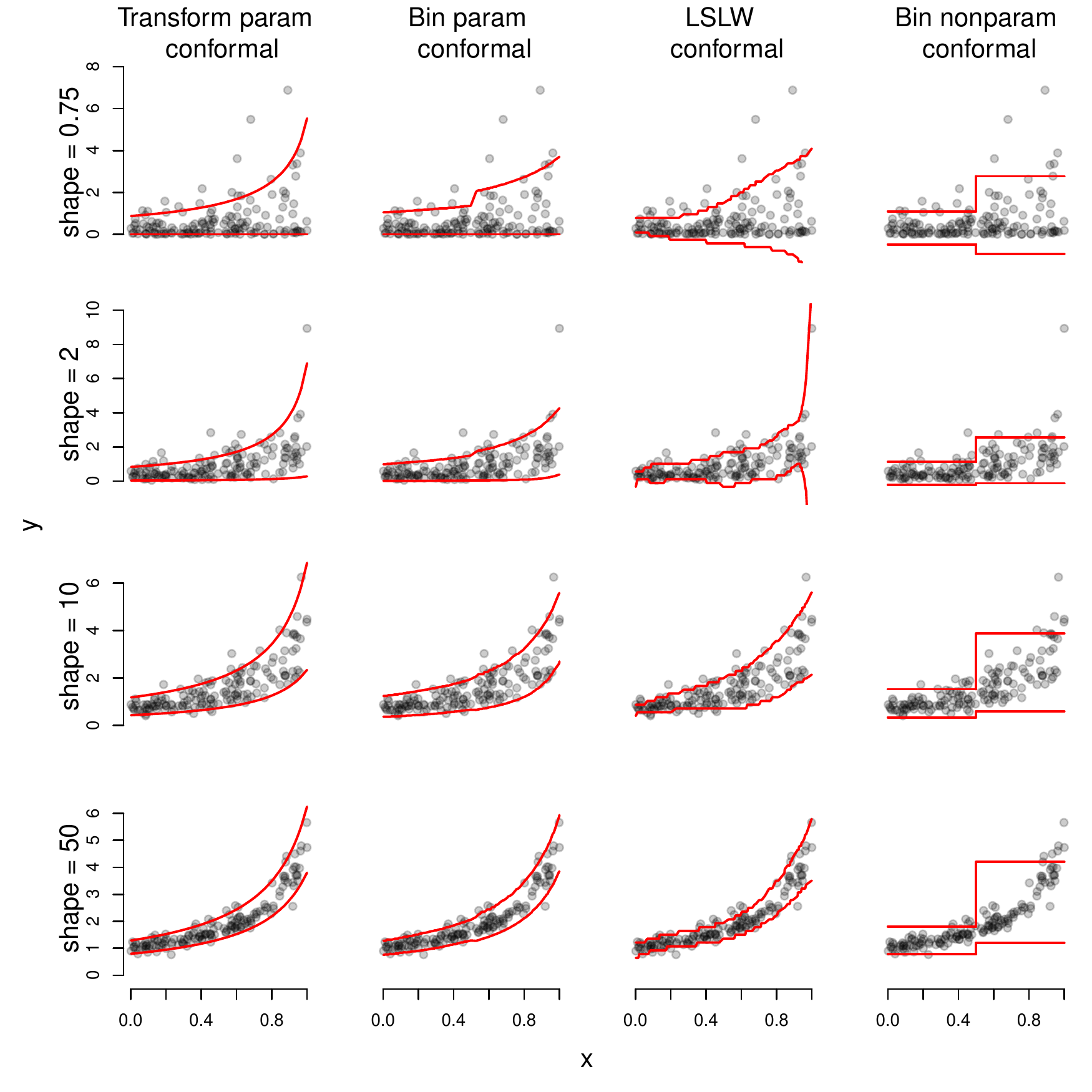}
\end{center}
\caption{Illustration of conformal prediction regions for Gamma 
  regression under simulation setting A with 
  $n=150$, $\alpha = 0.10$, and the number of bins equals 2.}
\label{conformal-plots-A-150}
\end{figure}

\newpage
\begin{figure}[h!]
\begin{center}
\includegraphics{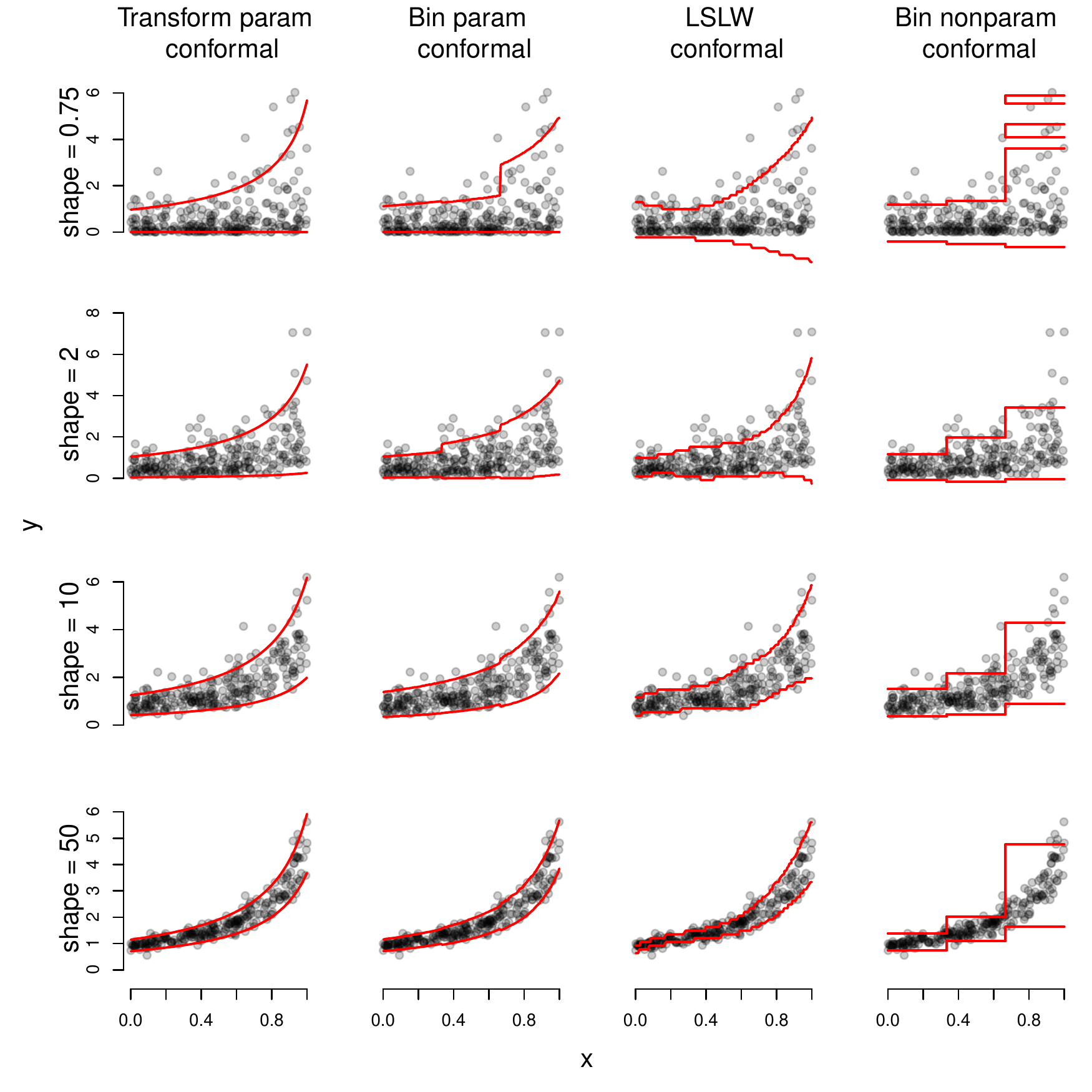}
\end{center}
\caption{Illustration of conformal prediction regions for Gamma 
  regression under simulation setting A with 
  $n=250$, $\alpha = 0.10$, and the number of bins equals 3.}
\label{conformal-plots-A-250}
\end{figure}

\newpage
\begin{figure}[h!]
\begin{center}
\includegraphics{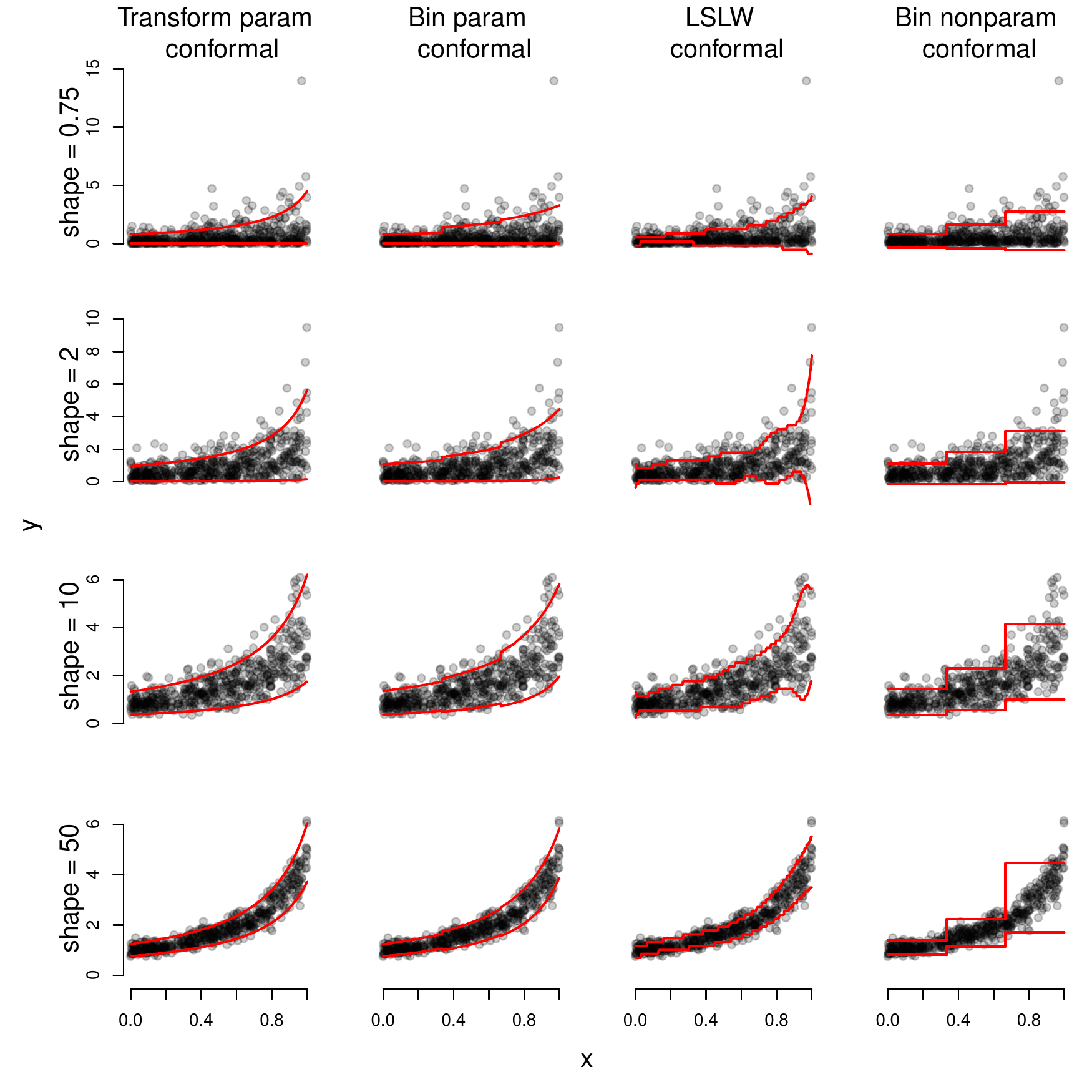}
\end{center}
\caption{Illustration of conformal prediction regions for Gamma 
  regression under simulation setting A with 
  $n=500$, $\alpha = 0.10$, and the number of bins equals 3.}
\label{conformal-plots-A-500}
\end{figure}

\newpage
\subsection{Plots when model is misspecified}

\begin{figure}[h!]
\begin{center}
\includegraphics{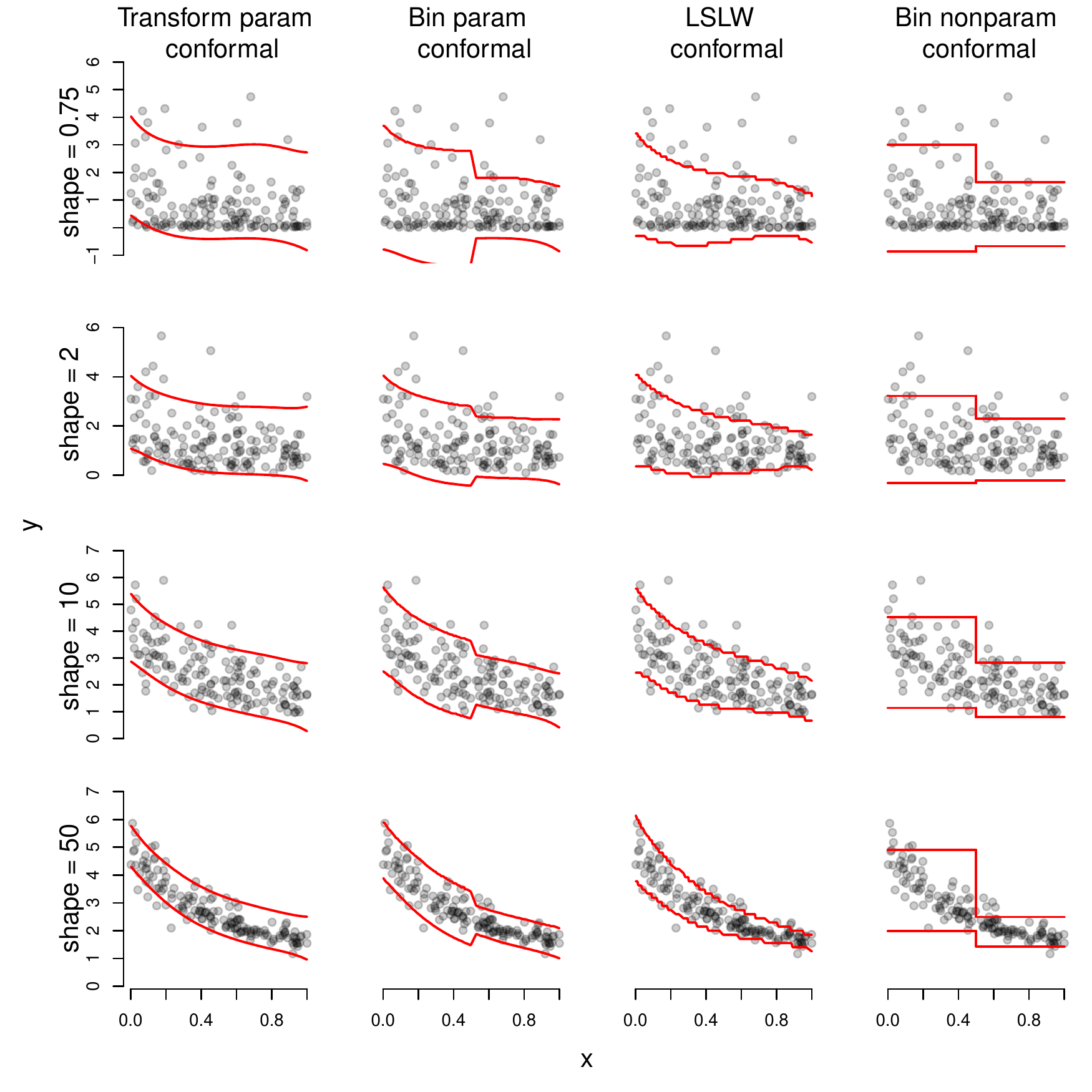}
\end{center}
\caption{Illustration of conformal prediction regions for Gamma 
  regression under simulation setting B with 
  $n=150$, $\alpha = 0.10$, and the number of bins equals 2.}
\label{conformal-plots-B-150}
\end{figure}

\newpage
\begin{figure}[h!]
\begin{center}
\includegraphics{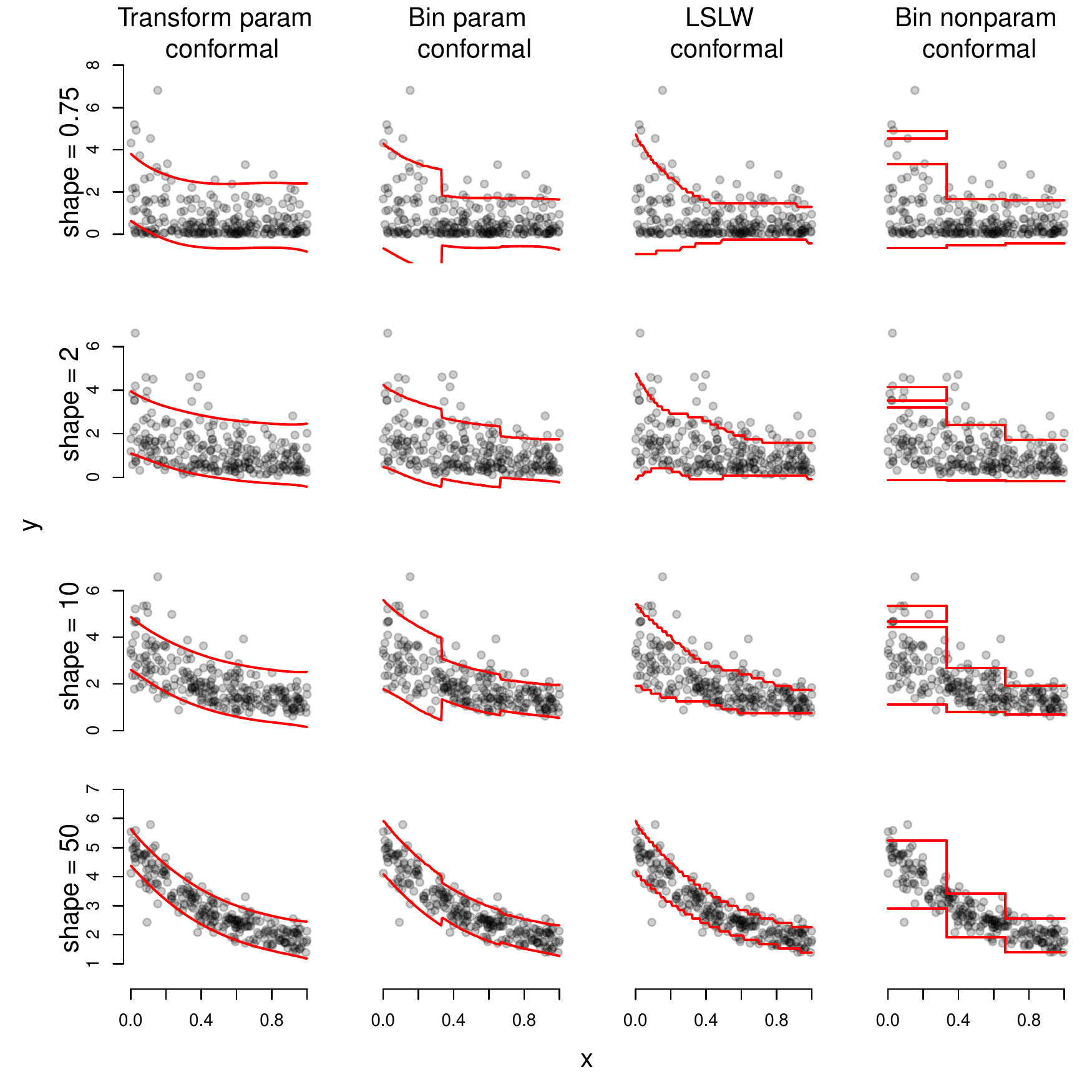}
\end{center}
\caption{Illustration of conformal prediction regions for Gamma 
  regression under simulation setting B with 
  $n=250$, $\alpha = 0.10$, and the number of bins equals 3.}
\label{conformal-plots-B-250}
\end{figure}

\newpage
\begin{figure}[h!]
\begin{center}
\includegraphics{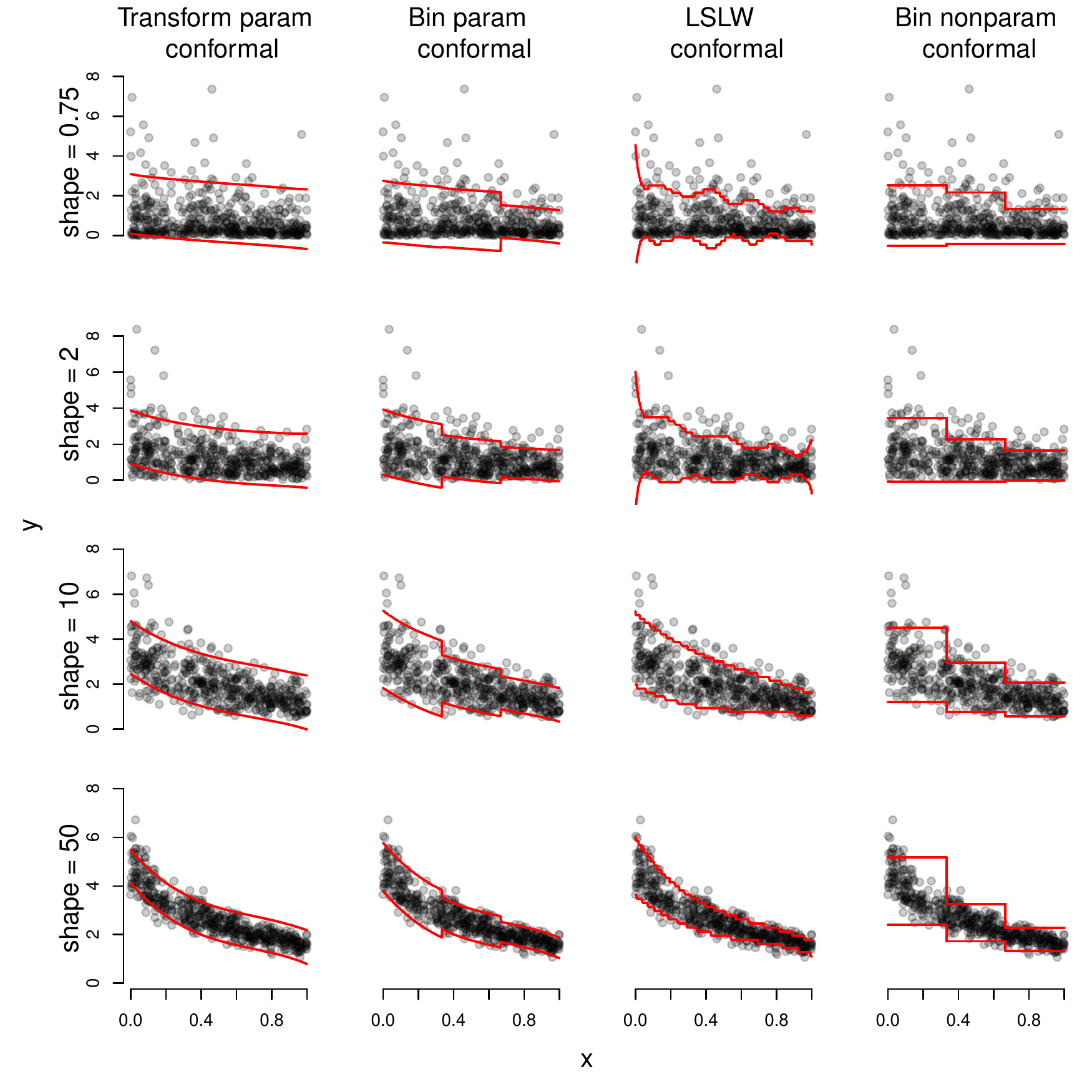}
\end{center}
\caption{Illustration of conformal prediction regions for Gamma 
  regression under simulation setting B with 
  $n=500$, $\alpha = 0.10$, and the number of bins equals 3.}
\label{conformal-plots-B-500}
\end{figure}

\newpage
\subsection{Plots when linear regression model is correctly specified}

\begin{figure}[h!]
\begin{center}
\includegraphics{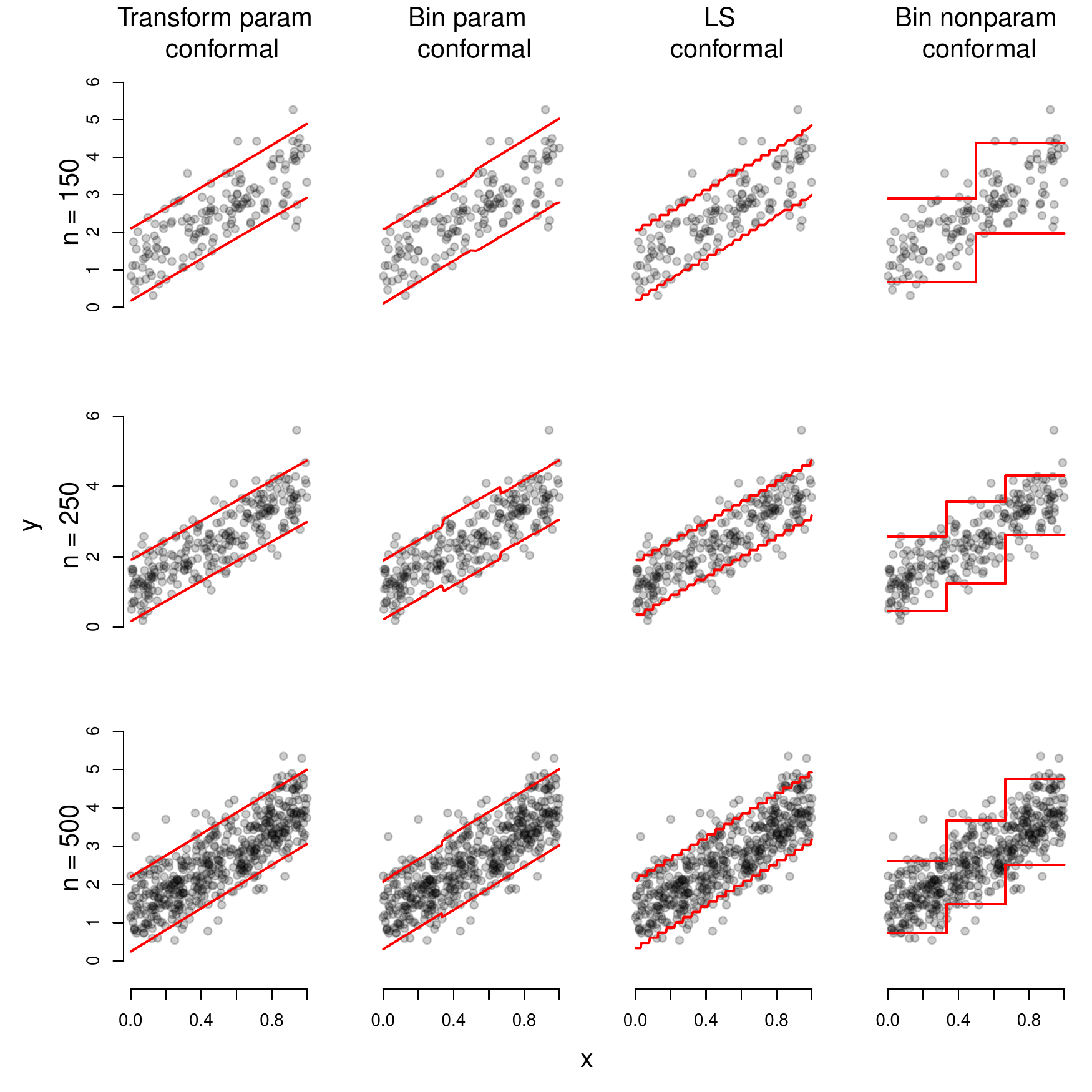}
\end{center}
\caption{Illustration of conformal prediction regions for linear  
  regression under simulation setting C with $\alpha = 0.10$.}
\label{conformal-plots-C}
\end{figure}

\newpage
\section{Example from README file in \citet{eck2018conformalR}}

In this section we illustrate conformal prediction regions through a gamma 
regression example with perfect model specification 
when the model is known, the model can be parameterized as in displayed 
equation 12 of the main text, and the model does not have additive 
symmetric errors.  We also compare conformal prediction regions to the oracle 
highest density region under the correct model.  This example is included in 
the README file of the corresponding R package \citet{eck2018conformalR}.  

In this example we set the sample size to $n = 500$ and we consider the 
gamma regression model with a single predictor with inverse link function 
(default in \texttt{glm} package).  The predictors were generated 
independently as $X_i \sim U(0,1)$, $i = 1$, $\ldots$, $n$.  We specified 
that $\beta = [0.25,\, 2]'$.  
The response variable was generated using \texttt{rgamma} with 
the shape parameter $\phi = 2$ and the rate as $2[1,\, X_i]'\beta$.  
Therefore we have iid data $(X_i,Y_i)$ from a continuous distribution where 
$Y_i|X_i$ is a GLM that can be parameterized as in displayed equation 12 of 
the main text.

\begin{figure}[ht!]
\begin{center}
\includegraphics[width=0.640\textwidth]{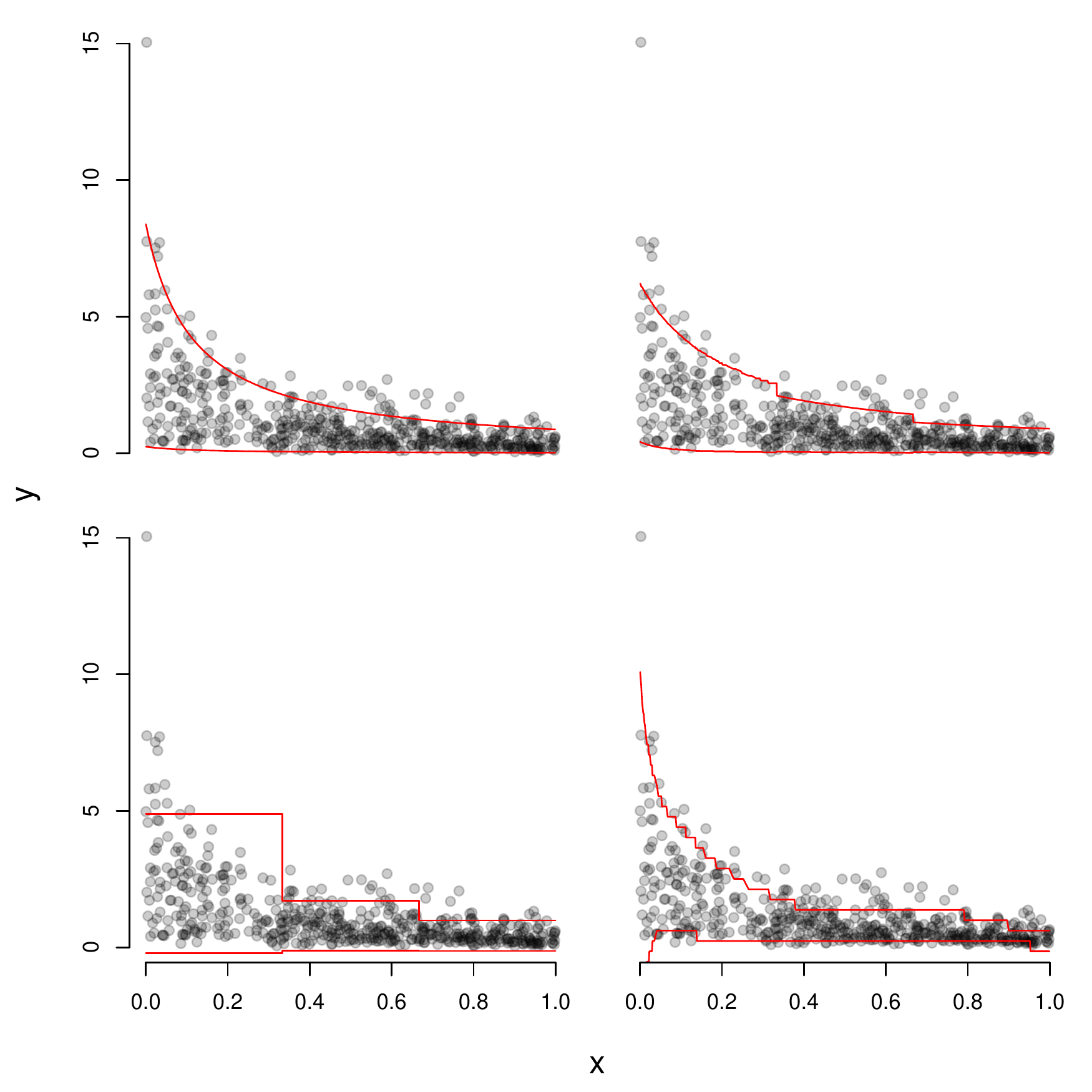}
\end{center}
\caption{
  Depiction of four prediction regions when $n = 500$.  
  The description of the gamma regression data generating process is outlined 
  in the README file of the corresponding R package 
  \citet{eck2018conformalR}.  Sim setting: shape = $2$, bins = $3$. 
  The regions are depicted as follows: 
    transformation based parametric conformal prediction (top-left panel),
    binned parametric conformal prediction region (top-right panel),
    binned nonparametric conformal prediction region (bottom-left panel), and 
    LSLW conformal prediction region (bottom-right panel).
}
\label{Fig:gammaexample}
\end{figure}

Figure~\ref{Fig:gammaexample} shows prediction regions for four of the five 
prediction regions considered.  
The top row depicts the transformation parametric conformal prediction 
region (left panel) and the binned parametric conformal prediction region 
(right panel).  The bottom row depicts the binned nonparametric 
conformal prediction region (left panel) and the LSLW conformal prediction 
region (right panel).  The bin width was specified as 1/3 for the binned 
conformal prediction regions.  We see that the binned parametric conformal 
prediction region is a close discretization of the transformation conformal 
prediction region, the nonparametric conformal prediction region is quite 
jagged and unnatural, and the LSLW conformal prediction region includes 
negative response values and is jagged. 
In Figure~\ref{Fig:gammaexample2} we see that the transformation based 
parametric conformal prediction region closely resembles the HD prediction 
region, and that the binned parametric conformal prediction region is a close 
descretization of the HD prediction region.

\begin{figure}[ht!]
\begin{center}
\includegraphics[width=0.640\textwidth]{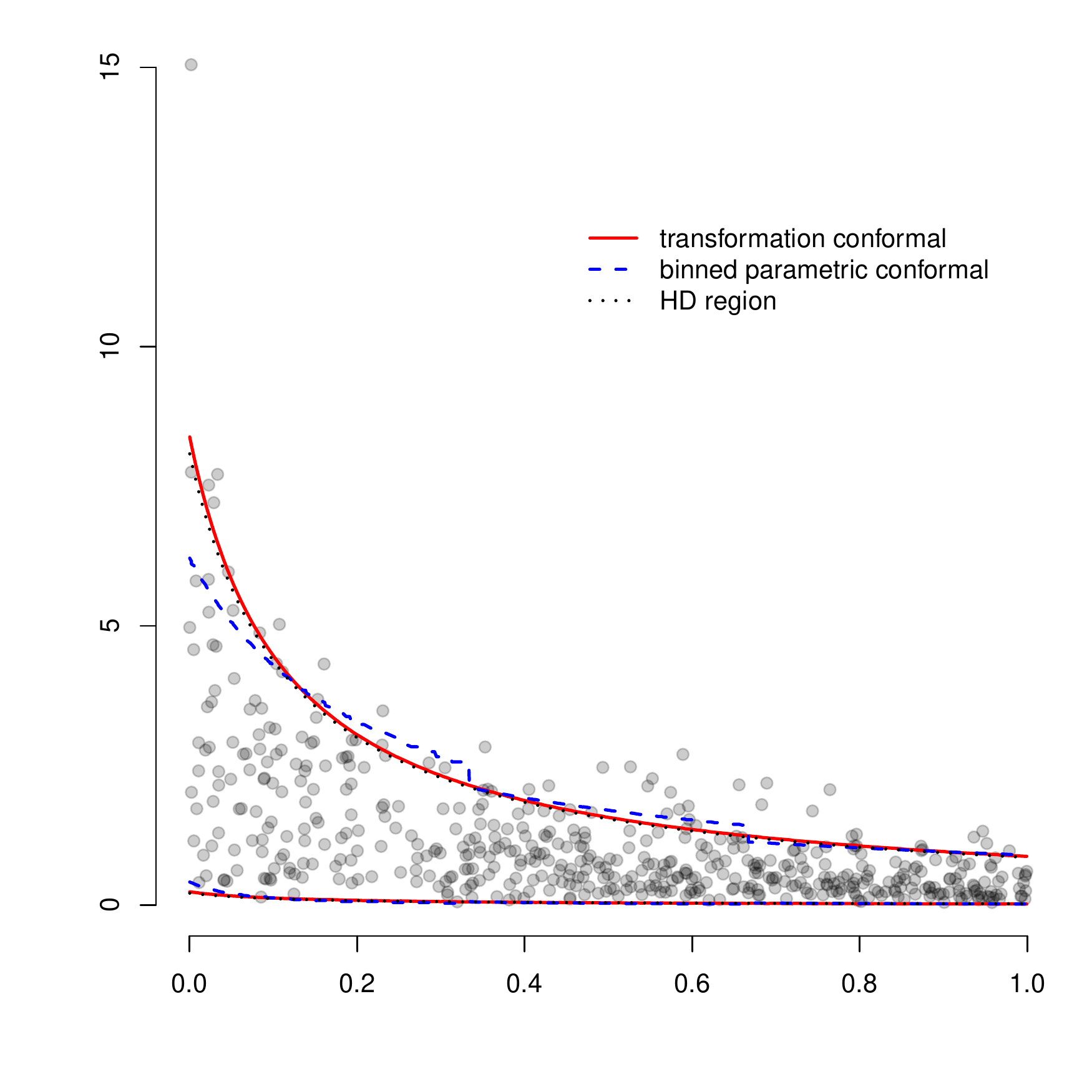}
\end{center}
\caption{
  Depiction of both parametric conformal prediction regions in 
  Figure~\ref{Fig:gammaexample} and the HD region.
}
\label{Fig:gammaexample2}
\end{figure}

All of the presented prediction regions exhibit close to finite-sample 
marginal validity and local validity with respect to binning.  However, 
the transformation conformal prediction region, LSLW conformal prediction 
region, and the HD prediction region do not exhibit finite-sample local 
validity in the second bin and the HD prediction region does not quite 
possess finite-sample marginal validity.  
The binned parametric conformal prediction region is smallest in size 
with an estimated area of 2.19.  The HD prediction region is a close second 
with an estimated area of 2.21.  The transformation conformal prediction 
region is a respectable third with an estimated area of 2.26.
LSLW conformal prediction region has an estimated area of 2.56 and 
The nonparametric conformal prediction region has an estimated area of 2.68. 
Under correct model specification, the parametric conformal prediction 
regions are similar in performance to that of the highest density prediction 
region. 

Performance of these prediction regions under the same model 
specification in this example is investigated further in a Monte Carlo 
simulation study.  The results of this Monte Carlo simulation study are 
consistent with those presented here.  Details are included in the 
reproducible technical report available at 
\url{https://github.com/DEck13/conformal.glm/tree/master/techreport}.

\bibliographystyle{plainnat}
\bibliography{conformalsources}

\end{document}